\documentclass[10pt]{article}

\usepackage{amsmath,amssymb,amsthm,color,epsfig,mathrsfs,amsfonts,bbm,dsfont}
\usepackage{amsbsy,empheq}
\usepackage{xfrac}
\usepackage{comment}
\usepackage{placeins}
\usepackage{cite}

\renewcommand{\theequation}{\arabic{section}{.}\arabic{equation}}

\newtheorem{theorem}{Theorem}
\newtheorem*{theorem*}{Theorem}

\newtheorem{corollary}[theorem]{Corollary}

\newtheorem{lemma}[theorem]{Lemma}

\newtheorem{proposition}[theorem]{Proposition}

\newcounter{spslist}
\newenvironment{spslist}{
  \begin{list}
  {\begin{picture}(1,1)
     \setlength{\unitlength}{0.5cm}
     \put(0,0.22){\circle*{0.2}}
    \end{picture}}
  {\usecounter{spslist}
  \setlength{\leftmargin}{1em}
  \setlength{\labelsep}{0.6em}
  \setlength{\labelwidth}{1em}
  \setlength{\topsep}{1ex}
  \setlength{\rightmargin}{0em}
  \setlength{\itemsep}{0.5ex}
  \setlength{\parsep}{0em}
  \setlength{\itemindent}{0em} }}
  {\end{list}}

\newlength\enumsep
\newcounter{enumlist}

\newcommand{\col}[3]{ \renewcommand{\arraystretch}{#1}
                \left[\!\! \begin{array}{c} #2 \\ #3 \end{array} \!\!\right] }

\newcommand{\mat}[5]{ \renewcommand{\arraystretch}{#1}
                    \left[\! \begin{array}{cc}
                            #2 & #3 \\
                            #4 & #5 \end{array} \!\right] }

\newcounter{geqncount}
    {\refstepcounter{equation}%
     \setcounter{geqncount}{\value{equation}}%
     \setcounter{equation}{0}%
  \renewcommand{\theequation}{\arabic{section}.\arabic{geqncount}.\alph{equation}}}%
    {\setcounter{equation}{\value{geqncount}}}

\newcommand{\eps}{\varepsilon}

\newcommand{\bedge}{\mathfrak{e}}
\newcommand{\dedge}{\mathfrak{d}}

\newcommand{\ZZ}{\mathbb{Z}}
\newcommand{\RR}{\mathbb{R}}

\newcommand{\cR}{\mathcal{R}}

\newcommand{\cU}{\mathcal{U}}
\newcommand{\cF}{\mathcal{F}}

\newcommand{\GGeps}{\mathbb{G}_{\eps}}

\newcommand{\Ge}{\Gamma_{\!\eps}}
\newcommand{\Gebi}{\dot{\Gamma}_{\!\eps}}
\newcommand{\Dom}{{\mathcal D}}
\newcommand{\ooo}{\mathrm{o}}
\newcommand{\eee}{\mathrm{e}}
\newcommand{\free}{\mathrm{f}}
\newcommand{\He}{{H_\eps}}
\newcommand{\Hebi}{{\dot{H}_\eps}}
\newcommand{\Hemin}{{H_\eps^\mathrm{min}}}
\newcommand{\Heper}{{H_\eps^\mathrm{per}}}
\newcommand{\Hepero}{{H_\eps^\mathrm{per,D}}}
\newcommand{\Hepere}{{H_\eps^\mathrm{per,N}}}
\newcommand{\Heperbi}{{\dot{H}_\eps^\mathrm{per}}}
\newcommand{\Heperbio}{{\dot{H}_\eps^\mathrm{per,o}}}
\newcommand{\Heperbie}{{\dot{H}_\eps^\mathrm{per,e}}}
\newcommand{\Hedef}{{H_\eps^\mathrm{def}}}
\newcommand{\Hedefo}{{H_\eps^\mathrm{def,D}}}
\newcommand{\Hedefe}{{H_\eps^\mathrm{def,N}}}
\newcommand{\Hedefbi}{{\dot{H}_\eps^\mathrm{def}}}
\newcommand{\Hedefbio}{{\dot{H}_\eps^\mathrm{def,o}}}
\newcommand{\Hedefbie}{{\dot{H}_\eps^\mathrm{def,e}}}

\DeclareMathOperator*{\dom}{\mathrm{dom}}
\DeclareMathOperator*{\supp}{\mathrm{supp}}

\DeclareMathOperator*{\ran}{\mathrm{ran}}
\DeclareMathOperator*{\sinc}{\mathrm{sinc}}

\allowdisplaybreaks

\begin{document}

\begin{center}
{\bf \Large Homogenized graphs\\\vspace{5pt}with spectrally embedded defect states}
\end{center}

\vspace{0.2ex}

\begin{center}
{\scshape \large Kirill D\hspace{-1pt}. Cherednichenko\, and\,  Stephen P\hspace{-2.5pt}. Shipman} \\
\vspace{1ex}
{\itshape University of Bath \,and\, Louisiana State University}
\end{center}

\vspace{3ex}
\centerline{\parbox{0.9\textwidth}{
{\bf Abstract.}\
By homogenizing high-contrast periodic metric graphs, we create continuous media that admit bound states produced by local defects at energies embedded in the continuous spectrum.
This is achieved by fitting two copies of a lattice graph simultaneously into space and coupling them by connecting edges.
The homogenization of an ODE on a graph to a PDE in space is achieved by representing the associated operators in a common functional space through interpolation by Fourier analysis.
The two coupled graphs result in an effective medium with internal degrees of freedom that allow a decomposition into two orthogonal spaces of hybrid states.
Critical high contrast retains microscopic resonances in the effective medium and allows a spectral band of one hybrid state space to overlap a spectral gap of the other. 
This serves as the mechanism for creating bound states at frequencies embedded in the continuous spectrum of the effective medium.
These bound states in the continuum persist for $\eps$-scale metric graphs due to $O(\eps)$ operator-norm estimates for reduced resolvents on the defect region.
}}

\vspace{3ex}
\noindent
\begin{mbox}
\end{mbox}
\vspace{3ex}

\hrule
\vspace{1.1ex}

\section{Introduction}

A bound state in the continuum (BIC) is an $L^2$-eigenfunction whose eigenvalue lies inside the continuum of propagating states, so that the surrounding medium supports propagating waves at the same energy or frequency. In this paper, we construct such states at local defects in a spatially homogeneous dispersive continuum and connect them quantitatively to embedded states of a microscopic resonant medium. The construction uses high-contrast homogenisation of two coupled periodic metric graphs. Their one-dimensional constituents form networks extending in three independent spatial directions, and their distinct layer amplitudes survive in the effective medium as internal degrees of freedom, which allow the effective system to split into two uncoupled wave channels: one supports a defect state localised in all three spatial directions, while the other supports propagating waves at the same frequency. The effective medium both elucidates this mechanism and allows us to determine the defect strength required to produce a localised state at a prescribed frequency; the resulting construction can then be transferred to sufficiently fine microscopic graphs. The principal conclusion is summarised as follows.

\begin{theorem*}[Informal statement of Theorem\,\ref{thm:bilayer} in \S\ref{sec:whatweprove}]
  A linear second-order high-contrast periodic differential operator on the metric graph in Fig.\,\ref{fig:SingleDouble3D} embedded in $\RR^d$ has a formal continuum limit, as the period tends to zero, given by a coupled dispersive system of linear partial differential operators in $\RR^d$ (see Eq.\,(\ref{limitsystem}) and Sec.\,\ref{sec:fulllimitsystem}).  When the PDE system is homogeneous except for a local perturbation, it may admit a defect state at a frequency embedded in the continuous spectrum, and when this occurs, the metric graph operator also admits a spectrally embedded defect state at the same frequency when the period is small enough.
\end{theorem*}

The passage from the effective construction to microscopic embedded states is justified by operator-norm estimates for the reduced defect problem. The defect strength on the graph is chosen as $\mu(\eps)$, converging to the effective value as the period $\eps$ tends to zero, so that the prescribed frequency is retained exactly. The estimates control this passage between scales, while the decomposition preserved by the defect ensures that the localised state remains decoupled from the propagating component. Both features are essential to this construction: even norm-resolvent convergence of self-adjoint operators does not, in general, ensure persistence of an embedded eigenvalue, since small perturbations coupling the localised state to propagating waves can destroy it \cite{AgmonHerbstSkibsted1989,LeeZworski2016}.

BICs have a long history, beginning with the von Neumann--Wigner construction \cite{vonNeumannWigner1929} for the Schr\"odinger equation and continuing through the scattering-theoretic analysis of Fonda \cite{Fonda1963}, the configuration-interaction and resonance picture of Fano \cite{Fano1961}, and later interference constructions of Stillinger--Herrick
\cite{StillingerHerrick1975} and Friedrich--Wintgen \cite{FriedrichWintgen1985}. In wave physics the same phenomenon appears under the names trapped mode, embedded guided mode, or non-radiating state, in acoustics, water waves, elasticity and electromagnetism; see, for example, \cite{EvansLintonUrsell1993, EvansLevitinVassiliev1994,HsuEtAl2016} and the references therein.

The connection with scattering is fundamental. Heuristically, a bound state embedded in a propagating continuum is necessarily decoupled from all open radiation channels. When the mechanism responsible for this decoupling is perturbed, the embedded eigenvalue may disappear into the continuum and give rise to a nearby resonance pole, producing a sharp Fano-type scattering anomaly \cite{Fano1961,Fonda1963,ReedSimonIV,FriedrichWintgen1985,AslanyanParnovskiVassiliev2000, ShipmanVenakides2005,Shipman2010,ShipmanWelters2013,ChesnelNazarov2018a,ChesnelNazarov2020a}. This instability--protection dichotomy is one reason BICs have become important in photonics and metamaterials: under a small symmetry-breaking perturbation, the real embedded eigenvalue is typically replaced by a nearby resonance pole of small imaginary part \cite{HsuEtAl2016,KangEtAl2023}.

There are constructive precedents for defect BICs in continuum wave models. Bulgakov and Sadreev \cite{BulgakovSadreev2008} developed a procedure for finding defect-localised BICs in photonic-crystal waveguides by tuning the defect parameters, while Vaidya et al.\ \cite{VaidyaEtAl2021} constructed point-defect BICs in two-dimensional photonic crystals through frequency matching and symmetry mismatch. These constructions are supported by numerical calculations.
An explicit derivation in three dimensions is given by Silveirinha \cite{Silveirinha2014}: a lossless core--shell inclusion with vanishing shell permittivity at a selected frequency supports a field that is identically zero outside the inclusion. Related three-dimensional embedded photonic states are studied in \cite{MonticoneAlu2014}. BICs in double-network metamaterials have also been investigated using effective-medium descriptions \cite{WangEtAl2023}. Our objective is to construct a fully localised defect state in a homogeneous effective bulk medium and to establish quantitative control connecting that state to exact embedded states of an increasingly fine periodic graph.

For a homogeneous scalar elliptic medium in the whole space, with sufficiently regular, nondegenerate coefficients equal to constant ambient values outside a bounded defect, such localisation in a positive-energy propagation regime is strongly obstructed by exterior Rellich-type uniqueness arguments and unique continuation. More generally, for periodic Schr\"odinger operators perturbed by a sufficiently localised impurity, embedded eigenvalues are absent under suitable irreducibility and decay assumptions \cite{KuchmentVainberg2000}. Periodic graph and quantum-graph operators are different: local perturbations can create embedded eigenvalues, and when the Fermi surface is irreducible the corresponding eigenfunctions are forced to be compactly supported under the hypotheses of \cite{KuchmentVainberg2006}. This contrast makes periodic graph structures a particularly transparent setting in which to engineer and analyse mechanisms that suppress radiation channels.

The mechanism used here is based on a symmetry-induced decomposition into reducing subspaces. In \cite{Shipman2014}, coupled periodic graph systems were constructed with a symmetry that decomposes the operator into orthogonal invariant components associated with hybrid states. The corresponding Floquet fibres decompose accordingly, and hence the Fermi surface is reducible. At an energy for which one component supports propagating waves while the other supports only evanescent waves, a local perturbation preserving this decomposition can produce an exponentially decaying eigenfunction in the evanescent component, without coupling it to the propagating waves in the other component. The present work starts from this mechanism but places it inside a singular high-contrast homogenisation problem. The objective is not only to produce an embedded eigenstate on a periodic graph, but to construct a family of increasingly fine metric graphs whose embedded defect eigenstates are related, with quantitative operator control, to a bound state of a homogeneous effective dispersive medium.

More precisely, we use a periodic metric graph consisting of two geometrically separated but identical lattice layers connected by soft edges. The graph has an order-two involution that interchanges the layers and commutes with lattice translations. The defect perturbation is imposed identically on the Robin vertex parameters in both layers and therefore preserves this involution. Consequently, the periodic and defective operators split orthogonally into even and odd parity sectors. Cutting each connecting edge at its midpoint identifies these sectors with two monolayer problems carrying, respectively, Neumann and Dirichlet terminal conditions on a soft decoration. Their dispersion laws are therefore different. We choose a frequency $k$ for which $k^2$ lies in a spectral gap of one parity sector and in the interior of a spectral band of the other. A localised state created in the gapped sector is then an eigenstate of the full bilayer operator whose eigenvalue is embedded in the continuous spectrum contributed by the other sector---this is the symmetry protection in the construction. As the period tends to zero, the graph system gives rise to a homogeneous dispersive continuum medium. The resulting effective equations retain internal resonant degrees of freedom inherited from the metric-graph microstructure.

We confine attention to the layer-symmetric configuration in order to expose the essential steps of the analytical method with minimal algebraic complexity. The method also applies to non-symmetric configurations for which the reduced equations admit a frequency-dependent channel decomposition preserved by the defect. The distinction between geometric symmetry and reducibility
is illustrated by the non-symmetric bilayer quantum graphs studied in \cite{Shipman2020NonSymmetric}. For example, one may retain identical edge structures in the two layers but assign different constant Robin parameters to their vertices, while imposing the same defect perturbation on both layers. After elimination of the edge fields, the resulting discrete
system consists of a common scalar lattice operator and a spatially constant real symmetric $2\times2$ matrix acting on the layer amplitudes, together with a defect proportional to the identity. At each fixed frequency away from the edge resonances, an orthogonal change of layer variables diagonalises this matrix; the defect remains diagonal. The resulting channels replace the even and odd sectors of the symmetric configuration. Provided that the target frequency remains in a gap of one channel and in a band of the other, with the uniform
bounds needed for the resolvent comparison, the same coupling-constant reformulation and Fourier interpolation give the corresponding homogenisation and embedded-state
construction. Thus layer-exchange symmetry provides a particularly transparent realisation of the mechanism but is not essential for the analysis.

The second ingredient is critical high contrast. The edge coefficients are scaled critically so that the resonant effect of the soft components survives in the limit $\varepsilon\to0$, producing frequency-dependent effective coefficients and spectral gaps rather than a classical static homogenized coefficient. Specifically, the edge lengths are of order $\varepsilon$, while the coefficients are of order one on the stiff edges and of order $\varepsilon^2$ on the soft edges. The soft-edge resonances therefore remain at frequencies
of order one as the period tends to zero. High-contrast homogenisation of this kind goes back to Zhikov's double-porosity and spectral-gap theory \cite{Zhikov2000,Zhikov2005} and has subsequently been developed in both two-scale and operator-theoretic forms; see, among others, \cite{Cherdantsev2009,KamotskiSmyshlyaev2018, CherednichenkoCooper2016,CherednichenkoKiselev2017, CherednichenkoErshovaKiselevNaboko2019, CherednichenkoErshovaKiselev2019}.
A characteristic feature of the critical regime is a nonlinear frequency dependence of the macroscopic coefficients (often encoded by a Zhikov-type function), equivalently a time-dispersive effective law. In our graph model this frequency dependence can be calculated explicitly from Dirichlet-to-Neumann maps of the soft edges.

There is an extensive literature on defect modes generated in spectral gaps and on their homogenisation. Cherdantsev established spectral convergence for localised defect modes in high-contrast periodic elliptic media \cite{Cherdantsev2009}; Kamotski and Smyshlyaev derived localised modes by high-contrast two-scale homogenisation and obtained quantitative eigenvalue estimates through asymptotic expansions together with a detailed boundary-layer analysis near the defect \cite{KamotskiSmyshlyaev2018}; and Hoefer and Weinstein analysed defect modes bifurcating from band edges through an effective homogenised Schr\"odinger problem \cite{HoeferWeinstein2011}. The approach taken here is different in that, after reduction by Dirichlet-to-Neumann maps to a discrete problem, the defect operators are compared directly in operator norm on a common Hilbert space. This avoids the construction of defect boundary layers and allows the relevant isolated eigenvalues and spectral subspaces to be transferred by spectral perturbation theory. The cited works concern localisation in gaps of the ambient periodic spectrum. In the present construction, by contrast, a gap is used only in one invariant component, while the complementary component provides the propagating states at the same spectral value. The result is therefore a defect eigenstate whose eigenvalue is embedded in the spectrum of the full system, rather than a conventional defect mode associated with a gap eigenvalue.

A third ingredient is a coupling-constant, or Birman--Schwinger-type, reformulation of the defect problem. After eliminating the fields in the edge interiors by means of Dirichlet-to-Neumann maps, the metric-graph problem at a fixed real frequency $k$, away from the exceptional frequencies of the reduction, reduces to a problem for the vertex values on the lattice $\varepsilon\mathbb Z^d$. The resulting discrete Schr\"odinger operator is frequency-dependent, with the dependence on $k$ inherited from the Dirichlet-to-Neumann maps. If $\cR(k,\varepsilon)$ denotes the resolvent of the undefected discrete operator in the parity sector that is gapped at $k^2$, and $\Psi$ is the prescribed smooth, compactly supported
defect profile, the bound-state equation can be rewritten, on the defect sites, as a finite-dimensional eigenvalue problem of the form $\Psi(\varepsilon\,\cdot)\cR(k,\varepsilon)v=\lambda v$, with $\lambda$ algebraically related to the defect strength $\mu$ by \eqref{lambda}. Thus the frequency is prescribed first, and the nonzero eigenvalues of a compact operator supported on the defect region select the values of the physical coupling parameter for which $k^2$ becomes an eigenvalue. This is the familiar Birman--Schwinger philosophy viewed as a coupling-constant spectrum \cite{ReedSimonIV,KlausSimon1980}; in the present setting it becomes a practical design step linking a prescribed frequency $k$, equivalently the target spectral value $k^2$, to a defect amplitude.

The main analytical difficulty is then not merely to homogenise the ambient operator, but to compare these coupling-constant eigenproblems as $\varepsilon\to0$. Their natural Hilbert spaces change with $\varepsilon$: the discrete problem acts in $\ell^2(\varepsilon\mathbb Z^d)$, whereas the limiting problem acts in $L^2(\mathbb R^d)$. We resolve this by an explicit Bloch--Gelfand identification. The unitary discrete Bloch--Gelfand transform $U_\varepsilon$ sends $\ell^2(\varepsilon\mathbb Z^d)$ to $L^2(\varepsilon^{-1}[-\pi,\pi]^d)$; we extend the transformed function by zero to all of $\mathbb R^d$ and then apply the inverse Fourier transform. In other words, $J_\varepsilon=\mathcal F^{-1}\iota U_\varepsilon$ embeds the lattice space isometrically as a band-limited subspace of $L^2(\mathbb R^d)$. Discrete-to-continuum resolvent estimates using identification operators have been developed in \cite{NakamuraTadano2021,CorneanGardeJensen2021,SCHMIDT2025129247}. Our approach makes the comparison directly in Fourier variables: the isometric identification places the discrete and continuum defect operators in a common Hilbert space, the lattice Laplacian is represented by its explicit Fourier multiplier, and the expanding Brillouin zones are accommodated by extension by zero. These features allow us to compare the defect operators directly in operator norm.

This common-space representation gives a direct comparison of symbols at each fixed admissible frequency. For frequencies $k$ lying in a spectral gap of the component in which the defect state is constructed, the extended discrete ambient resolvent converges to the continuum resolvent in operator norm with order $O(\varepsilon)$. More importantly for the defect construction, using the Poisson summation formula and representing multiplication by $\Psi$ as convolution in Fourier variables, we prove that the compact operators
governing the coupling constants also converge in $L^2(\mathbb R^d)$-operator norm at rate $O(\varepsilon)$. Norm convergence is the decisive level of control here: it transports isolated nonzero eigenvalues and their Riesz spectral subspaces, and hence the defect strengths and reduced states selected by the effective problem. The corresponding metric-graph states are recovered from the vertex values by the edge reconstruction --- this is stronger than a formal continuum limit and is tailored to the determination of defect strengths that produce localised states at a prescribed frequency.The broader operator-norm homogenisation programme for periodic and high-contrast media is represented by \cite{BirmanSuslina2004,CherednichenkoCooper2016, CherednichenkoKiselev2017, CherednichenkoErshovaKiselevNaboko2019}; the discrete-to-continuum part is related to \cite{NakamuraTadano2021,CorneanGardeJensen2021}.

The effective equations are frequency-dispersive because the resonant soft-edge variables have been eliminated. Their nonlinear dependence on the spectral parameter is compatible with an underlying linear self-adjoint spectral problem. Retaining the soft degrees of freedom produces an enlarged Hilbert space and a frequency-independent self-adjoint
effective operator; the dispersive macroscopic equation is its Schur complement. This operator-theoretic interpretation is consistent with the general link between critical-contrast homogenisation and time-dispersive media \cite{CherednichenkoErshovaKiselevNaboko2019, CherednichenkoErshovaKiselev2019}. We include the augmented formulation to make explicit the internal degrees of freedom underlying the frequency-dependent effective equations and to provide a self-adjoint spectral problem that remains well defined at the poles of the scalar Dirichlet-to-Neumann functions.

The paper therefore combines three constructions that are usually treated separately. First, symmetry produces orthogonal radiation and localisation channels and hence protects an embedded state. Second, high-contrast homogenisation turns the graph family into a homogeneous but dispersive continuum system while preserving the relevant band--gap separation. Third, a coupling-constant eigenproblem, compared across scales by an explicit Fourier--Bloch interpolation, selects the defect strengths and yields quantitative convergence. The outcome is a constructive connection between microscopic and macroscopic BICs: an effective continuum eigenproblem can be solved at a prescribed frequency, and the resulting
defect parameter determines, through a convergent scale-dependent adjustment, embedded eigenstates of sufficiently fine microscopic graphs. To our knowledge, this is the first construction of fully localised defect BICs in a spatially homogeneous dispersive continuum obtained by high-contrast homogenisation of periodic metric graphs, with quantitative operator-norm control of the reduced defect eigenproblem.

The structure of the paper is as follows. 
Sections~\ref{sec:single}--\ref{sec:bilayer} define the monolayer and bilayer metric graphs and derive the parity decomposition. Section~\ref{sec:reduction} eliminates the edge fields through Dirichlet-to-Neumann maps and reduces the fixed-frequency problem to a discrete Schr\"odinger equation and then to the coupling-constant eigenproblem. Sections~\ref{sec:results}--\ref{sec:convergence} identify
the continuum limit and prove the common-space operator-norm convergence, including the $O(\varepsilon)$ estimate for the defect operator. Section~\ref{sec:BIC} chooses frequencies for which a band of one parity overlaps a gap of the other and constructs embedded defect states, with one- and two-dimensional numerical illustrations. Appendix~\ref{sec:calculations} contains the Dirichlet-to-Neumann calculations, and Appendix~\ref{sec:explicit-soft-limit} gives the augmented self-adjoint macroscopic formulation whose Schur complement is the dispersive effective equation.

The following is an outline of our strategy.
\begin{spslist}

\item
On a graph of period size $\eps$ embedded in $\RR^d$, we define an operator $\Heperbi$ representing the ambient bilayer medium and a perturbed operator $\Hedefbi$ representing the defective medium. The defect is defined by a function $\Psi(x)$ of compact support and a scalar strength $\mu$, applied identically to the Robin parameters of corresponding vertices in the two layers. These vertex conditions are interpreted within the framework of self-adjoint extensions \cite{AlbeverioKurasov1999a}. We are interested in square-integrable solutions of $(\Hedefbi-k^2)u=0$ with $k^2$ embedded in the continuous spectrum.

\item
The operators $\Heperbi$ and $\Hedefbi$ are orthogonally decomposed into their even and odd components. Each component is unitarily equivalent to a monolayer graph operator, with a Neumann or Dirichlet terminal condition on the soft decoration.

\item
At a fixed frequency away from the exceptional values of the edge reduction, the monolayer problems are reduced to equivalent discrete problems on $\eps\ZZ^d$ of the form 
$(\Delta_\eps+V_\free(k,\eps)+\lambda^{-1}\Psi(\eps\,\cdot))\bar u=0.$ The parameter $\lambda$ is related to the physical defect strength $\mu$ by \eqref{lambda}.

\item
In the component that is gapped at $k^2$, we convert the discrete equation into an eigenproblem supported on the defect sites: $\Psi(\eps\,\cdot)(-\Delta_\eps-V_\free(k,\eps))^{-1}\bar v=\lambda\bar v,$ The vertex amplitudes and then the edge fields are recovered from $\bar v$.

\item
The corresponding continuum problem is $\Psi(-\Delta-V_\free(k,0))^{-1}v=\lambda v,$ Its nonzero eigenvalues select the effective defect strengths for which the prescribed frequency supports a localised state.

\item
Using Fourier interpolation, we place the discrete and continuum problems in a common Hilbert space and prove an $O(\eps)$ operator-norm estimate for the compact defect operators. This yields convergence of their isolated nonzero eigenvalues and associated spectral subspaces, and
hence microscopic defect strengths $\mu(\eps)$
converging to the effective value.

\item
We choose the target frequency so that the complementary component supports propagation. The reconstructed localised state is then a BIC of the full bilayer operator. The exact decomposition prevents coupling to the propagating component, while the operator estimates justify the transfer of the effective design to sufficiently small periods.

\end{spslist}

\begin{figure}[ht] 
\centerline{
\raisebox{20pt}{\scalebox{0.42}{\includegraphics{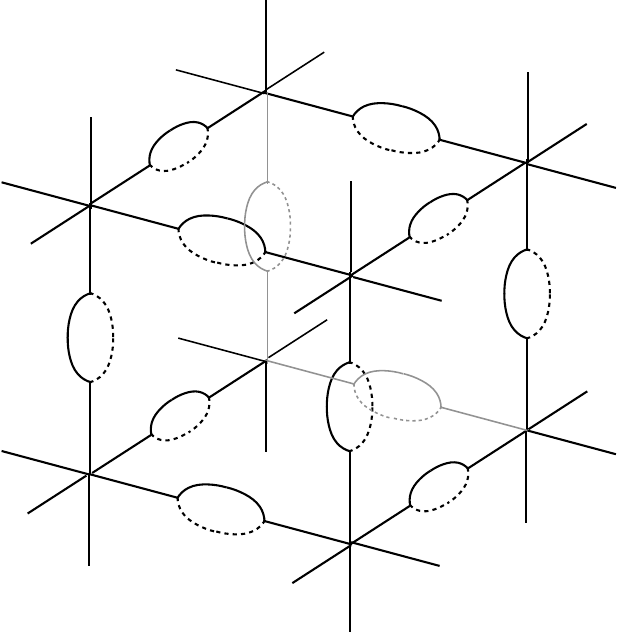}}}
\hspace{0.4em}
\scalebox{0.42}{\includegraphics{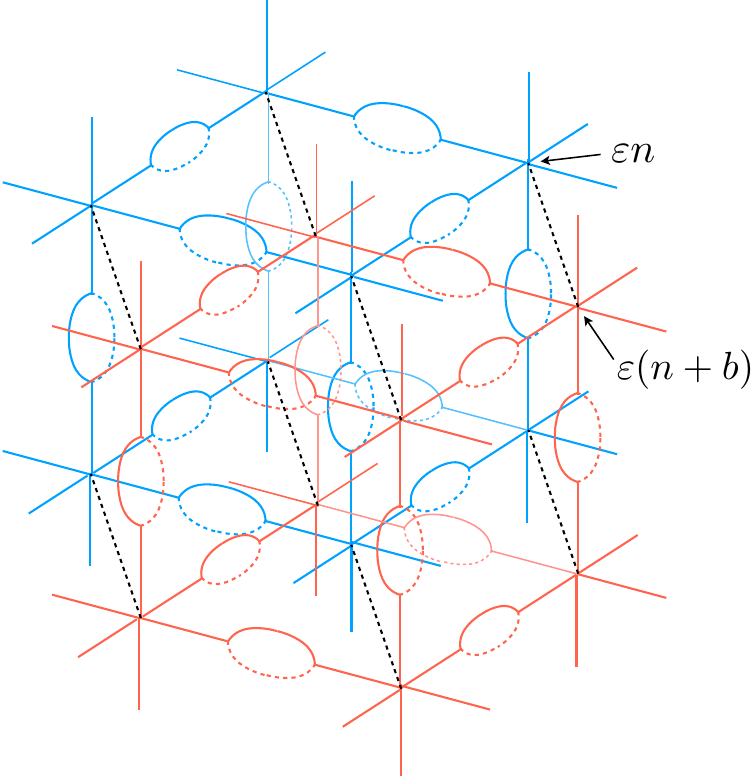}}
\hspace{-0.2em}
\raisebox{22pt}{\scalebox{0.42}{\includegraphics{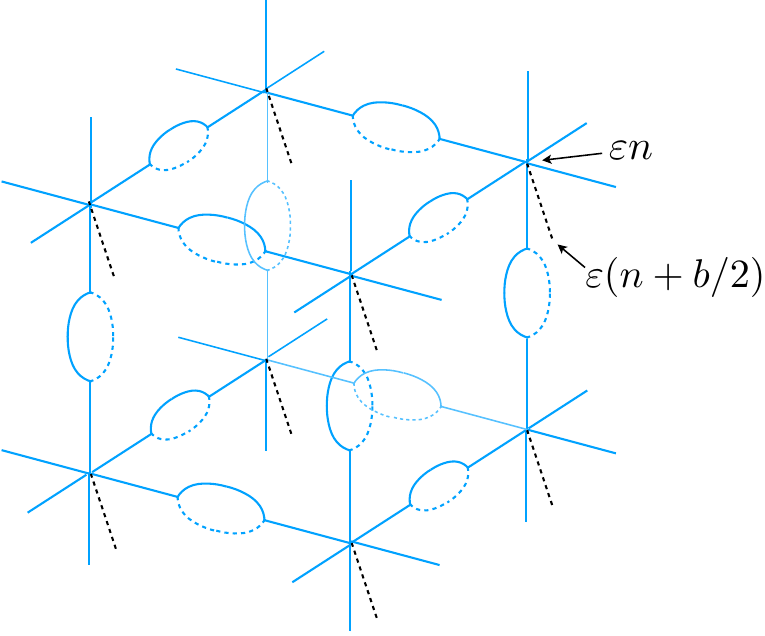}}}
}
\caption{\small Left: A simple example of a ``monolayer" 3D periodic graph $\GGeps$ with stiff (solid) and soft (dotted) edges.  It consists of the scaled integer grid $\eps\ZZ^3$ with each pair of nearest neighbors connected by a compound edge $\bedge$ shown in Fig.\,\ref{fig:bdedge}.  Middle: The bilayer graph consists of two copies of the monolayer graph, connected periodically by a soft edge connecting corresponding grid points.  Right: Half of the bilayer graph, which is a fundamental domain for the involution that interchanges the two copies (blue and red) of the monolayer within the bilayer graph.  It consists of the monolayer graph with a decoration $\dedge$ attached to each vertex $v\in\eps\ZZ^3$.}
\label{fig:SingleDouble3D}
\end{figure}

\begin{figure}[ht] 
\centerline{
\raisebox{17pt}{\scalebox{0.36}{\includegraphics{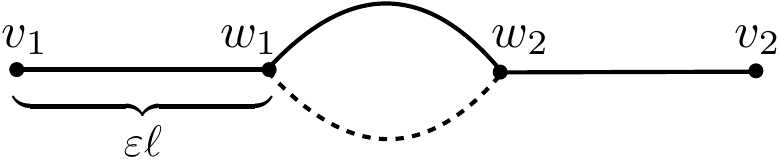}}}
\hspace{4em}
\scalebox{0.29}{\includegraphics{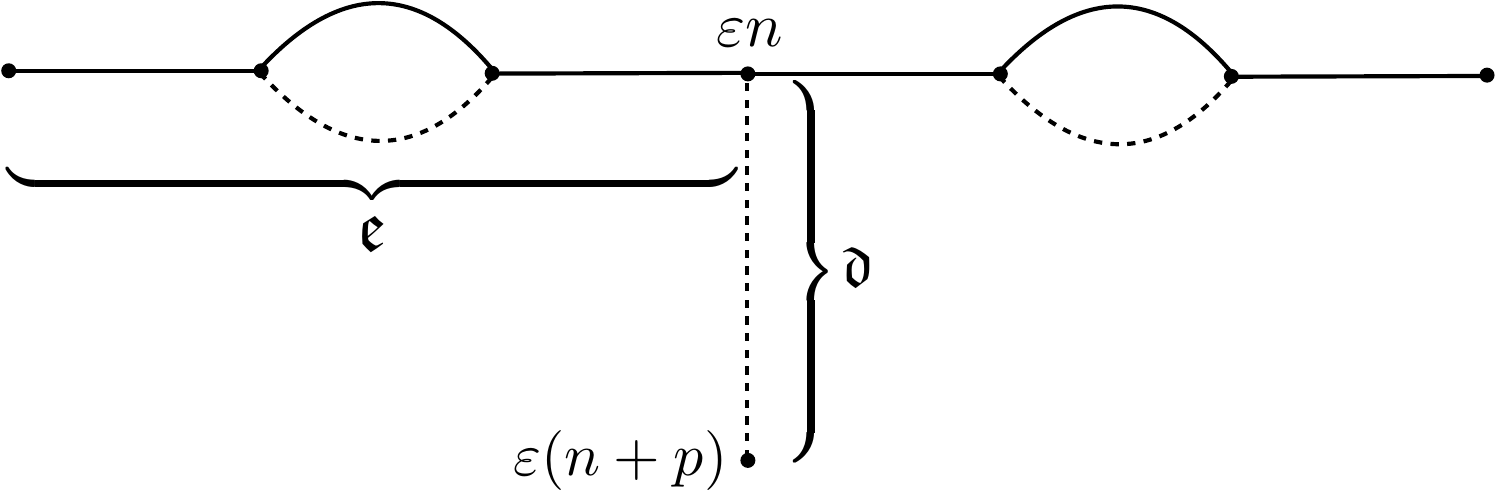}}
}
\caption{\small Left: A ``compound edge" $\bedge$ that serves as the basic element in the monolayer bulk metric graph.  The solid edges are stiff, and the dotted one is soft.  Each edge has length $\eps\ell$, and the distance between $v_1$ and $v_2$, when embedded in~$\RR^d$ is~$\eps$.  Right: A one-dimensional depiction of attaching a decoration edge $\dedge$ to each vertex of the lattice $\eps\ZZ^d$.}
\label{fig:bdedge}
\end{figure}

\section{Monolayer metric graph and defect states}\label{sec:single}

We will describe a specific prototypical metric graph structure as a prototype of a family of structures that are handled by the analysis of the specific one without invoking any essential changes.

\subsection{Periodic monolayer graph}\label{sec:periodicsingle}

The monolayer metric graph is constructed by connecting copies of an $\eps$-dependent {\em compound resonant edge}~$\bedge$, depicted in Fig.\,\ref{fig:bdedge}, to form a periodic square-lattice graph with $\eps$ as its characteristic length scale.  Each pair of nearest-neighbor vertices in the lattice~$\eps\ZZ^d\subset\RR^d$ is joined by a copy of~$\bedge$, as depicted in Fig.\,\ref{fig:SingleDouble3D}, resulting in each point in $\eps\ZZ^d$ having $2d$ copies of $\bedge$ emanating from it.  Of course, $\bedge$ has to be constructed so that, when embedded into $\RR^d$, the distance between its endpoints is~$\eps$.
Then, to each vertex of this graph located at $\eps n\!\in\!\eps\ZZ^d$, an identical soft decoration edge~$\dedge$ is attached, with one vertex at $\eps n$ and the other vertex at $\eps(n+p)$, for some fixed vector $p\in\RR^d$; see Fig.\,\ref{fig:bdedge}.  The resulting graph is denoted by $\Ge$.  In dimension $d\geq3$, the vector $p$ chosen such that $\Ge$ has no self-intersections.  (A more general decoration graph~$\dedge$ can be attached at each vertex.)

Let each edge $e=(v_1,v_2)$ in $\Ge$ be identified with a coordinate interval $x^e\in\eps[x^e_1,x^e_2]$, with $x^e_1$ (resp.\,$x^e_2$) being identified with $v_1$ (resp.\,$v_2$).  (The coordinate of the same edge directed in the opposite direction, namely $(v_2,v_1)$, is the coordinate of $e$ reversed to the opposite direction.)  This renders $\Ge$ a metric graph.  Observe that the underlying combinatorial graph remains unchanged with $\eps$, but scaling the edge coordinates by $\eps$ changes the metric graph.

This paragraph describes how to construct an elliptic differential operator $\He$ on the metric graph~$\Ge$.  After that, we will specialized the construction to two operators $\Heper$ and $\Hedef$ that we will study.  To the arbitrary edge~$e$ is assigned a differential operator that acts on complex valued functions $u$ in the Sobolev space $W^{1,2}(e)$ ($u$ has absolutely continuous derivative on the closed interval) and a ``flux" $\Phi[u]\!=\!\Phi_e[u]$ at the vertices of the edge,
\begin{equation}\label{HandPhi}
  u \mapsto -\frac{d}{dx}a(x)\frac{du}{dx},
  \qquad
  \Phi[u](v_1) = a(\eps x_1)u'(\eps x_1),\quad \Phi[u](v_2)=-a(\eps x_2)u'(\eps x_2)
\end{equation}
(let us assume that $a$ is smooth for this discussion).  The value $\Phi_e[u](v)$ is the flux of $u$ out of the vertex $v$ into the edge~$e$.  On a {\em stiff edge}, $a(x)$ is strictly $O(1)$ as $\eps\to0$ and we take it to be a constant $\sigma^2$, which we will set to $\sigma=1$ below.  On a {\em soft edge}, $a(x)\!=\!(\sigma\eps)^2$.  The domain of the operator requires a condition on how functions on $\Ge$ behave at the vertices.  To define this condition, fix a vertex $v$ of $\Ge$, and consider functions $u_e\in W^{1,2}(e)$ that are defined on the edges $e$ incident to $v$ such that their values $u_e(v)$ at $v$ are all equal to a common value~$u(v)$.  Then the following Robin vertex condition makes sense:
\begin{equation}\label{Rcondition}
  \sum_{e:v\sim e} \Phi[u_e](v) \;=\; \eps\alpha_v\,u(v),
\end{equation}
in which $\alpha_v$ is a real number or~$\infty$.  If $\alpha_v\!=\!0$, this is known as the Neumann, Kirchoff, or free vertex condition.  If $\alpha_v=\pm\infty$, the condition is interpreted to mean that $u(v)=0$ (with no conditions on the fluxes), and it is known as the Dirichlet vertex condition.  The domain of $\He$, as an operator in $L^2(\Ge)$,~is
\begin{equation}\label{DomHe}
  \Dom(\He) \;=\; 
  \left\{ u\in L^2(\Ge) \cap \oplus_{e\in\Ge} W^{1,2}(e) : \text{$u$ is continuous and satisfies (\ref{Rcondition}) } \forall v\in\Ge \right\},
\end{equation}
and its action is $u\mapsto-d(a(x)du/dx)/dx$ on each edge.  The operator $\He$ is self-adjoint in~$L^2(\Ge)$.  

The condition $\alpha_v>0$ corresponds to a typical linear spring with spring constant $\alpha_v$ providing a restoring force at vertex~$v$ directed in the opposite direction of the displacement from~$u\!=\!0$.  The condition $\alpha_v\!=\!0$ means that there is no spring at $v$, and $\alpha<0$ describes an artificially created auxetic spring.

By taking the vertex condition to be spatially homogeneous, one obtains a periodic metric-graph operator with period $\eps$; we call the operator~$\Heper$.  By this, we mean that
\begin{equation}
  \alpha_v = \alpha \in\RR
  \qquad \forall v\in\eps\ZZ^d.
\end{equation}
The parameter $\alpha_v$ can depend on the vertex within a single copy of $\bedge$ or $\dedge$, but the values are the same from copy to copy.  For example, $\dedge$ could be a single dangling soft edge, with one edge in the lattice $\eps\ZZ^d$ and the other (terminal) edge having either the Dirichlet or Neumann condition, uniformly throughout the graph.  The operator $\Heper$ commutes with the translational symmetry group $\ZZ^d$, where $n\in\ZZ^d$ effectuates a shift by~$\eps n$.

\subsection{Defect states}\label{sec:defectstates}

A local defect is created in $\Ge$ by adding a term to the Robin parameter $\alpha$ at the vertices located within a finite region of the lattice~$\eps\ZZ^d$.  A nonempty open bounded subset $\Omega\subset\RR^d$ identifies the defect region.  A~defect is placed at each vertex of $\Ge$ lying in the set $\Omega_\eps=\eps\ZZ^d\cap\Omega$ by changing the Robin parameter.  These sets are depicted in Fig.\,\ref{fig:DefectRegion}.  Specifically, a defect distribution is given by a smooth enough function (smoothness to be specified later) $\Psi:\RR^d\to\RR$ with $\supp\Psi=\overline\Omega$; then the defect in $\Ge$ is defined by adding~$\mu\Psi(\eps n)$ to~$\alpha$, with $\mu>0$ controlling the defect strength.  Thus the vertex condition (\ref{Rcondition}) at $v=\eps n$ ($n\in\ZZ^d$) is replaced~by
\begin{equation}\label{Rdefcondition}
  \sum_{e:\,\eps n\sim e} \Phi[u_e](\eps n) \;=\; \eps\big(\alpha+\mu\Psi(\eps n)\big)\,u(\eps n)
\end{equation}
in (\ref{DomHe}) to obtain the domain of the defective operator~$\Hedef$.

We are interested in the bound-state problem
\begin{equation}\label{BS1}
  (\Hedef - k^2)u \;=\; 0,
\end{equation}
in which $u\in\Dom(\Hedef)$ and $k$ is interpreted as the nondimensionalized frequency of a time-harmonic field.

\begin{figure}[ht] 
\centerline{
\scalebox{0.34}{\includegraphics{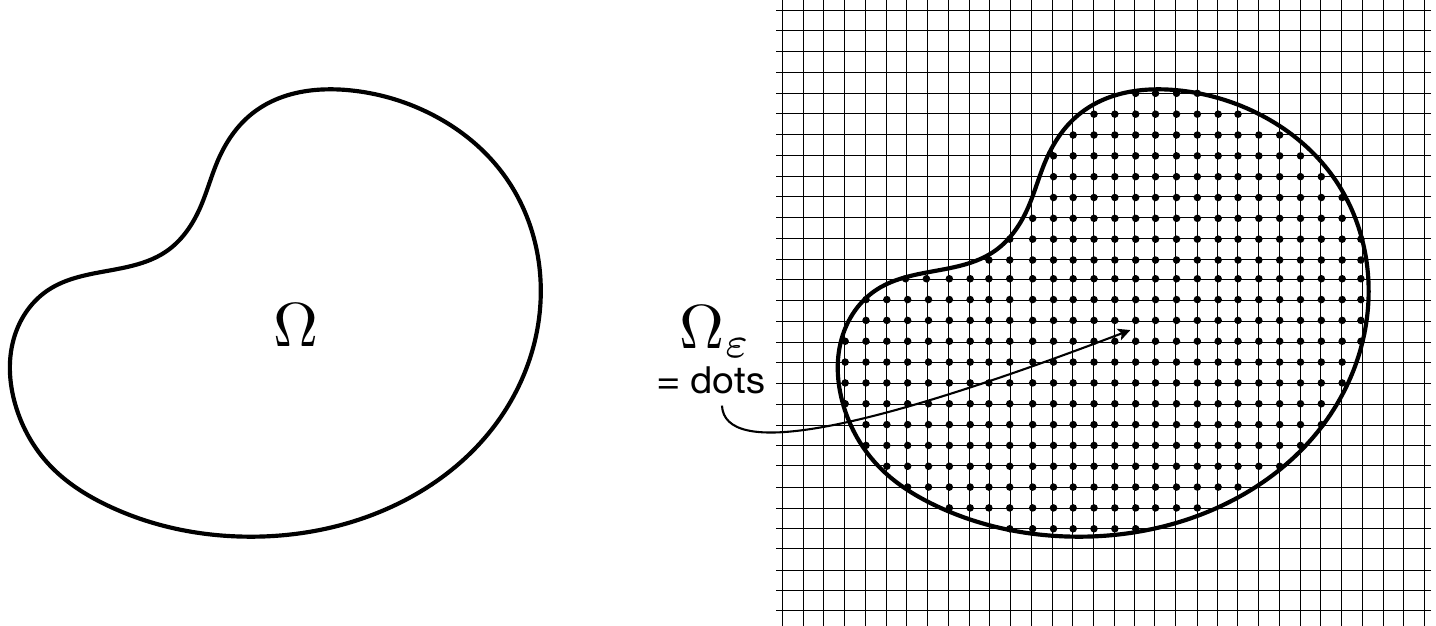}}
}
\caption{\small Left:  The continuum defect region is the open set $\Omega\subset\RR^d$, pictured here in $\RR^2$ as the interior region of a smooth simple curve.  Right: The set $\Omega_\eps$ of vertices that comprise the support of the defect for the operator $\He$, on the metric graph of period $\eps$, consists of those points of $\eps\ZZ^d$ that lie within $\Omega$.}
\label{fig:DefectRegion}
\end{figure}

\subsection{Spectrum of the periodic and defective operators}

As $\Heper$ is a periodic metric graph operator, is spectrum consists of intervals (bands) of continuous spectrum, each determined by the image of a band function, which is a frequency-valued function of wavevector.  The band functions are the branches of a dispersion function.  (A periodic graph operator can also have ``flat bands functions", when a constant band function corresponds to an infinite-multiplicity eigenvalue.)

Both $\Heper$ and $\Hedef$ are self-adjoint extensions of a single symmetric operator $\Hemin$ with additional homogeneous boundary conditions at the vertices $\Omega_\eps=\eps\ZZ^d\cap\Omega$, which are the vertices that participate in the defect.  This operator $\Hemin$ retains the action $u\mapsto-d(a(x)du/dx)/dx$ on the edges, and its domain is
\begin{equation}\label{DomHemin}
\begin{aligned}
  \Dom(\Hemin) \;=\; &
  \left\{ u\in L^2(\Ge) \cap \oplus_{e\in\Ge} W^{1,2}(e) \;{\big|}\; \forall v\in\Ge\!\setminus\!\Omega_\eps\; \text{$u$ satisfies (\ref{Rcondition})};\, \right.\\
& \left.\;\;  \forall v\in\Omega_\eps\; u(v)=0 \text{ and } u'(v)=0 \, \right\}.
\end{aligned}
\end{equation}
(Note that functions in this domain are continuous.)  Each vertex $v$ in the finite set $\Omega_\eps$ contributes $2\deg(v)$ to the deficiency indices of $\Hemin$.  Theorem 4.1.4 of \cite{AlbeverioKurasov1999a} states that two self-adjoint extensions of the same symmetric operator with finite deficiency indices have the same continuous spectrum.  We conclude that $\Heper$ and $\Hedef$ have the same continuous spectrum.

We are interested in the finite-multiplicity eigenvalues of $\Hedef$, particularly those that are embedded in the continuous spectrum.

\section{Bilayer metric graph and defect states}\label{sec:bilayer}

This section concerns the coupling of two monolayer metric graph operators to create a bilayer operator.  The construction allows decoupling by symmetry into even and odd states.  Each of these two components is equivalent to a monolayer with different decoration~$\dedge$ and therefore different spectral bands.  For a frequency that is within a spectral gap of one component but in a spectral band of the other component, we construct in \S\ref{sec:BIC} spectrally embedded defect states.

\subsection{Connecting two monolayer graphs}

A bilayer metric graph is created by connecting two identical copies of $\Ge$ with the decorations $\dedge$ omitted.  Two copies of $\Ge$ are placed next to each other in $\RR^d$---one with vertices on the lattice $\eps\ZZ^d$ and another with vertices at a shifted lattice $\eps(\ZZ^d + b)$, where $b\!\in\!\RR^d$ is chosen in such a way that the two copies do not intersect each other.  Then, for each $n\in\ZZ^d$, the vertex $\eps n$ is connected to the vertex $\eps (n+b)$ by a soft edge of length~$\eps|b|$, denoted by~$e_n$.  The resulting bilayer graph is denoted by~$\Gebi$.

Let $\Hebi$ denote the operator on the bilayer graph $\Gebi$ constructed according to the previous section.  Particularly, denote by $\Heperbi$ the periodic bilayer operator.  This means that the action of the operator is $u\mapsto-d(a(x)du/dx)/dx$ on each edge, the domain is of the form (\ref{DomHe}), and $\alpha_v\!=\!\alpha$ at all vertices $v$ located at the sites $\eps(\ZZ^d)$ and $\eps(\ZZ^d + b)$.
Let $\Hedefbi$ denote the corresponding defective operator.  Specifically, the defect is given by adding the quantity $\Psi(\eps n)$ to the Robin parameter on both layers, that is, for each $n\in\ZZ^d$ and both vertices $v\in\{\eps n,\,\eps(n+b)\}$,
\begin{equation}\label{Rdefconditionbi}
  \sum_{e:\,v\sim e} \Phi[u_e](v) \;=\; \eps\big(\alpha+\mu\Psi(\eps n)\big)\,u(v).
\end{equation}

Again, we are interested in the bound-state problem
\begin{equation}\label{BS2}
  (\Hedefbi - k^2)u \;=\; 0,
\end{equation}
in which $u\in\dom(\Hedefbi)\subset L^2(\Gebi)$.  Of particular interest are {\em bound states in the continuum}, meaning that the eigenvalue $k^2$ of the bound state lies within a band of continuous spectrum of the operator.  Such states can occur in bilayer structures, as we show later.

\subsection{Decoupling into monolayer components}

An order-2 symmetry of the bilayer graphs described above permits an orthogonal decomposition into two noninteracting monolayer decorated graph operators.  Particularly, the spectrum of the bilayer graph is the union of the spectra of these two monolayer graphs.

The graph $\Gebi$ has a natural reflection involution $\iota$ that interchanges the lattices $\eps\ZZ^d$ and $\eps(\ZZ^d+b)$ (so $i(\eps n)=\eps(n+b)$) and permutes the edges correspondingly.  This involution is carried over to functions $f$ on $\Gebi$ by $(\iota f)(x)=f(\iota x)$.
The eigenspaces of this action of $\iota$ corresponding to eigenvalues $+1$ and $-1$ are the even and odd subspaces of $L^2(\Gebi)$; these lead to the decomposition $L^2(\Gebi)=L^{2,\mathrm{e}}(\Gebi)\oplus L^{2,\mathrm{o}}(\Gebi)$.  The operators $\Heperbi$ and $\Hedefbi$ commute with $\iota$, and are thus decomposed~as
\begin{equation}\label{eodecomposition}
  \Heperbi \;=\; \Heperbie \oplus \Heperbio,
  \qquad
  \Hedefbi \;=\; \Hedefbie \oplus \Hedefbio.
\end{equation}
The domain of $\Heperbie$ consists of the even functions in the domain of~$\Heperbi$,
\begin{equation}
  \Dom(\Heperbie) \;=\; \left\{ f\in\Dom(\Heperbi) : \iota f = f \right\},
\end{equation}
and the domain of $\Heperbio$ has $\iota f = -f$; the domains of $\Hedefbie$ and $\Hedefbio$ are analogous.

The operators $\Heperbie$ and $\Heperbio$ can be identified with monolayer operators by restricting them to a fundamental region of the involution~$\iota$.  A fundamental region consists of ``half of $\Gebi$" obtained by cutting the connecting edges $e_n$ in half, that is, a monolayer square-lattice graph with vertices at $\eps\ZZ^d$ decorated with soft edges of length~$\eps|b|/2$.  Each copy of this decoration edge $\dedge$ has one vertex attached to the monolayer graph at $\eps n$ and a terminal vertex at $\eps(n+b/2)$ (so $p\!=\!b/2$ in \S\ref{sec:periodicsingle}).  We denote this graph by $\Ge$, as it coincides with the decorated monolayer construction in the previous section.  It is depicted in Fig.\,\ref{fig:SingleDouble3D}.

Functions in $L^{2,e}(\Ge)$ are determined by their values on the fundamental region~$\Ge$ by extending functions on $\Ge$ to all of $\Gebi$ by the reflection~$\iota$.  This extension is an isomorphism (up to constant) $\mathcal{E}:L^2(\Ge)\to L^{2,e}(\Gebi)$, which is inverted by the restriction of functions from $L^{2,e}(\Gebi)$ to $L^2(\Ge)$.  The restriction of functions in $\Dom(\Heperbie)$ to~$\Ge$ satisfy the Neumann condition at the terminal vertices of~$\Ge$.  Therefore, pulling the operator $\Heperbie$ back under $\mathcal{E}$ yields a unitarily equivalent operator in $L^2(\Ge)$ whose domain is (\ref{DomHe}) with the Neumann condition on the terminal vertices of the decorations.  Denote this self-adjoint operator by~$\Hepere$.
Similarly, the operator $\Heperbio$ is unitarily equivalent to an operator $\Hepero$ in $L^2(\Ge)$.  Its domain has the Dirichlet condition at the terminal endpoints.

In summary, the study of the bilayer operator $\Heperbi$ is reduced to the study of two monolayer operators $\Hepere$ and $\Hepero$.  These differ only in the self-adjoint condition at the terminal vertices.  Particularly, the spectrum of $\Heperbi$ is the union of the spectra of these two.

An analogous argument can be made for the defective bilayer operator $\Hedefbi$.  It is unitarily equivalent to the direct sum of two self-adjoint monolayer operators $\Hedefe$ and $\Hedefo$.  Observe that the defective Neumann operator $\Hedefe$ (resp. defective Dirichlet operator $\Hedefo$) is related to the periodic Neumann operator $\Hepere$ (resp. $\Hepero$) exactly as $\Hedef$ is constructed from $\Heper$ in~\S\ref{sec:defectstates}.

Briefly, we conclude that 
\begin{equation}\label{nddecomposition}
  \Heperbi \;\cong\; \Hepere \oplus \Hepero,
  \qquad
  \Hedefbi \;\cong\; \Hedefe \oplus \Hedefo.
\end{equation}

\subsection{Defect states in the continuum}

Because of the orthogonal decomposition $\Hedefbi\cong\Hedefe\oplus\Hedefo$, the bound state equation (\ref{BS2}) is equivalent to the simultaneous pair
\begin{equation}\label{BS3}
  (\Hedefe - k^2)u_\eee \;=\; 0,
  \qquad
  (\Hedefo - k^2)u_\ooo \;=\; 0
\end{equation}
in which at least one of $u_\eee$ and $u_\ooo$ is nonzero.

If $k^2\in\RR$ lies within the resolvent set (a spectral gap) of one of the corresponding periodic operators, say $\Hepere$, then one can prove (see Section~\ref{sec:results}) that, for fixed $\eps$, there exists a sequence of values of $\mu$, diverging to infinity (see (\ref{Rdefcondition}) for $\mu$) for which $(\Hedefe - k^2)u_\eee\!=\!0$ has a nonzero solution~$u_\eee$ in $\Dom(\Hedefe)$.  By taking $u_\ooo\!=\!0$, one obtains a bound state for the defective bilayer operator~$\Hedefbi$.  If $k^2$ lies within a spectral band of $\Hepero$ and therefore also of $\Hedefo$, then $k^2$ is a spectrally embedded eigenvalue of~$\Hedefbi$.

\section{Reduction to discrete graph operator}\label{sec:reduction}

The vertex condition (\ref{Rcondition}) involves fluxes of a field $u$ at a vertex $v$ coming from all edges incident to~$v$.  When $u$ satisfies the eigenvalue equation $(\He-k^2)u=0$, the fluxes are determined by the values of $u$ at the vertices, except when $k^2$ is a Dirichlet eigenvalue of an edge.  This allows one to express the eigenvalue equation as a discrete graph operator equation, with the operator depending on~$k$.  The basic object in this construction is the Dirichlet-to-Neumann (DtN) matrix for an edge, which relates the values at the endpoints to the fluxes there.

Starting at this point, we set some parameters in our metric graphs; these choices simplify the expressions but otherwise essentially do not compromise the generality of the calculations.
 Let an edge of the $\eps$-scale graph be coordinatized by a variable $x\in[0,\eps\ell]$, where $\ell$ is a parameter.  If the edge is stiff, we take the stiffness in (\ref{HandPhi}) to be $a(x)\!=\!\ell^2$, and if the edge is soft, we use $a(x)\!=\!(\eps\ell)^2$.  We also take the length of the connecting edge to be $|b|=2|p|=2\ell$, so the length of the decoration for the even and odd components is~$\ell$.

\subsection{Dirichlet-to-Neumann matrix for an edge}

 The DtN matrix, or $M$-matrix, for an edge $e$ joining vertices $v_1$ and $v_2$ is a $2\times2$ matrix $M_e(k,\eps)$ that takes Dirichlet data to Neumann (flux) data of $u$ at the boundary (vertices) of an edge.  This means that if $u$ satisfies $(\He-k^2)u=0$ on $e$, then
\begin{equation}\label{DtNgen}
  M_e(k,\eps) \col{1.2}{u(v_1)}{u(v_2)} = \col{1.2}{\Phi[e](v_1)}{\Phi[e](v_2)}.
\end{equation}
Appendix~\ref{sec:calculations} contains details that lead to
\begin{equation}\label{DtNe}
  M_\mathrm{stiff}(k,\eps) \;=\; \frac{\ell k}{\sin k\eps}\mat{1.3}{-\cos k\eps}{1}{1}{-\cos k\eps},
  \qquad
  M_\mathrm{soft}(k,\eps) \;=\; \frac{\eps\ell k}{\sin k}\mat{1.3}{-\cos k}{1}{1}{-\cos k}.
\end{equation}

Similarly, a decoration edge $\dedge$ has one terminal vertex where a homogeneous self-adjoint condition is placed, and its DtN matrix is scalar, as it only concerns the value and flux of $u$ at the other vertex $v$ of~$\dedge$.  For a soft decoration, the notation
\begin{equation}\label{Robinsoft}
  \Phi[u](v) = \eps\, m(k)u(v)
\end{equation}
is convenient.  When the length of a soft decoration $\dedge$ is $\eps\ell$ and the terminal vertex has the Dirichlet or Neumann condition, $m(k)$ is equal~to
\begin{equation}\label{DtNdn}
  m_\mathrm{D}(k)=-\ell k\cot k,
  \qquad
  m_\mathrm{N}(k)=\ell k\tan k.
\end{equation}

\subsection{Dirichlet-to-Neumann matrix for a compound edge}\label{sec:DtNcompound}

Consider a basic ``compound edge" $\bedge$ as in Fig.\,\ref{fig:bdedge}\,(left) consisting of sub-edges of length~$\eps\ell$.  Denote the two outside (terminal) vertices by $v_1$ and $v_2$ and the interior vertices by $w_1$ and $w_2$, as in the figure.  The DtN matrix $M_{\bedge}(k,\eps)$ for $\bedge$, as a compound edge joining $v_1$ to $v_2$, is defined as the matrix $M_e(k,\eps)$ in~(\ref{DtNgen}).

This matrix can be obtained as a Schur complement of the $4\times4$ matrix $\tilde M_{\bedge}(k,\eps)$ that maps the values $[u(v_1),u(v_2),u(w_1),u(w_2)]^\top$ to the fluxes $[\Phi(v_1),\Phi(v_2),\Phi(w_1),\Phi(w_2)]^\top$, in which
\begin{equation}
  \Phi(v) = \sum_{v\sim e}\Phi[e](v).
\end{equation}
Combining the matrices in (\ref{DtNe}) for the four edges of the compound edge $\bedge$ leads~to 
\begin{equation}\label{DtNb1}
\tilde M_{\bedge}(k,\eps) \;=\;
  \renewcommand{\arraystretch}{1.5}
\ell k\left[
  \begin{array}{cc|cc}
    -\cot \eps k & 0 & \csc \eps k & 0 \\
    0 & -\cot \eps k & 0 & \csc \eps k \\\hline
    \csc \eps k & 0 & -2\cot \eps k -\eps \cot k & \csc\eps k + \eps \csc k \\
    0 & \csc\eps k & \csc\eps k + \eps \csc k & -2\cot \eps k -\eps \cot k
  \end{array}
\right]
=\;\ell k \left[\!\!
  \begin{array}{cc}
     A & B \\
    B & D
  \end{array}
\!\right].
\end{equation}

Let $g_v\!=\![u(v_1),u(v_2)]^\top$ and $g_w\!=\![u(w_1),u(w_2)]^\top$ and $\phi_v\!=\![\Phi(v_1),\Phi(v_2)]^\top$.  Set $[\Phi(w_1),\Phi(w_2)]^\top\!=\![0,0]^\top$, which amounts to imposing the Kirchoff condition (Robin parameter equal to zero) at the internal vertices of~$\bedge$.  We obtain
\begin{equation}
  \ell k\mat{1.1}{A}{B}{B}{D} \col{1.1}{g_v}{g_w} \;=\; \col{1.1}{\phi_v}{0}.
\end{equation}
Eliminating $g_w$ yields
\begin{equation}
  M_{\bedge}(k,\eps)g_v \;=\; \ell k \left(A - BD^{-1}B\right) g_v \;=\; \phi_v.
\end{equation}
By denoting the Schur complement by $\Sigma\!=\!(A - BD^{-1}B)$, we obtain
$M_\bedge(k,\eps) \;=\; \ell k\hspace{1pt}\Sigma(k,\eps)$.
A computation detailed in the appendix yields
\begin{equation}
   \Sigma(k,\eps) \;=\; \frac{-1}{\sin\eps k} \left\{ \cos\eps k I_2 + \frac{1}{2} \mat{1.1}{\tau_++\tau_-}{\tau_+-\tau_-}{\tau_+-\tau_-}{\tau_++\tau_-} \right\}
\end{equation}
with
\begin{equation}
  \tau_\pm \;=\; \big( \pm1-2\cos\eps k + \eps\sin\eps k\csc k (\pm1-\cos k) \big)^{-1}.
\end{equation}

\subsection{Monolayer discrete graph operator family}

We are interested in single-frequency bound states, and therefore in the source-free equation
\begin{equation}\label{Eeqn}
  (\He-k^2)u=0,
\end{equation}
for a fixed real value of $k$.  This equation allows generally $u\in H^2_\mathrm{loc}(\Gamma)$ with $u$ being continuous at the vertices, although to be a bound state, $u$ must also lie in~$L^2(\Gamma)$.  This section derives a discrete equation for the values of $u$ on the lattice $\eps\ZZ^d$ that is a homogeneous discrete Laplacian with $k$-dependent potential; see equations (\ref{equivalence1}) and~(\ref{dlaplacian}).

At the vertex located at $\eps n\in\eps\ZZ^d$, denote the Robin parameter simply by $\alpha_n$; thus for the homogeneous graph with a local defect, $\alpha_n=\alpha+\mu\Psi(\eps n)$.  Let $\bar u_n$ denote the value of a solution $u$ of (\ref{Eeqn}) at vertex~$\eps n$.  Let $\phi^{j\pm}_n$ denote the flux out of vertex $\eps n$ into the edge between $n$ and ${n\pm e_j}$, and let $\tilde\phi_n$ denote the flux out of $\eps n$ into the decoration edge.  The vertex condition (\ref{Rcondition})~is
\begin{equation}
  \sum_{j\in[1,d]} \left(\phi_n^{j-}+\phi_n^{j+}\right) + \tilde\phi_n \;=\; \eps\alpha_n\bar u_n.
  \label{Robin_cond}
\end{equation}
In view of (\ref{Robinsoft}), this becomes
\begin{equation}\label{generalRobin}
  \sum_{j\in[1,d]} \left(\phi_n^{j-}+\phi_n^{j+}\right) \;=\; \eps \varrho_n(k)\bar u_n,
\end{equation}
in which
\begin{equation}\label{Rnk}
  \varrho_n(k) \;=\; \alpha_n - m(k),
\end{equation}
with the choice of $m(k)$ depending on the homogeneous vertex condition at the endpoint of the decoration as in~(\ref{DtNdn}).
One can think of the coefficient on the right-hand side of (\ref{generalRobin}) as a frequency-dependent Robin parameter.

By extracting the first row of each of the equations
\begin{equation*}
  M_\bedge(k,\eps) \col{1}{\bar u_n}{\bar u_{n\pm e_j}} \,=\,\col{1}{\phi_n^{j\pm}}{\phi_{n\pm e_j}^{j\mp}}
  \qquad (n\in\ZZ^d)
\end{equation*}
and summing those over $j\in[1,d]$ and $\pm\in\{+,-\}$, for fixed $n$, one obtains
\begin{equation}
  \sum_{j\in[1,d]} \left(\phi_n^{j-}+\phi_n^{j+}\right) \;=\; \ell k\Big( 2d\,\Sigma_{11}\bar u_n + \,\Sigma_{12}\!\!\sum_{j\in[1,d]} (\bar u_{n-e_j}+ \bar u_{n+e_j}) \Big),
\end{equation}
in which $\Sigma_{ij}$ are the entries of the matrix~$\Sigma$.
This, together with the generalized Robin condition (\ref{generalRobin}) yields, for $k\not=0$,
\begin{equation}\label{Jacobi1}
  \big( 2d\,\ell k\,\Sigma_{11} - \eps \varrho_n(k) \big) \bar u_n + \,\ell k\,\Sigma_{12}\!\!\sum_{j\in[1,d]} (\bar u_{n-e_j}+ \bar u_{n+e_j}) \;=\; 0,
\end{equation}
or
\begin{equation}
  \sum_{j\in[1,d]} (\bar u_{n-e_j}+ \bar u_{n+e_j}) + \tilde V(k,\eps,n) \bar u_n \;=\; 0,
\end{equation}
with
\begin{equation}\label{Vken}
  \tilde V(k,\eps,n) \;=\; \frac{2d\,\ell k\,\Sigma_{11}(k,\eps) - \eps \varrho_n(k)}{\ell k\,\Sigma_{12}(k,\eps)}.
\end{equation}

We conclude, except when $k^2$ is an Dirichlet eigenvalue of an edge, that $u$, continuous and in $H^2(e)$ for each edge~$e$, is a solution to $\He u=k^2u$ if and only if its restriction $\bar u$ to $\eps\ZZ^d$ satisfies a discrete graph problem.  Precisely,
\begin{equation}
  (\He - k^2) u= 0
  \quad\iff\quad
  H(k,\eps)\bar u = 0,
\end{equation}
in which $H(k,\eps)$ is the Jacobi operator.
\begin{equation}
  H(k,\eps) \;=\; \sum_{j\in[1,d]}\left( S_j+S_j^{-1} \right) + \tilde V(k,\eps;\cdot),
\end{equation}
where $S_j$ is the $j^\mathrm{th}$ elementary shift operator acting on functions defined on~$\eps\ZZ^d$.

A  calculation in the Appendix shows that the diagonal potential function is equal~to
\[
\tilde V(k,\eps,n) \;=\; -2d+\eps^2\,V(k,\eps,n),
\]
in which
\begin{equation}\label{V}
  V(k,\eps,n) \;=\; d\left[ 9k^2 + 6k\left( \csc k - \cot k \right) \right] - \frac{3}{\ell}\left( \alpha_n - m(k) \right) + \eps^2\,\left[ g_1(k,\eps) + g_2(k,\eps)(\alpha_n-m(k))\right],
\end{equation}
in which the functions $g_1(k,\eps)$ and $g_2(k,\eps)$ are computed in the Appendix.

Thus $H$ is a discrete Laplacian plus a potential,
\begin{equation}\label{Apken}
  H(k,\eps) \;=\; \sum_{j\in[1,d]}\left( S_j+S_j^{-1} \right) - 2d\, + \eps^2\,V(k,\eps;\cdot)
  \;=\; \Delta_\mathrm{d} + \eps^2\,V(k,\eps;\cdot).
\end{equation}
The appropriately scaled discrete Laplacian is
\begin{equation}\label{dlaplacian}
  \Delta_\eps \;=\; \eps^{-2}\Delta_\mathrm{d} \;=\; \eps^{-2}\! \sum_{j\in[1,d]}\left( S_j+S_j^{-1} -2 \right),
\end{equation}
and we obtain the equivalence
\begin{equation}\label{equivalence1}
  (\He - k^2) u= 0 
  \quad\iff\quad
  \big( \Delta_\eps + V(k,\eps;\cdot) \big)\bar u = 0.
\end{equation}
The operator $-\Delta_\eps$ is a positive operator having absolutely continuous spectrum equal to the interval~$[0,4d\eps^{-2}]$.

\subsection{Bilayer discrete graph operator family}

The analogous equation for a bilayer metric graph operator 
\begin{equation}\label{Eeqnbi}
  (\Hebi-k^2)u=0,
\end{equation}
is treated identically as the previous section and results in a discrete operator on the union $\eps\ZZ^d\cup\eps(\ZZ^d+b)$ of two lattices.  But because of the decomposition (\ref{nddecomposition}), we can directly apply the analysis of the previous section to the Neumann and Dirichlet components.  We obtain that, for an even function $u_\mathrm{e}$ on $\Gebi$,
\begin{equation}\label{equivalence2}
  (\Hebi - k^2) u_\mathrm{e} = 0 
  \quad\iff\quad
  \big( \Delta_\eps + V_\mathrm{N}(k,\eps;\cdot) \big)\bar u_\mathrm{e} = 0,
\end{equation}
in which $\bar u_\mathrm{e}$ is the restriction of $u_\mathrm{e}$ to~$\ZZ^d$ and $V_\mathrm{N}$ is the potential $V$ above resulting from using $m_\mathrm{N}$ in place of $m$ (see \ref{DtNdn}) in the definition of $\varrho_n$ (see \ref{Rnk}).  And for an odd function $u_\mathrm{o}$ on $\Gebi$,
\begin{equation}\label{equivalence3}
  (\Hebi - k^2) u_\mathrm{o} = 0 
  \quad\iff\quad
  \big( \Delta_\eps + V_\mathrm{D}(k,\eps;\cdot) \big)\bar u_\mathrm{o} = 0,
\end{equation}
in which $\bar u_\mathrm{o}$ is the restriction of $u_\mathrm{o}$ to~$\ZZ^d$ and $V_\mathrm{D}$ is the potential $V$ with $m_\mathrm{D}$ in place of~$m$.

\subsection{Defect eigenvalues}
\label{sec:defect_eigenvalues}

Because of the previous two sections, finding eigenvalues of the bilayer graph operator $\Hedefbi$ (with symmetric defect) reduces to finding solutions of the monolayer discrete graph equation
\begin{equation}\label{BS4}
  ( \Delta_\eps + V(k,\eps;\cdot) )\bar u = 0,
\end{equation}
with $V\!=\!V_\mathrm{D}$ or $V\!=\!V_\mathrm{N}$ and $\alpha_n=\alpha+\mu\Psi(\eps n)$.  Let $V_\free(k,\eps)$ denote the corresponding discrete potential function for the ``free", undefective, system with constant Robin parameter~$\alpha$.

From the expression (\ref{V}) for $V(k,\eps,n)$, we obtain
\begin{equation}
  V(k,\eps;n) - V_\free(k,\eps) \;=\; \left[ -\frac{3}{\ell} + \eps^2g_2(k,\eps) \right] \mu\Psi(\eps n).
\end{equation}
Thus, (\ref{BS4}) is equivalent to
\begin{equation}\label{BS5}
  \left( -\Delta_\eps - V_\free(k,\eps) \right)\bar u \;=\; \left[ -\frac{3}{\ell} + \eps^2g_2(k,\eps) \right] \mu\Psi(\eps\,\cdot)\bar u
\end{equation}
for $\bar u\in\ell^2(\eps\ZZ^d)$.
The quantity on the right-hand side identifies a new eigenvalue parameter
\begin{equation}\label{lambda}
  \lambda \;=\; \lambda(k,\eps) \;=\; -\mu^{-1}\left[ \frac{3}{\ell} - \eps^2g_2(k,\eps) \right]^{-1},
\end{equation}
so that (\ref{BS5}) becomes
\begin{equation}\label{BS6}
  \left( -\Delta_\eps - V_\free(k,\eps) - \lambda(k,\eps)^{-1} \right)\bar u \;=\; 0.
\end{equation}

The function $V_\free(k,0)$ is the ``Zhikov function" for the ambient medium, indicating the effective square frequency of the medium at frequency~$k$.  We assume that $V_\free(k,0)<0$ so that, for sufficiently small $\eps$, the operator $-\Delta_\eps - V_\free(k,\eps)$ is invertible,
\begin{equation}
  \cR(k,\eps) \;:=\; \big( -\!\Delta_\eps - V_\free(k,\eps) \big)^{-1}.
\end{equation}
The bound-state problem (\ref{BS5}) becomes
\begin{equation}\label{BS6-7}
  \mu\left[ -\frac{3}{\ell} + \eps^2g_2(k,\eps) \right] \cR(k,\eps)\Psi(\eps\,\cdot)\,\bar u \;=\; \bar u
\end{equation}
or, equivalently, with $\bar v:=\Psi(\eps\,\cdot)\,\bar u$,
\begin{equation}\label{BS7}
    \Psi(\eps\,\cdot) \cR(k,\eps) \,\bar v \;=\; \lambda \,\bar v
\end{equation}
for $\bar v\in\ell^2(\Omega_\eps)$.

This last formulation parses the bound-state problem for a defective medium in terms of an eigenvalue problem for the resolvent of the ambient medium, within the defect region.  Specifically, for a fixed $k$, the eigenvalues of the operator $\Psi(\eps\,\cdot) \cR(k,\eps)$ in the defect region determine the defects for which the graph admits a bound state.  Note that $\ell^2(\Omega_\eps)$ is finite dimensional.

To recover the bound state in all of $\eps\RR^d$, one uses (\ref{BS6-7}):
\begin{equation}\label{recoverubar}
  \bar u \;=\; \lambda^{-1} \cR(k,\eps) \,\bar v\,.
\end{equation}
Then, the eigenfunction $u$ on the metric graph (Equation~\ref{Eeqn}) can be reconstructed from~$\bar u$ while preserving the norm.

\section{Homogenization results}\label{sec:results}

Before presenting the details of convergence of the metric graph bound-state problem to a continuum one,  this section provides a concise overview of the results of this paper in light of the foregoing exposition.  We present the heuristics of the homogenized operator as a continuum limit of graph operators and describe the main rigorous results of this paper and philosophy underlying the generation of bound states.

\subsection{Heuristics of the monolayer continuum limit}\label{sec:heuristics1}

We have seen that the bound-state problem for a monolayer metric graph at scale $\eps$ and frequency~$k$, namely $(\He-k^2)u=0$, is equivalent to a discrete Schr\"odinger operator problem~(\ref{BS4}).  When the Robin parameter is a constant $\alpha$ plus a local defect, the bound-state problem can be written in terms of the resolvent of a free discrete Schr\"odinger operator~(\ref{BS7}).

In~(\ref{BS4}), the discrete Laplace operator has the continuous Laplace operator as its continuum limit
\begin{equation}
  \Delta_\eps \;\to\; \Delta \;=\;  \sum_{j\in[1,d]}\!\frac{\partial^2}{\partial x_j^2},
\end{equation}
and the potential has limit
\begin{equation}
  V_\free(k,\eps) + \lambda(k,\eps)^{-1}\Psi(\eps n)
  \;\to\;
  V_\free(k,0) + \lambda(k,0)^{-1}\Psi(x),
\end{equation}
in which by definition $V_\free(k,0)=\lim_{\eps\to0}V_\free(k,\eps)$ ($\alpha_n\equiv\alpha$ in (\ref{V})), so that
\begin{equation}\label{Vfree0}
  V_\free(k,0) \;=\; d\left[ 9k^2 + 6k\left( \csc k - \cot k \right) \right] - \frac{3}{\ell}\left( \alpha - m(k) \right).
\end{equation}
The bound-state formulation (\ref{BS5}) has as its continuum limit
\begin{equation}\label{eig_prob_hom}
  \bigl(-\Delta-V_\free(k,0)\bigr)u(x)=\lambda^{-1}\Psi(x)u(x),
\end{equation}
or, equivalently,
\begin{equation}\label{eig_prob_hom2}
  \bigl(-\Delta-V_\free(k,0) + \lambda^{-1}\Psi(x)\bigr) u(x) \;=\; 0,
\end{equation}
which is a dispersive inhomogeneous linear PDE of the form Laplacian plus $k$- and $x$- dependent potential.
Equation (\ref{eig_prob_hom}) is equivalent to the following continuum version of~(\ref{BS7}):
\begin{equation}
    \Psi\cR(k,0)v=\lambda\, v.
\end{equation}
Here, $v(x)=\Psi(x)u(x)$ lies in $L^2(\Omega)$ and $\cR(k,0)$ is by definition the resolvent
\begin{equation}
  \cR(k,0):=\bigl(-\Delta-V_\free(k,0)\bigr)^{-1},
\end{equation}
when $V_\free(k,0)<0$.  The bound state in all of $\RR^d$ is recovered by
\begin{equation}\label{recoveru}
  u \;=\; \lambda^{-1} \cR(k,0) \,v\,.
\end{equation}

\subsection{Bilayer continuum limit as a coupled dispersive system}\label{sec:heuristics2}

The homogenization limit of the defect-state problem $(\Hedefbi - k^2)u\!=\!0$ (\ref{BS2}) for the bilayer metric graph operator is described by a PDE eigenvalue problem of the form $\dot H(k)u\!=\!\lambda u$.  It consists of two identical copies of a dispersive medium in $\RR^d$ with a dispersive coupling.  The partial differential operator for the ``monolayer" dispersive medium without defect is taken to be the formal continuum limit of a monolayer metric graph operator without decoration.  Thus it has the~form
\begin{equation}\label{PDE1}
  -\Delta-V_\mathrm{s}(k,0).
\end{equation}
To define $V_\mathrm{s}(k,0)$, let $V_\mathrm{s}(k,\eps)$ be the modification of (\ref{V}) obtained by putting $\alpha_n\!=\!\alpha$ and removing $m(k)$. 
Then set $V_\mathrm{s}(k,0)=\lim_{\eps\to0}V_\mathrm{s}(k,\eps)$, so that
\begin{equation}\label{Vs}
  V_\mathrm{s}(k,0) \;=\; d\left[ 9k^2 + 6k\left( \csc k - \cot k \right) \right] - \frac{3\,\alpha}{\ell}. 
\end{equation}

Two copies of this operator are then coupled by the DtN map for the soft connecting edge of length $2\ell\eps$, which is given in (\ref{DtNe2}) and a local defect supported on $\supp\Psi\!=\!\overline{\Omega}$ is added to each layer identically,
\begin{equation}\label{limitsystem}
  \dot H(k) \;:=\; 
  \mat{1.2}{\!-\Delta-V_\mathrm{s}(k,0)\!}{0}{0}{\!-\Delta-V_\mathrm{s}(k,0)\!}
  - \frac{3}{\ell} \frac{\ell k}{\sin 2k} \mat{1.2}{\!-\cos 2k\!}{1}{1}{\!-\cos 2k\!}
  + \frac{3}{\ell} \mat{1.2}{\!\mu\Psi(x)\!}{0}{0}{\!\mu\Psi(x)\!}.
\end{equation}
The eigenvalues of (\ref{DtNe2}) are $\eps\,m_\mathrm{D}(k)$ and $\eps\,m_\mathrm{N}(k)$, with eigenvectors $[1,-1]^t$ and $[1,1]^t$.  Thus this coupled operator is diagonalized~to
\begin{multline}\label{limitdiag}
  \mat{1.2}{\!-\Delta-V_\mathrm{s}(k,0)\!}{0}{0}{\!-\Delta-V_\mathrm{s}(k,0)\!}
  - \frac{3}{\ell} \mat{1.2}{m_\mathrm{D}(k)}{0}{0}{m_\mathrm{N}(k)}
  + \frac{3}{\ell} \mat{1.2}{\mu\Psi(x)}{0}{0}{\mu\Psi(x)}
  \\
  \;=\; \mat{1.4}{\!-\Delta-V_\mathrm{f,D}(k,0) - \lambda^{-1}\Psi(x)}{0}{0}{-\Delta-V_\mathrm{f,N}(k,0) - \lambda^{-1}\Psi(x)\!},
\end{multline}
in which $\lambda = -(3\mu/\ell)^{-1}$ and $V_\mathrm{f,D}$ and $V_\mathrm{f,N}$ are the free (constant-$\alpha$) potentials 
\begin{equation}\label{limitV}
  V_\free(k,0) \;=\; d\left[ 9k^2 + 6k\left( \csc k - \cot k \right) \right] - \frac{3}{\ell}\big( \alpha - m(k) \big)
\end{equation}
with $m(k)$ set to $m_\mathrm{D}(k)$ or $m_\mathrm{N}(k)$.
The operator in the first block acts on odd functions and the second acts on even functions, with respect to the order-2 symmetry of switching the layers.  These two blocks are precisely the formal continuum limits of the operator in~(\ref{BS6}).

These dispersive systems can be realized as the projection of a self-adjoint operator onto the macroscopic ``observable" $x$-dependent fields.  The self-adjoint operator couples the macroscopic fields to ``hidden", microscopic, fields, and it can be realized in an explicit way.  While the dispersive system is frequency dependent, the extended self-adjoint one is a closed system thus independent of frequency.  A detailed construction is presented in Appendix~\ref{sec:explicit-soft-limit}.

\subsection{What we prove}\label{sec:whatweprove}

In Section~\ref{sec:convergence}, we prove that the operators $\Psi(\eps\cdot)\cR(k,\eps)$ converge in norm to $\Psi\cR(k,0)$ as $\eps\to0$.  Since the former operator acts in $\ell^2(\Omega_\eps)$ and the limit operator acts in $L^2(\Omega)$, the task requires an interpolation that places all of these operators unitarily into the same space.  We realize this interpolation through the Fourier transform.  The proof of operator convergence constitutes the bulk of the work in this paper.

Once this convergence is established, it is relatively simple to construct bound defect states, including embedded ones for the bilayer system, albeit the computational details may be involved, as seen in Section~\ref{sec:BIC}.  First fix a real frequency~$k$ in a gap for the continuum operator $-\Delta-V_\free(k,0)$, which means $V_\free(k,0)<0$.  Let $\lambda$ be an eigenvalue of the compact operator $\Psi\cR(k,0)$.  Because of the operator convergence of $\Psi(\eps\cdot)\cR(k,\eps)$ to $\Psi\cR(k,0)$, for small enough $\eps$ the operator $\Psi(\eps\cdot)\cR(k,\eps)$ has an eigenvalue $\lambda_\eps$ that converges to $\lambda$ as $\eps\to0$.  The defect strength for the corresponding $\eps$-scale metric graph is 
\begin{equation*}
  \mu \;=\; -\lambda(k,\eps)^{-1}\left[ \frac{3}{\ell} - \eps^2g_2(\eps,0) \right]^{-1},
\end{equation*}
which converges to the defect value of $\mu$ in the continuum problem,
\begin{equation*}
  \mu \;=\; -\ell/(3\lambda).
\end{equation*}

One also obtains $L^2$ convergence of the bound states.  One has
\begin{equation*}
  \Psi(\eps\cdot)\cR(k,\eps)\bar v_\eps = \lambda_\eps\,\bar v_\eps\,,
  \qquad
  \Psi\cR(k,0) v = \lambda\,v,
\end{equation*}
with $\bar v_\eps \!\to\! v$ in $L^2$ in the sense of interpolation when suitably normalized.  Thus the bound states converge by the reconstruction formulas (\ref{recoverubar} with $\bar v\mapsto\bar v_\eps$ and $\lambda\mapsto\lambda_\eps$) and (\ref{recoveru}).

\begin{theorem}[Monolayer defect states]\label{thm:monolayer}
  Let $k$ be a real number such that $V_\free(k,0)\!<\!0$, and suppose that the locally defective dispersive PDE~(\ref{eig_prob_hom2}) has a square-integrable solution~$u$.  Then there exists $\epsilon'\!>\!0$ and a real-valued function $\mu(\eps)$ defined for $0\!<\!\eps\!<\!\eps'$  that converges to $-\ell/(3\lambda)$ as $\eps\!\to\!0$ and is such that the monolayer metric-graph problem~(\ref{BS1}) with defect strength $\mu\!=\!\mu(\eps)$ has a square-integrable solution~$u_\eps$ that converges to $u$ in $L^2$.  The corresponding eigenvalue $k^2$ is isolated in the spectrum of the monolayer operator~$\Hedef$.
\end{theorem}

To construct a spectrally embedded eigenvalue for a bi-layer metric graph, one starts with the two-component coupled continuum system~(\ref{limitsystem}), or equivalently its block-diagonal form~(\ref{limitdiag}) and chooses a frequency $k$ that lies in a gap of one of the components and in a band of the other, say $V_\mathrm{f,D}(k,0)<0$ and $V_\mathrm{f,N}(k,0)>0$.  The defect state constructed for the Dirichlet (odd) component will lie within the continuous spectrum of the Neumann (even) component and therefore within the continuous spectrum of the whole metric graph operator.  This is because the perturbation does not change the continuous spectrum (see Theorem 4.1.4 of \cite{AlbeverioKurasov1999a}).

\begin{theorem}[Spectrally embedded bilayer defect states]\label{thm:bilayer}
  Let $k$ be a real number such that $V_\mathrm{f,N}(k,0)\!<\!0$ and $V_\mathrm{f,D}(k,0)\!>\!0$ (or vice-versa),
  and suppose that the locally defective dispersive PDE system~(\ref{limitsystem}) has a square-integrable solution~$u$.  Then there exists $\epsilon'\!>\!0$ and a function $\mu(\eps)$ defined for $0\!<\!\eps\!<\!\eps'$ that converges to $-\ell/(3\lambda)$ as $\eps\!\to\!0$ and is such that the bilayer metric-graph problem~(\ref{BS2}) with defect strength $\mu\!=\!\mu(\eps)$ has a square-integrable solution~$u_\eps$ that converges to $u$ in $L^2$.  The corresponding eigenvalue $k^2$ is embedded in the continuous spectrum of the bilayer operator~$\Hedefbi$.
\end{theorem}

\section{Convergence of operators}\label{sec:convergence}

This section provides a proof that $\Psi(\eps\cdot)\cR(k,\eps)$ converges to $\Psi\cR(k,0)$.  The reader is likely ready to point out that all of these operators act in different spaces.  This is resolved by isometrically embedding the domains of the operators into a single Hilbert space by interpolation.

\subsection{Convergence of resolvents}\label{sec:resolventconvergence}

This subsection carries out the interpolation, then the convergence of $\cR(k,\eps)$ to $\cR(k,0)$, and finally $\Psi(\eps\cdot)\cR(k,\eps)$ to $\Psi\cR(k,0)$.

\subsubsection{Interpolation for operators in the ambient medium}\label{sec:interp}

The space $\ell^2(\eps\ZZ^d)$ can be isometrically embedded into $L^2(\RR^d)$ by interpolation through Fourier analysis according to the diagram
\begin{equation*}
  \ell^2(\eps\ZZ^d) \;\overset{\cU_\eps}{\twoheadrightarrow}\; L^2(\eps^{-1}Q') \;\overset{\iota}{\hookrightarrow}\; L^2(\RR^d) \,\overset{\cF^{-1}}{\longrightarrow}\, L^2(\RR^d),
\end{equation*}
as we explain presently.

Define the dual cells $Q=[-\sfrac{1}{2},\sfrac{1}{2}]^d$ and $Q'=[-\pi,\pi]^d$.
On $\ell^2(\eps\ZZ^d)$, we put the norm $\| f \|_{2,\eps}$ defined by
\begin{equation*}
  \| f \|_{2,\eps}^2 \;=\; \eps^d\! \sum_{x\in\eps\ZZ^d} |f(x)|^2.
\end{equation*}
The unitary discrete Fourier transform (Gelfand transform) $\cU_\eps$ is defined by
\begin{equation}\label{Gelfand}
  (\cU_\eps f)(\xi) \;=\; \frac{\eps^d}{(2\pi)^{d/2}}\!\sum_{x\in\eps\ZZ^d} f(x) e^{-ix\cdot\xi}
\end{equation}
for $\xi\in\eps^{-1}Q'$ and is inverted by
\begin{equation*}
  f(x) \;=\; \frac{1}{(2\pi)^{d/2}} \int\limits_{\eps^{-1}Q'} (\cU_\eps f)(\xi) e^{ix\cdot\xi}\,dV(\xi)
\end{equation*}
for $x\in\eps\ZZ^d$.  This formula is extended to an interpolation of $f$ by a real-analytic function on $\RR^d$ simply by allowing $x$ to reside in $\RR^d$. 
This amounts to first including $L^2(\eps^{-1}Q')$ into $L^2(\RR^d)$ through extension by zero ($\iota$ in the diagram above) and then applying the inverse Fourier transform
$\cF^{-1}:L^2(\RR^d)\to L^2(\RR^d)$:
\begin{equation}\label{Finv}
  \left(\cF^{-1}\!g\right)(x) \;=\; \frac{1}{(2\pi)^{d/2}} \int\limits_{\RR^d} g(\xi) e^{ix\cdot\xi}\,dV(\xi),
\end{equation}
to obtain the isometric interpolation operator
\begin{equation*}
  J_\eps \;:=\; \cF^{-1}\circ\iota\circ\cU_\eps : \ell^2(\eps\ZZ^d) \to L^2(\RR^d).
\end{equation*}
A basis of interpolating functions consists of the family $\cF^{-1}(\mathds{1}_{\eps^{-1}Q'})(x-\eps n)$ indexed by $n\in\ZZ^d$. 

Now consider the diagram
\begin{equation*}
  \ell^2(\eps\ZZ^d) \;\overset{J_\eps}{\hookrightarrow}\; L^2(\RR^d) \;\overset{\cF}{\longrightarrow}\; L^2(\RR^d)
\end{equation*}
and observe that $\cF\circ J_\eps = \iota\circ\cU_\eps$.  In other words, interpolating and then taking the Fourier transform is the same as taking the discrete Fourier transform $\cU_\eps$.

Under $\cU_\eps$, the operator $-\Delta_\eps$ is conjugated into multiplication by its symbol
\begin{equation*}
  \eps^{-2}\! \sum\limits_{j\in[1,d]} (2-2\cos\eps\xi_j)
\end{equation*}
in $L^2(\eps^{-1}Q')$.  When $V_\free(k,0)<0$, the (generalized) resolvent $\cR(k,\eps)=(-\Delta_\eps-V_\free(k,\eps))^{-1}$ has symbol
\begin{equation}
  R_\eps(\xi) \;=\; \Big( \eps^{-2}\! \sum\limits_{j\in[1,d]} (2-2\cos\eps\xi_j) - V_\free(k,\eps) \Big)^{-1}.
  \label{Rkeps_symbol}
\end{equation}
After applying the inclusion map $\iota:L^2(\eps^{-1}[-\pi,\pi]^d)\to L^2(\RR^d)$, we extend this symbol by zero (but we refrain from choosing an extension of the symbol of~$-\Delta_\eps$).  This amounts to extending the multiplication operator by zero on the orthogonal complement of $L^2(\eps^{-1}[-\pi,\pi]^d)$ in $L^2(\RR^d)$.

We denote the multiplication operator with this symbol by $\tilde\cR(k,\eps)$.  It acts in $L^2(\RR^d)$ and is related to $\cR(k,\eps)$ through an isometric transformation, namely interpolation by $J_\eps$ and then extension by zero on $\ran(J_\eps)^\perp$:
\begin{equation*}
  \tilde\cR(k,\eps) \;:=\; J_\eps \cR(k,\eps) J_\eps^{-1}P_{\ran(J_\eps)}.
\end{equation*}
in which $P_{\ran(J_\eps)}$ is the orthogonal projection onto $\ran(J_\eps)$ and $J_\eps^{-1}$ acts on $\ran(J_\eps)$.

\subsubsection{Convergence of resolvents for the ambient medium: $R$}

As we have just seen, by Fourier transform, the the operator $\cR(k,\eps)$ is converted to multiplication by its symbol
\begin{equation}
  R_\eps(\xi) \;=\; 
  \renewcommand{\arraystretch}{1.3}
\left\{
\begin{array}{ll}
     \bigg(\eps^{-2}\!\! \sum\limits_{j\in[1,d]} (2-2\cos\eps\xi_j) - V_\free(k,\eps)\bigg)^{-1}, & \xi\in\eps^{-1}Q', \\
     0, & \xi\not\in\eps^{-1}Q'.
\end{array}
\right.
\label{Repsilon_def}
\end{equation}
in which $\xi = (\xi_1,\dots,\xi_d)\in\RR^d$.  The limit of this function is the symbol $R_0(\xi)$ of the resolvent $\cR(k,0)$ of the differential operator $-\Delta = -\sum_{j=1}^d\partial^2/\partial x_j^2$ at spectral value $A'(k,0)$, that is
\begin{align}
  \cR(k,0) &\;=\; \left( -\Delta - V_\free(k,0) \right)^{-1},\label{R0_operator}\\
  R_0(\xi) &\;=\; \left( \,|\xi|^2 -  V_\free(k,0) \right)^{-1}.\label{R0_formula}
\end{align}
The convergence of resolvents in the operator norm is of order $O(\eps)$, that is, there exists $C$ such that for all sufficiently small~$\eps$,
\begin{equation*}
  \big\| \tilde\cR(k,\eps) - \cR(k,0) \big\| \;<\; C\eps.
\end{equation*}
This is equivalent to $O(\eps)$-uniform convergence of the corresponding symbols.

\begin{lemma}
\label{Reps_R0_diff}
  Let $k\in{\mathbb R}$ be such that $V_\free(k,0)=\lim_{\varepsilon\to0}V_\free(k, \varepsilon)$ is negative. There exist $\eps_0>0$ and $C>0$ such that for all  $\eps\in(0,\eps_0]$ and for all $\xi\in\RR^d$,
\begin{equation}
  |R_\eps(\xi)-R_0(\xi)| \;<\; C\eps.
\end{equation}
\end{lemma}

\begin{proof}
First, we prove the lemma for $d=1$.  Since $V_\free(k,0)<0$, there exists $a>0$ and $\varepsilon_0>0$ such that $0<a\le -V_\free(k,\eps)$ for all $\eps\in(0,\varepsilon_0]$.
Therefore, both functions 
\begin{equation}	
  R_\eps(\xi) \;=\; \left(\eps^{-2}(2-2\cos\eps\xi) -V_\free(k,\eps)\right)^{-1} \mathds{1}_{Q'}(\eps\xi),
  \qquad
  R_0(\xi) \;=\; \left(\xi^2 -V_\free(k,0)\right)^{-1}
  \label{R0_expr}
\end{equation}
are positive and attain their maximum values of $-V_\free(k,\eps)^{-1}$ and $-V_\free(k,0)$, respectively, at $\xi=0$. Observe that there exist $c>0$ and $x_0\in(0,\pi)$ such that
\begin{equation}
\label{Lbound}
  2-2\cos x \;\geq\; cx^2 \qquad \forall x : |x|\leq x_0,
\end{equation}
and that there is an analytic function $f(x)$ and a number $F$ such that
\begin{equation*}
  2-2\cos x \;=\; x^2 - x^4 f(x) \;\;\;\text{and}\;\;\; |f(x)|<F,
  \qquad \forall x : |x|\leq \pi.
\end{equation*}
Thus, for all $\xi\in\eps^{-1}[-\pi,\pi]$,
\begin{equation*}
  \eps^{-2}(2-2\cos\eps\xi) \;=\; \xi^2 - \eps^2\xi^4 f(\eps\xi)
\end{equation*}
and the following resolvent formulas hold:
\begin{align}
   R_\eps(\xi) - R_0(\xi) &\;=\; \eps^2\xi^4\,f(\eps\xi)\,R_0(\xi)R_\eps(\xi) \label{Rest1} \\[0.2em]
   &\;=\; \eps^2\xi^4\,f(\eps\xi)\,R_0(\xi)^2 \left[ 1 + \eps^{-2}\xi^4 f(\eps\xi)\,R_\eps(\xi) \right]. \label{Rest2}
\end{align}
In the region $|\xi|\leq\eps^{-1/4}$, the expression (\ref{Rest1}) yields
\begin{equation*}
\bigl|R_\eps(\xi)-R_0(\xi)\bigr| \,\leq\, \eps Fa^{-2}.
\end{equation*}
In the region $\eps^{-1/4}\leq|\xi|\leq\eps^{-1/2}$, one has $a+\xi^2>\eps^{-1/2}$, so that $R_0(\xi)<\eps^{-1/2}$, and $\eps^2\xi^4\leq1$.  Using these inequalities in (\ref{Rest2}) yields
\begin{equation*}
  |R_\eps(\xi)-R_0(\xi)| \,\leq\, \eps F(1+Fa^{-1}).
\end{equation*}
In the region $\eps^{-1/2}\leq|\xi|\leq\eps^{-1}x_0$, the lower bound (\ref{Lbound}) yields
$\eps^{-2}(2-2\cos\eps\xi) \,\geq\, c\eps^{-1}$, so that
\begin{equation*}
  0 \,\leq\, R_0(\xi) \,\leq\, R_\eps(\xi) \,\leq\, c^{-1}\eps,
\end{equation*}
which implies
\begin{equation*}
\bigl|R_\eps(\xi)-R_0(\xi)\bigr| \,\leq\, c^{-1}\eps.
\end{equation*}
In the region $\eps^{-1}x_0\leq|\xi|\leq\eps^{-1}\pi$, one has
\begin{equation*}
  \xi^2 \,\geq\, \eps^{-2}(2-2\cos\eps\xi) \,\geq\, \eps^{-2}cx_0^2 > 0,
\end{equation*}
and therefore
\begin{equation}
  0 \,\leq\, R_0(\xi) \,\leq\, R_\eps(\xi) \,\leq\, \eps^2c^{-1}x_0^{-2},
\label{R_est_eps_squared}
\end{equation}
which yields
\begin{equation*}
\bigl|R_\eps(\xi)-R_0(\xi)\bigr| \,\leq\, \eps^2c^{-1}x_0^{-2}.
\end{equation*}
In the region $|\xi|\geq\eps^{-1}\pi$,
\begin{equation*}
   \bigl|R_\eps(\xi)-R_0(\xi)\bigr| \;=\; R_0(\xi) \;\leq\; \eps^2\pi^{-2}.
\end{equation*}
The lemma is now proved for the case $d=1$, which can now be used in a telescoping sum to obtain the lemma for general values of~$d$.
For $\xi\in{\mathbb R}^d$,
\begin{equation}
	\begin{aligned}
 \bigl|R_\eps(\xi)&- R_0(\xi)\bigr| \\
  &\leq \Bigg| \Big( \eps^{-2} \sum_{j\in[1,d]} (2-2\cos\eps\xi_j) - V_\free(k,\eps) \Big)^{-1}
                    -  \Big( \eps^{-2} \sum_{j\in[1,d]} (2-2\cos\eps\xi_j) - V_\free(k,0) \Big)^{-1} \Bigg|+
                    \\
  &+\sum_{\ell\in[1,d]}  
        \Bigg| \Big( \eps^{-2} \sum_{j<\ell} \xi_j^2 + \sum_{j\geq\ell} (2-2\cos\eps\xi_j) - V_\free(k,0) \Big)^{-1}\\
        &\hspace{4cm}
              - \Big( \eps^{-2} \sum_{j\leq\ell} \xi_j^2 + \sum_{j>\ell} (2-2\cos\eps\xi_j) - V_\free(k,0) \Big)^{-1} \Bigg|\\
   &\hspace{6cm}\le C_1\eps+d\max\{Fa^{-2}, F(1+Fa^{-1}), c^{-1}\}\varepsilon \;=\; C\eps,
   \end{aligned}
   \label{diff_bound}
\end{equation}
where $C:=C_1+d\max\{Fa^{-2}, F(1+Fa^{-1}), c^{-1}\}$, $\varepsilon\in(0,\varepsilon_0]$ for suitable $\varepsilon_0>0$.
Considering that all of the terms of all of the sums in \eqref{diff_bound} are nonnegative, the $C_1\eps$ bound for the first part comes from the analyticity of $A'(k,\eps)$ at $\eps=0$ and for the $\ell^\text{th}$ difference in the second part, the $\max\{Fa^{-2}, F(1+Fa^{-1}), c^{-1}\}\varepsilon$ bound comes from the lemma for $d=1$ with $\xi_\ell$ in place of~$\xi$.
\end{proof}

\subsubsection{Convergence of $\Psi R$}
\label{res_est_sec}

We will use the same notation ${\mathcal U}_\varepsilon$ for the operator that maps 
$f\in\ell^2(\varepsilon{\mathbb Z}^d)$ to the $L^2({\mathbb R}^d)$ function obtained by extending ${\mathcal U}_\eps f$ by zero in ${\mathbb R}^d\setminus\varepsilon^{-1}Q$.  In other words, we write ${\mathcal U}_\eps$ in place of $\iota\circ {\mathcal U}_\eps$.
We will make use of the mapping $\widetilde{\mathcal U}_\varepsilon$ that takes $f\in\ell^2(\varepsilon{\mathbb Z}^d)$ to the $\varepsilon^{-1}Q'$-periodic extension of~${\mathcal U}_\varepsilon f$,
\begin{equation}
\widetilde{\mathcal U}_\varepsilon f=\sum_{m\in{\mathbb Z}^d}({\mathcal U}_\varepsilon f)(\cdot+2\pi m\,\eps^{-1}).
\label{tilde_no_tilde}
\end{equation}
The Fourier transform ${\mathcal F}$ on $L^2({\mathbb R}^d)$ corresponding to (\ref{Finv})~is
\[
({\mathcal F}f)(\xi)=\frac{1}{(2\pi)^{d/2}}\int_{{\mathbb R}^d}f(x){\rm e}^{-{\rm i}x\cdot\xi}dV(x),\qquad \xi\in{\mathbb R}^d.
\]
This definition can be extended to tempered distributions in the usual way.  Particularly, for $z\in{\mathbb R}^d$ one has 
\[
{\mathcal F}\bigl[\delta(\cdot-z)\bigr](\xi)=(2\pi)^{-{d/2}}{\rm e}^{-{\rm i}z\cdot\xi},\qquad \xi\in{\mathbb R}^d.
\]  

For every $f\in\ell^2(\varepsilon{\mathbb Z}^d)$, 
consider the $\varepsilon$-periodic moderate-growth distribution  
\[
\tilde{f}=\sum_{z\in\varepsilon{\mathbb Z}^d}f(z)\delta(\cdot-z).
\]
Note that for all $f\in\ell^2(\varepsilon{\mathbb Z}^d)$ one has 
\begin{equation}
(\widetilde{\mathcal U}_\varepsilon f)(\xi)=\biggl(\frac{\varepsilon^2}{2\pi}\biggr)^{d/2}\sum_{z\in\varepsilon{\mathbb Z}^d}f(z){\rm e}^{-{\rm i}z\cdot\xi}=\varepsilon^d({\mathcal F}\tilde{f})(\xi),\qquad \xi\in{\mathbb R}^d.
\label{UFlink_f}
\end{equation}

Furthermore, fix $\Psi\in C_0^\infty({\mathbb R}^d)$ that is nonzero in $\Omega$ and vanishes in $\RR^d\setminus\Omega$. The Fourier transform of $\Psi$ satisfies the standard decay estimate: for each $l>0$, there is $c_l>0$ such that
\begin{equation}
\bigl|({\mathcal F}\Psi)(\xi)\bigr|\le\frac{c_l}{(|\xi|^2+1)^{l/2}},\qquad \xi\in{\mathbb R}^d.
\label{Fourier_estimate}
\end{equation}
Recall that, by the Poisson summation formula, one has 
\begin{equation}\label{PS}
\varepsilon^d\sum_{y\in\varepsilon{\mathbb Z}^d}\Psi(y){\rm e}^{-{\rm i}y\cdot\xi}=(2\pi)^{d/2}\sum_{m\in{\mathbb Z}^d}{\mathcal F}\bigl[\Psi(x){\rm e}^{-{\rm i}x\cdot\xi}\bigr](2\pi m/\varepsilon)\equiv(2\pi)^{d/2}\sum_{m\in{\mathbb Z}^d}({\mathcal F}\Psi)(\xi+2\pi m/\varepsilon).
\end{equation}
It follows that
\begin{equation}
({\mathcal U}_\varepsilon\Psi)(\xi)=\sum_{m\in{\mathbb Z}^d}({\mathcal F}\Psi)(\xi+2\pi m/\varepsilon)=({\mathcal F}\Psi)(\xi)+{\mathcal E}_\varepsilon(\xi),\qquad
\xi\in{\mathbb R}^d,
\label{UFlink_Psi}
\end{equation}
where, by virtue of (\ref{Fourier_estimate}) and noting that $|\xi_j|\le\pi/\varepsilon$, $j=1,\dots,d$, one has, for $\xi\in\varepsilon^{-1}Q'$ and $l>d$, 
\begin{equation}
\begin{aligned}
\bigl\vert{\mathcal E}_\varepsilon(\xi)\bigr\vert&\le\sum_{m\in{\mathbb Z}^d\setminus\{0\}}\frac{c_l}{(|\xi+2\pi m/\varepsilon|^2+1)^{l/2}}<
c_l\biggl(\frac{\varepsilon}{\pi}\biggr)^l\sum_{m\in{\mathbb Z}^d\setminus\{0\}}(2|m|-1)^{-l}<\tilde{c}_l\varepsilon^l, 
\\[0.3em]
\tilde{c}_l&:=c_l\pi^{-l}\sum_{m\in{\mathbb Z}^d\setminus\{0\}}|m|^{-l}.
\end{aligned}
\label{E1_est}
\end{equation}

\begin{lemma}
\label{Fourier_conv}
For $\Psi\in L^2({\mathbb R}^d)$ and $f\in\ell^2(\varepsilon{\mathbb Z}^d)$ one has 
\begin{equation*}
\widetilde{\mathcal U}_\varepsilon[\Psi f]=(2\pi)^{-d/2}\widetilde{\mathcal U}_\varepsilon\Psi*{\mathcal U}_\varepsilon f.
\end{equation*}

\end{lemma}

\begin{proof}
Using the identities (\ref{UFlink_f}), (\ref{tilde_no_tilde}), we obtain
\begin{equation*}
\begin{aligned}
\widetilde{\mathcal U}_\varepsilon[\Psi f]&=\varepsilon^d{\mathcal F}[\Psi\tilde{f}]=\biggl(\frac{\varepsilon^2}{2\pi}\biggr)^{\!\!d/2}{\mathcal F}\Psi*{\mathcal F}\tilde{f}=(2\pi)^{-d/2}{\mathcal F}\Psi*\widetilde{\mathcal U}_\varepsilon f=(2\pi)^{-d/2}{\mathcal F}\Psi*\biggl(\sum_{m\in{\mathbb Z}^d}({\mathcal U}_\varepsilon f)(\cdot+2\pi m/\varepsilon)\biggr)\\[0.3em]
&=(2\pi)^{-d/2}\biggl(\sum_{m\in{\mathbb Z}^d}({\mathcal F}\Psi)(\cdot+2\pi m/\varepsilon)\biggr)*{\mathcal U}_\varepsilon f=(2\pi)^{-d/2}\widetilde{\mathcal U}_\varepsilon\Psi*{\mathcal U}_\varepsilon f,
\end{aligned}
\end{equation*}
as claimed.
\end{proof}

\subsection{Asymptotics of resolvents and convergence of $\mu$ as an eigenvalue}

Within this section, $F$ is an arbitrary function in $L^2({\mathbb R}^d)$. For each $\varepsilon$, the transform ${\mathcal U}_\varepsilon$ maps $f\in L^2(\varepsilon{\mathbb Z}^d)$ to functions in $L^2(\varepsilon^{-1}Q')$, and $\widetilde{\mathcal U}_\varepsilon$ is its $\varepsilon^{-1}Q'$-periodic extension, see (\ref{tilde_no_tilde}).  Consider the operator $A_\varepsilon$ on $L^2({\mathbb R}^d)$ defined in the following way: 
\begin{equation}
(A_\varepsilon F)(\xi)=\mathds{1}_{\varepsilon^{-1}Q'}(\xi)\int_{\varepsilon^{-1}Q'}(\widetilde{\mathcal U}_\varepsilon\Psi)(\xi-\eta)R_\varepsilon(\eta)F(\eta)dV(\eta),\qquad \xi\in{\mathbb R}^d.
\label{Aeps_def}
\end{equation}

Applying the operator $\widetilde{\mathcal U}_\varepsilon$ to both sides of (\ref{BS7}), in view of the discussion at the end of Section \ref{sec:interp}, yields 
\[
\widetilde{\mathcal U}_\varepsilon\bigl[\Psi\,\cR(k,\eps)\Psi\bar u\bigr]=(3\mu)^{-1}\widetilde{\mathcal U}_\varepsilon[\Psi\bar u].
\]
Then, using Lemma \ref{Fourier_conv}, we obtain
\[
(2\pi)^{-d/2}\widetilde{\mathcal U}_\varepsilon\Psi*{\mathcal U}_\varepsilon\bigl[{\cR}(k,\eps)\Psi\bar u\bigr]=(3\mu)^{-1}\widetilde{\mathcal U}_\varepsilon[\Psi\bar u],
\]
and hence, in view of the definition of $R_\varepsilon$ (see (\ref{Repsilon_def})) and the symbol of the operator ${\mathcal R}(k, \varepsilon)$ (see (\ref{Rkeps_symbol})), we have 
\begin{equation}
(2\pi)^{-d/2}\mathds{1}_{\varepsilon^{-1}Q'}(\xi)\,\widetilde{\mathcal U}_\varepsilon\Psi*R_\eps\,{\mathcal U}_\varepsilon[\Psi\bar u]
=(3\mu)^{-1}{\mathcal U}_\varepsilon[\Psi\bar u].
\label{transform_eq}
\end{equation}
The equality (\ref{transform_eq}) can be written as 
\begin{equation}
(2\pi)^{-d/2}A_\varepsilon\,{\mathcal U}_\varepsilon[\Psi\bar u]=(3\mu)^{-1}{\mathcal U}_\varepsilon[\Psi\bar u].
\label{whole_space_eig}
\end{equation}

Note that, by virtue of the first equality in (\ref{UFlink_Psi}), one has (for all $F\in L^2({\mathbb R}^d)$) 
\begin{align*}
A_\varepsilon F&=\mathds{1}_{\varepsilon^{-1}Q'}\int_{\varepsilon^{-1}Q'}\sum_{m\in{\mathbb Z}^d}({\mathcal F}\Psi)(\cdot-2\pi\varepsilon^{-1}m-\eta)R_\varepsilon(\eta)F(\eta)dV(\eta)\\[0.3em]
&\hspace{2cm}=\mathds{1}_{\varepsilon^{-1}Q'}\int_{\varepsilon^{-1}Q'}({\mathcal F}\Psi)(\cdot-\eta)R_\varepsilon(\eta)F(\eta)dV(\eta)\\[0.3cm]
&\hspace{3cm}+\mathds{1}_{\varepsilon^{-1}Q'}\int_{\varepsilon^{-1}Q'}\biggl(\sum_{m\in{\mathbb Z}^d\setminus\{0\}}({\mathcal F}\Psi)(\cdot-2\pi\varepsilon^{-1}m-\eta)\biggr)R_\varepsilon(\eta)F(\eta)dV(\eta).
\end{align*}
In view of Lemma \ref{Reps_R0_diff}, as $\varepsilon\to0$ one has the following convergence on the sense of $L^2({\mathbb R}^d):$ 
\begin{equation}
\begin{aligned}
\mathds{1}_{\varepsilon^{-1}Q'}\int_{\varepsilon^{-1}Q'}({\mathcal F}\Psi)(\cdot-\eta)R_\varepsilon(\eta)F(\eta)dV(\eta)\to\int_{{\mathbb R}^d}({\mathcal F}\Psi)(\cdot-\eta)R_0(\eta)F(\eta)dV(\eta)
=:A_0F,
\end{aligned}
\label{A0_def}
\end{equation}
with an $O(\varepsilon)$ convergence error that is uniform with respect to $F\in L^2({\mathbb R}^d)$.
Furthermore, 
\begin{equation} 
\begin{aligned}
\int_{\varepsilon^{-1}Q'}&\biggl(\sum_{m\in{\mathbb Z}^d\setminus\{0\}}({\mathcal F}\Psi)(\cdot-2\pi\varepsilon^{-1}m-\eta)\biggr)R_\varepsilon(\eta)F(\eta)dV(\eta)\\[0.3em]
&=\int_{\varepsilon^{-1}Q'/2}\biggl(\sum_{m\in{\mathbb Z}^d\setminus\{0\}}({\mathcal F}\Psi)(\cdot-2\pi\varepsilon^{-1}m-\eta)\biggr)R_\varepsilon(\eta)F(\eta)dV(\eta)\\[0.3em]
&\hspace{1.2cm}+\int_{\varepsilon^{-1}(Q'\setminus(Q'/2))}\biggl(\sum_{m\in{\mathbb Z}^d\setminus\{0\}}({\mathcal F}\Psi)(\cdot-2\pi\varepsilon^{-1}m-\eta)\biggr)R_\varepsilon(\eta)F(\eta)dV(\eta).
\end{aligned}
\label{sumof2}
\end{equation}
The sum under the first integral in (\ref{sumof2}) converges to zero as $\varepsilon\to0$, since the argument in each its term is bounded below as follows.
On the one hand, for $\xi\in\varepsilon^{-1}Q$,
\begin{align*}
\bigl|\xi-2\pi\varepsilon^{-1}m-\eta\bigr|&\ge \sqrt{d}\min_j\bigl\{|2\pi\varepsilon^{-1}m_j-(\xi_j-\eta_j)|\bigr\}\\[0.3em]
&\ge\sqrt{d}\min_j\bigl\{2\pi\varepsilon^{-1}|m_j|-(|\xi_j|+|\eta_j|)|\bigr\} \ge\sqrt{d}\varepsilon^{-1}(2\pi-3\pi/2)=\varepsilon^{-1}\pi\sqrt{d}/2,
\end{align*}
On the other hand, again for $\xi\in\varepsilon^{-1}Q'$,
\begin{align*}
\bigl|\xi-2\pi\varepsilon^{-1}m-\eta\bigr|\ge2\pi\varepsilon^{-1}|m|-|\xi-\eta|\ge \varepsilon^{-1}(2\pi|m|-3\sqrt{d}\pi/2),
\end{align*}
since for $\xi\in\varepsilon^{-1}Q'$, $\eta\in\varepsilon(Q'\setminus(Q'/2))$ one has $|\xi-\eta|\le\varepsilon^{-1}\sqrt{d(\pi+\pi/2)^2}=\varepsilon^{-1}3\sqrt{d}\pi/2$. Now, for the sum in question we have (cf. (\ref{E1_est}))
\begin{align}
&\bigg|\sum_{m\in{\mathbb Z}^d\setminus\{0\}}({\mathcal F}\Psi)(\cdot-2\pi\varepsilon^{-1}m-\eta)\biggr|\nonumber\\[0.3em]
&\le\biggl|\sum_{m\in{\mathbb Z}^d: 0<|m|<\sqrt{d}}({\mathcal F}\Psi)(\cdot-2\pi\varepsilon^{-1}m-\eta)\biggr|+\biggl|\sum_{m\in{\mathbb Z}^d: |m|\ge\sqrt{d}}({\mathcal F}\Psi)(\cdot-2\pi\varepsilon^{-1}m-\eta)\biggr|\label{2abs_values}\\[0.3em]
&\hspace{0.4cm}\le \sum_{m\in{\mathbb Z}^d: 0<|m|<\sqrt{d}}\frac{c_l}{(|\varepsilon^{-1}\pi\sqrt{d}/2|^2+1)^{l/2}}+
\sum_{m\in{{\mathbb Z}^d:|m|\ge\sqrt{d}}}\frac{c_l}{(|\varepsilon^{-1}(2\pi|m|-3\sqrt{d}\pi/2)|^2+1)^{l/2}}.
\label{final_line_2terms}
\end{align}
(Note that the first term in (\ref{2abs_values}) vanishes in the case $d=1$, and so the first term in the parentheses in (\ref{final_line_2terms}) is absent.) 

We estimate the second integral in (\ref{sumof2}) by using Young's convolution inequality, as follows:
\begin{align*}
 \biggl\Vert\int_{\varepsilon^{-1}(Q'\setminus(Q'/2))}&\biggl(\sum_{m\in{\mathbb Z}^d\setminus\{0\}}
 {\mathcal F}\Psi(\cdot-2\pi\varepsilon^{-1}m-\eta)\biggr)R_\varepsilon(\eta)F(\eta)dV(\eta)\biggr\Vert_{L^2(\varepsilon^{-1}Q')}
 \\[0.3em]
 &\le\biggl\Vert\sum_{m\in{\mathbb Z}^d\setminus\{0\}}
 {\mathcal F}\Psi(\cdot-2\pi\varepsilon^{-1}m)\biggr\Vert_{L^1(\varepsilon^{-1}Q')}\Vert R_\varepsilon F\Vert_{L^2(\varepsilon^{-1}(Q'\setminus(Q'/2)))} 
 \le C\varepsilon^2\Vert F\Vert_{L^2({\mathbb R}^d)}, 
 \end{align*}
 since $|R_\varepsilon(\eta)|\le C\varepsilon^2$ for $\eta\in \varepsilon^{-1}(Q'\setminus(Q'/2))$ by virtue of the estimate (\ref{R_est_eps_squared}) utilised with $x_0=\pi/2$.

We have thus proved the following result. 
\begin{theorem}
\label{convergence_theorem}
There exist $\varepsilon_0, C>0$ such that  
\[
\bigl\Vert A_\varepsilon-A_0\bigr\Vert_{L^2({\mathbb R}^d)\to L^2({\mathbb R}^d)}\le C\varepsilon\qquad \forall\varepsilon\in(0,\varepsilon_0].
\]
where the operators $A_\varepsilon$, $A_0$ are defined by (\ref{Aeps_def}), (\ref{A0_def}), respectively.
\end{theorem}
\begin{corollary}
As $\varepsilon\to0$, the limit of every convergent sequence of eigenvalues of the problem (\ref{whole_space_eig}) is an  eigenvalue $(3\mu)^{-1}$ and 
(the Fourier transform of the corresponding eigenspace of the problem)
\begin{equation}
\Psi{\mathcal R}(k,0)[\Psi u]=(3\mu)^{-1}\Psi u
\label{eff_eig} 
\end{equation} 
The corresponding eigenspaces converge (in an appropriate sense) to the Fourier transform of the eigenspace for (\ref{eff_eig}).  

Conversely, for an eigenvalue  $(3\mu)^{-1}$ of (\ref{eff_eig}), there is a sequence of eigenvalues of (\ref{whole_space_eig}) indexed by $\varepsilon$ that converges to it as $\varepsilon\to0$, a statement regarding the convergence of the corresponding eigenspaces holds. 
 \end{corollary}  
 \begin{proof}
 Passing to the limit as $\varepsilon\to0$ in (\ref{whole_space_eig}), we obtain
 \begin{equation}
(2\pi)^{-d/2}A_0\,{\mathcal F}[\Psi u]=(3\mu)^{-1}{\mathcal F}[\Psi u].
\label{Fourier_eff}
 \end{equation}
Recall the definition of the operator $A_0$ as the multiplication by $R_0$ followed by convolution with ${\mathcal F}\Psi$. Taking the inverse Fourier transform now yields the claim. 
\end{proof}

\section{Spectrally embedded defect states in bilayer structures}\label{sec:BIC} 

The main promise of this work is carried out in this section.  We follow the strategy described in Sections~\ref{sec:defect_eigenvalues} and~\ref{sec:heuristics1} for constructing defect states, using $\ell\!=\!1$ in the computations.

\subsection{Construction of defect states}

Let us take the defect function $\Psi$ to be a positive cutoff function with maximal value $1$ and tapering to $0$ at the boundary of the continuum defect region $\Omega$.  Denote by $\alpha=\alpha_{\rm f}$ the Robin parameter in the ambient medium and $\alpha=\alpha_{\rm d}$ for the maximal value of the Robin parameter in the defect region.  Thus the strength of the added potential in the defect region~is
\[
\mu\;=\;\alpha_{\rm d}-\alpha_{\rm f}.
\]
The problem (\ref{eff_eig}) now reads
\begin{equation}
\Psi\left( -\Delta - V_\free(k,0) \right)^{-1}[\Psi u]=\frac{1}{3(\alpha_{\rm f}-\alpha_{\rm d})}\Psi u,
\label{transformed_effective}
\end{equation}
or, in the Fourier representation (\ref{Fourier_eff}),
\[
(2\pi)^{-d/2}A_0{\mathcal F}[\Psi u]=\frac{1}{3(\alpha_{\rm f}-\alpha_{\rm d})}{\mathcal F}[\Psi u].
\]

Recall from (\ref{A0_def}) that 
\[
A_0F=\int_{{\mathbb R}^d}({\mathcal F}\Psi)(\cdot-\eta)R_0(\eta)F(\eta)dV(\eta),\qquad F\in L^2({\mathbb R}^d),
\]
where (see (\ref{R0_formula}))
\[
R_0(\xi)=\left(|\xi|^2 -  V_\free(k,0)\right)^{-1}.
\]
Theorem \ref{convergence_theorem} justifies seeking eigenvalues of the operator $A_0$ for values of $k$ such that $V_\free(k,0)<0$. Each of these that has an eigenfunction of the form ${\mathcal F}[\Psi u]$, $u\in L^2({\mathbb R}^d)$ for some $u\in L^2({\mathbb R}^d)$ is also an eigenvalue (for the same value of $k$) of the operator (cf. (\ref{transformed_effective})) $\Psi\left( -\Delta - V_\free(k,0) \right)^{-1}$ --- the function $\Psi u$ is then an eigenfunction of the latter. 
The opposite also holds: every eigenvalue of $\Psi\left( -\Delta - V_\free(k,0) \right)^{-1}$, whose eigenfunctions necessarily have the form $\Psi u$ for some $u\in L^2({\mathbb R}^d)$ are eigenvalues of $A_0$ whose eigenspace contains the function ${\mathcal F}[\Psi u]$.  The eigenvalue $\lambda$ then generates a value of the parameter $\alpha_d$.

\subsection{Overlapping bands and gaps}
\label{decoupling_sec}

We seek values $k$ such that $V_\mathrm{f,D}(k,0)<0$ and $V_\mathrm{f,N}(k,0)>0$, and vice versa. To this end, notice first that the values $k$ at which $V_\mathrm{f,D}(k,0)$ is singular are those for which $\sin k=0$, i.e. $k=n\pi$, $n\in{\mathbb Z}$, and the values $k$ at which $V_\mathrm{f,N}(k,0)$ is singular are those for which $\cos k=0$ or $\cos(k/2)=0$ i.e. $k=\pi/2+n\pi$, $n\in{\mathbb Z}$, or $k=2\pi m+\pi$, $n\in{\mathbb Z}$.

We start by investigating the behaviour of $V_\mathrm{f,D}(k,0)$ and $V_\mathrm{f,N}(k,0)$ in the vicinity of the singularities of  $V_\mathrm{f,D}(k,0)$ given by $2\pi n$, $n\in{\mathbb Z}$. For $n\in{\mathbb Z}$ and sufficiently small $\sigma$, one has 
\begin{align*}
V_\mathrm{f,D}(2\pi n+\sigma,0)/3&=2d(2\pi n+\sigma)\tan(\pi n+\sigma/2)
-l(2\pi n+\sigma)\cot(2\pi n+\sigma)
+3d(2\pi n+\sigma)^2-\alpha\\[0.3em]
&=2d(2\pi n+\sigma)
\tan(\sigma/2)
-l(2\pi n+\sigma)\cot\sigma
+3d(2\pi n+\sigma)^2-\alpha,
\end{align*}
which is negative as $\sigma\searrow0$. 
At the same time,
\[
V_\mathrm{f,N}(2\pi n+\sigma,0)/3=2d(2\pi n+\sigma)\tan(\sigma/2)
+l(2\pi n+\sigma)\tan\sigma
+3d(2\pi n+\sigma)^2-\alpha=12\pi^2dn^2-\alpha+O(\sigma).
\]  
Therefore, for $n\in{\mathbb Z}$ such that $12\pi^2dn^2>\alpha$, a Neumann band overlaps a Dirichlet gap in the vicinity of points $2\pi n+\sigma$ with such $n$.

At the values $k$ that are singularities for both $V_\mathrm{f,D}(k,0)$ and $V_\mathrm{f,N}(k,0)$, i.e. $k=(2n+1)\pi+\sigma$, $n\in{\mathbb Z}$, we have, one the one hand,
\begin{align*}
V_\mathrm{f,D}((2n+1)\pi+\sigma,0)/3&=2d(\pi(2n+1)+\sigma)\tan(\pi n+\pi/2+\sigma/2)\\[0.2em]
&-l((2n+1)\pi+\sigma)\cot((2n+1)\pi+\sigma)
+3d((2n+1)\pi+\sigma)^2-\alpha\\[0.3em]
&=2d((2n+1)\pi+\sigma)
\cot(\sigma/2)
-l(\pi(2n+1)+\sigma)\cot\sigma
+3d((2n+1)\pi+\sigma)^2-\alpha,
\end{align*}
which is negative as $\sigma\searrow0$. On the other hand,  
\begin{align*}
V_\mathrm{f,N}((2n+1)\pi+\sigma,0)/3&=2d((2n+1)\pi+\sigma)\cot(\sigma/2)\\[0.2em]
&+l((2n+1)\pi+\sigma)\tan\sigma
+3d((2n+1)\pi+\sigma)^2-\alpha,
\end{align*}
which is negative as $\sigma\searrow0$. 

Finally, we consider the remaining singularities of $V_\mathrm{f,N}(k,0)$.  First, for $n\in{\mathbb Z}$ and sufficiently small $\sigma$, one has
\begin{align*}
V_\mathrm{f,N}(\pm\pi/2+2\pi n+\sigma,0)/3&=2d(\pm\pi/2+2\pi n+\sigma)\tan(\pm\pi/4+\pi n+\sigma/2)\\[0.3em]
&+l(\pm\pi/2+2\pi n+\sigma)\frac{1+O(\sigma)}{-\sigma+O(\sigma^2)}
+3d(\pm\pi/2+2\pi n+\sigma)^2-\alpha,
\end{align*}
which is negative as $\sigma\searrow0$ for $n=1,2,3,\dots$. Near such values of $k$, we also have 
\begin{align*}
V_\mathrm{f,D}(\pm\pi/2+2\pi n+\sigma,0)/3&=2d(\pm\pi/2+2\pi n+\sigma)\tan(\pm\pi/4+\pi n+\sigma/2)
\\[0.3em]
&-l(\pm\pi/2+2\pi n+\sigma)\frac{-\sigma+O(\sigma^2)}{1+O(\sigma)}+3d(\pm\pi/2+2\pi n+\sigma)^2-\alpha\\[0.3em]
&=d(\pm\pi/2+2\pi n)(2\pm3\pi/2+6\pi n)-\alpha+O(\sigma),
\end{align*}
which is positive for sufficiently large $n$. Therefore, a Dirichlet band overlaps a Neumann gap in the vicinity of points $\pm\pi/2+2\pi n$ for such $n$.

 \subsection{Construction of spectrally embedded defect states in 1D}

In the example below, we set $d=1$ and $\alpha_{\rm f}=5$. The functions $V_\mathrm{f,D}(k,0)$, $V_\mathrm{f,N}(k,0)$ are plotted in Fig.\,\ref{fig:Vplot_1D}.

\begin{figure}[!hbp] 
\centerline{
\scalebox{0.45}{\includegraphics{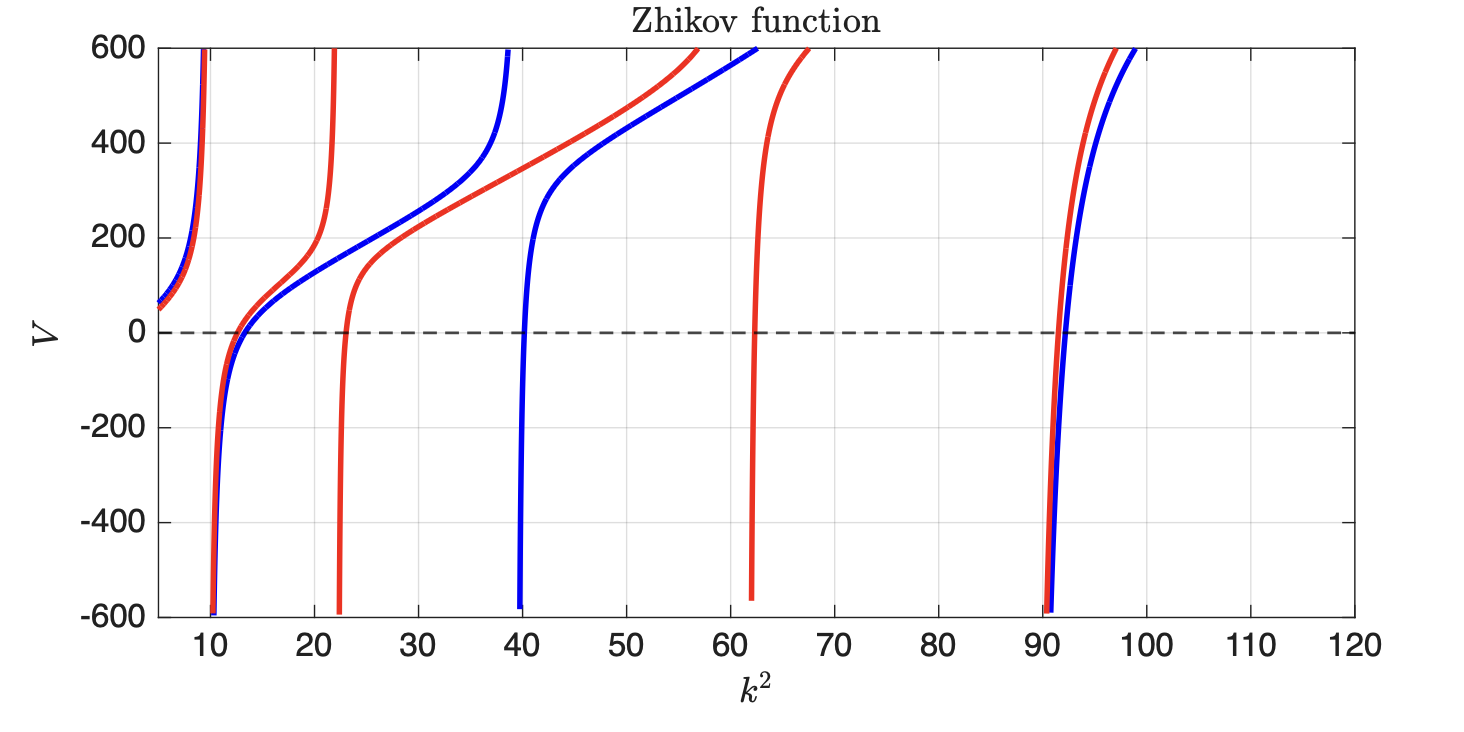}}
}
\caption{\small Blue: Dirichlet, red: Neumann.}
\label{fig:Vplot_1D}
\end{figure}

For $k=6.336$, one has $V_\mathrm{f,D}=-8.964$, $V_\mathrm{f,N}=348.386$, so the value $k^2=40.151$ is in a gap of the problem for the monolayer with Dirichlet decoration and within the spectrum of the problem for the monolayer with Neumann decoration. The first three eigenfunctions of the limit problem for the monolayer with Dirichlet boundary conditions and the defect   
\[
\Psi(x)=\left\{\begin{array}{ll}(1-x^2)^{2.4},& |x|\le1,\\[0.2em]
0,& |x|>1,\end{array}\right.
\]
are shown in Fig.\,\ref{Limit_problem_eigenfunctions}. Of the three corresponding eigenvalues, the first does not exceed $3\alpha_{\rm f}=15$.

\begin{figure}[!hbp] 
\centerline{
\scalebox{0.3}{\includegraphics{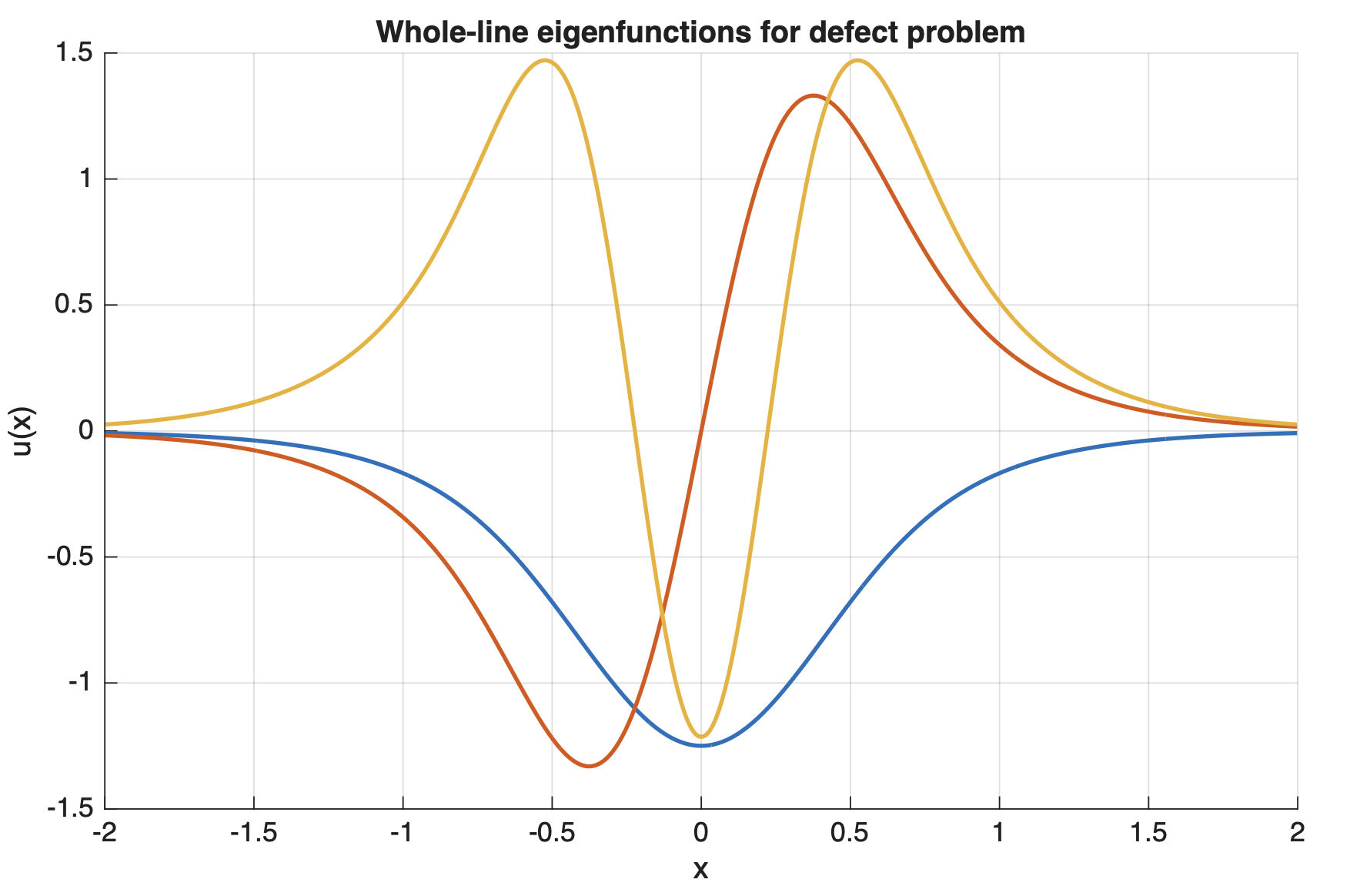}}
}
\caption{\small The corresponding eigenvalues are 14.145, 31.754, 59.646.}
\label{Limit_problem_eigenfunctions}
\end{figure}

The results of direct computation, for $\varepsilon=0.1$ of a defect eigenfunction for the original Dirichlet monolayer with eigenvalue located near 40.151 is provided in Fig.\,\ref{fig:reconstruction}.

 \begin{figure}[!hbp] 
\centerline{
\scalebox{0.23}{\includegraphics{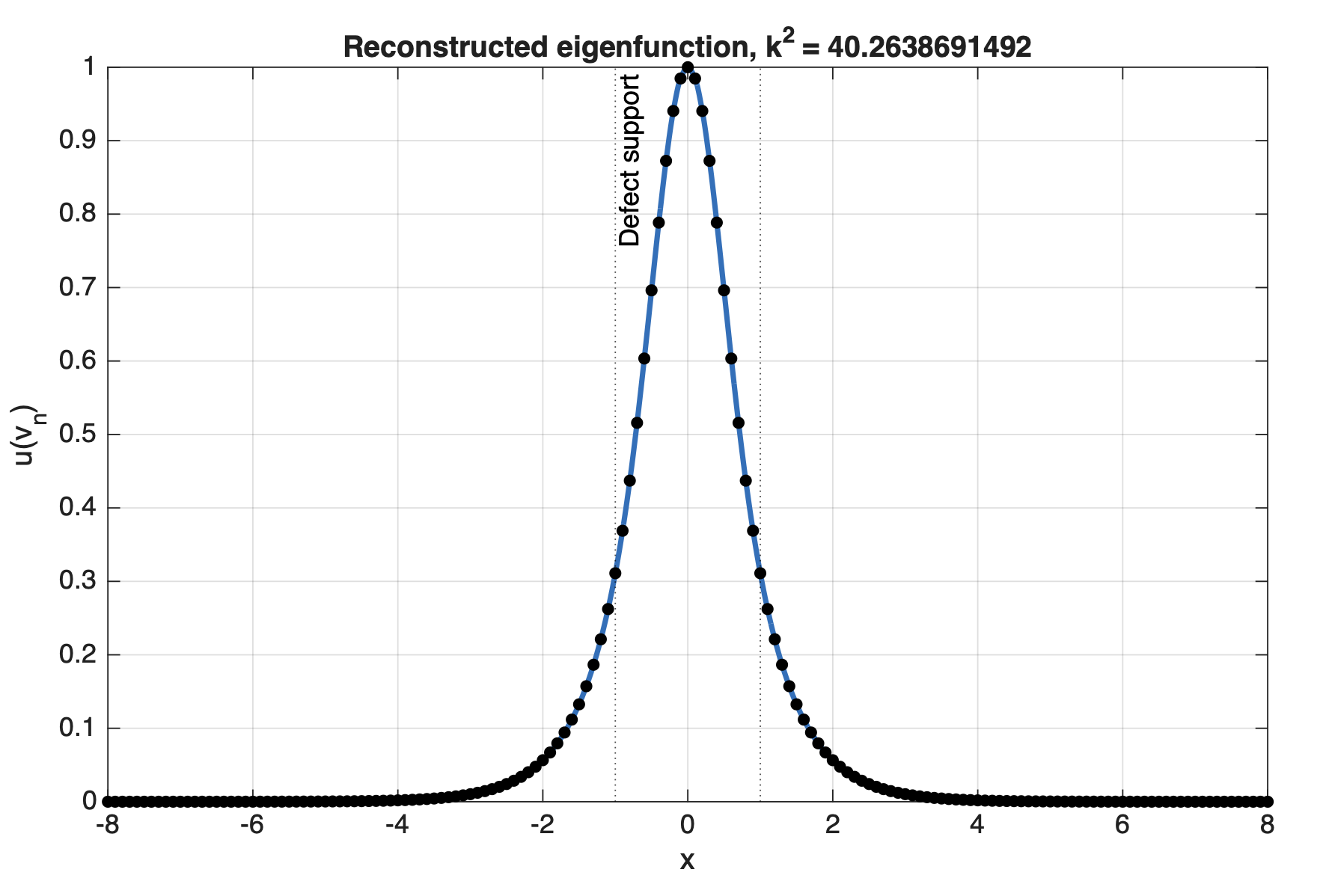}}
\hspace{0.4em}
\scalebox{0.23}{\includegraphics{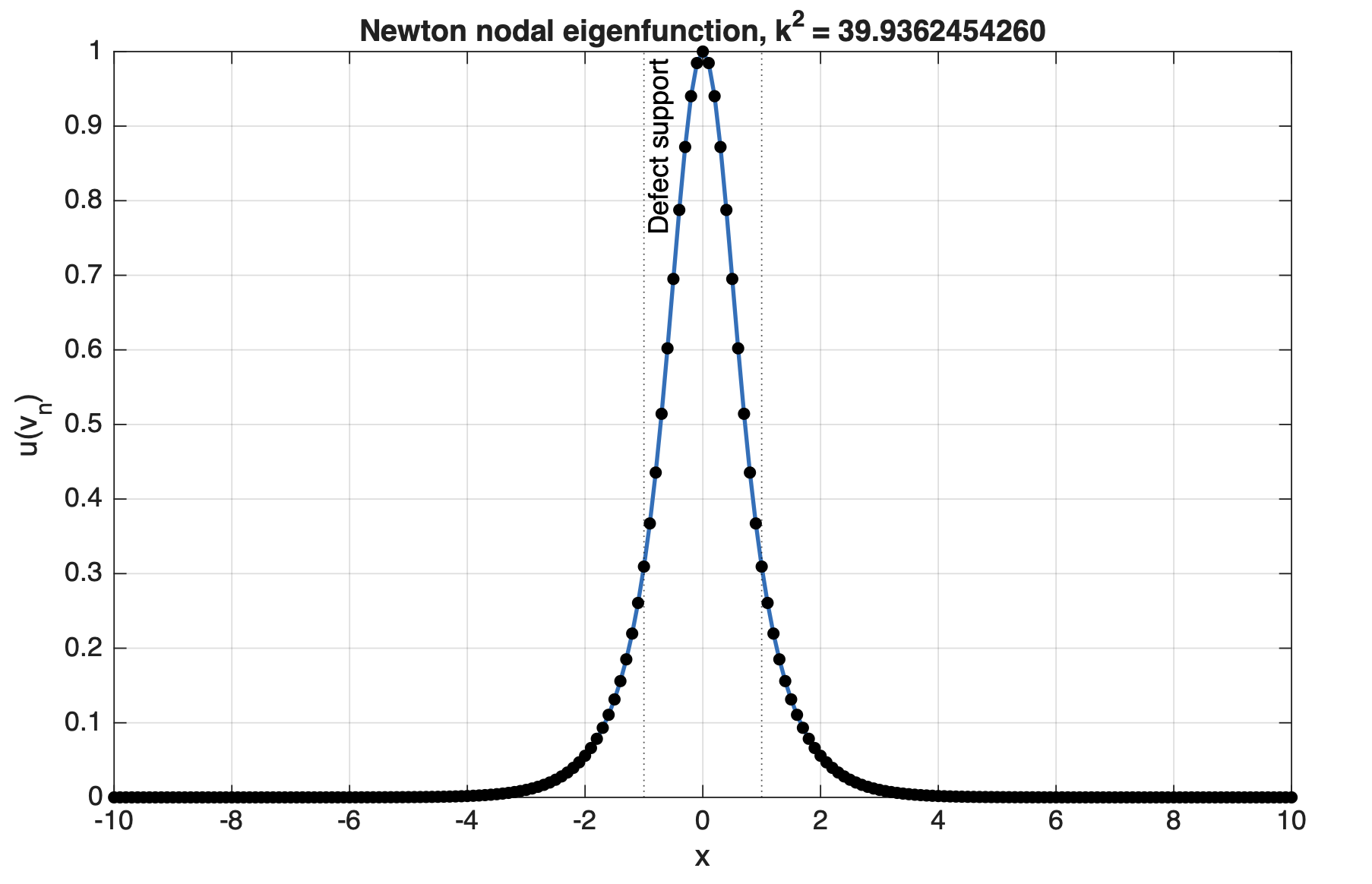}}
}
\caption{\small Left: Original graph via Birman-Schwinger.  Right: Discrete Dirichlet-to-Neumann eigenvalue via Newton method.}
\label{fig:reconstruction}
\end{figure}

 \subsection{Spectrally embedded defect states in 2D}

In the example below, we set $d=2$ and $\alpha_{\rm f}=5$. The functions $V_\mathrm{f,D}(k,0)$, $V_\mathrm{f,N}$ are plotted in Fig.\,\ref{Vplot_2D}.

\begin{figure}[!hbp] 
\centerline{
\scalebox{0.25}{\includegraphics{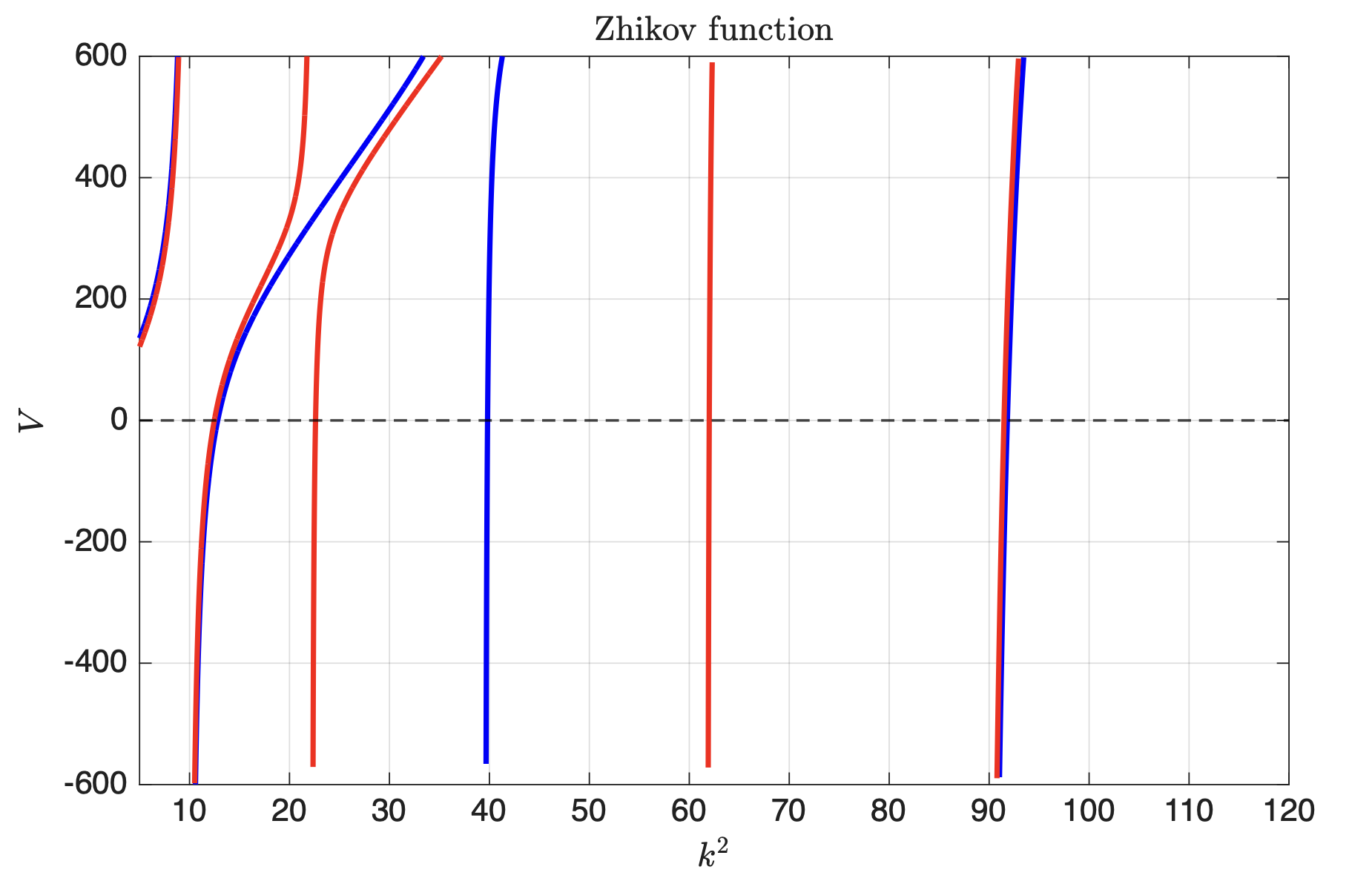}}
}
\caption{\small Blue: Dirichlet, red: Neumann.}
\label{Vplot_2D}
\end{figure}

For $k=6.336$, one has $V_\mathrm{f,D}=-9.4$, $V_\mathrm{f,N}=703.146$, so the value $k^2=39.813$ is in a gap of the problem for the monolayer with Dirichlet decoration and within the spectrum of the problem for the monolayer with Neumann decoration. The first four eigenfunctions of the limit problem for the monolayer with Dirichlet boundary conditions and the defect   
\[
\Psi(x)=\left\{\begin{array}{ll}(1-x^2)^{2.4},& |x|\le1,\\[0.2em]
0,& |x|>1,\end{array}\right.
\]
are shown in Fig.\,\ref{Limit_problem_eigenfunctions_2D}. Note that the second and third eigenfunctions form a basis of the eigenspace for  a double eigenvalue. Of the three corresponding eigenvalues, the first does not exceed $3\alpha_{\rm f}=15$.

\begin{figure}[!hbp] 
\centering 

\scalebox{0.25}{%
\includegraphics{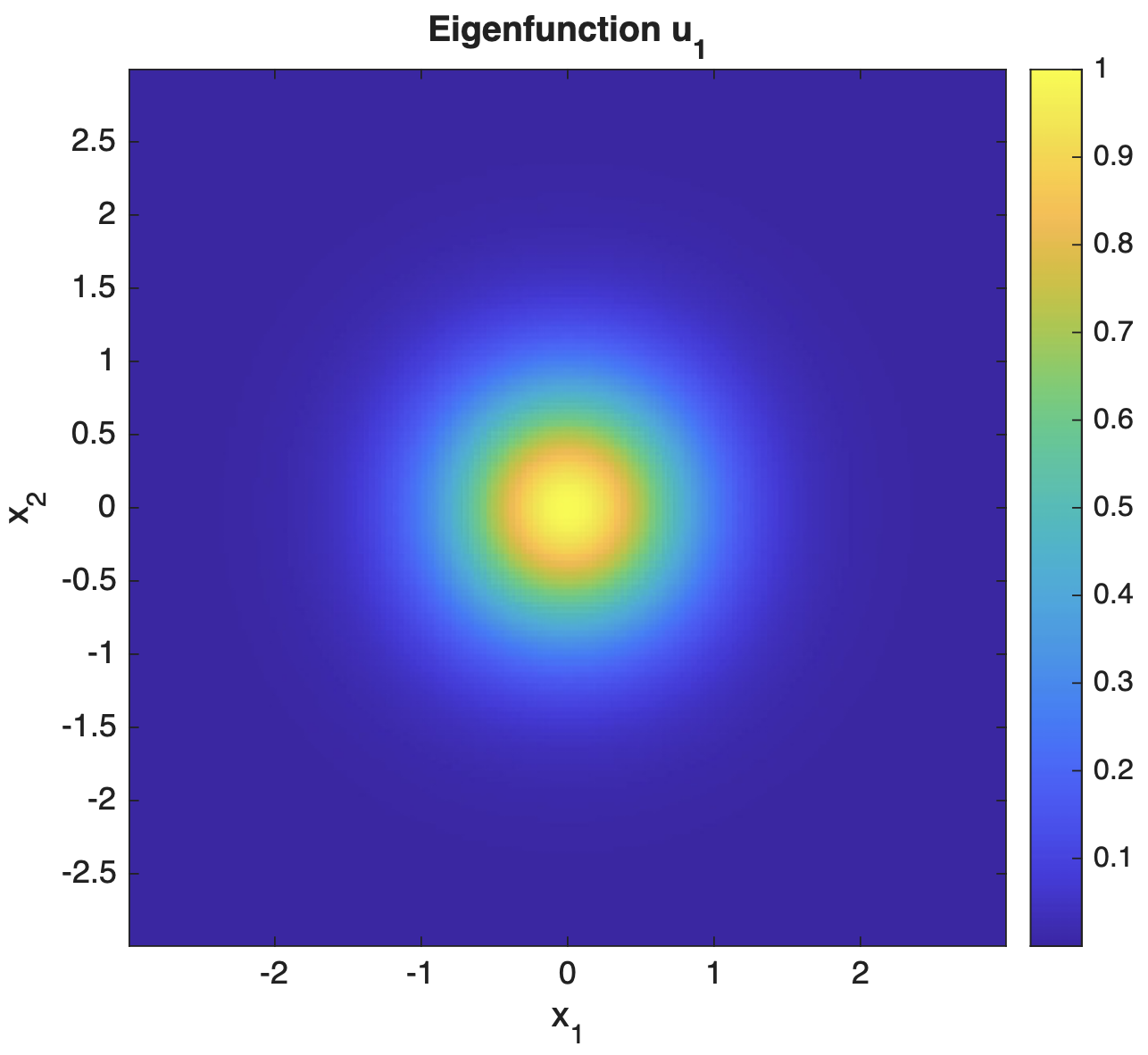}%
} 
\hspace{0.4em} 
\scalebox{0.25}{%
\includegraphics{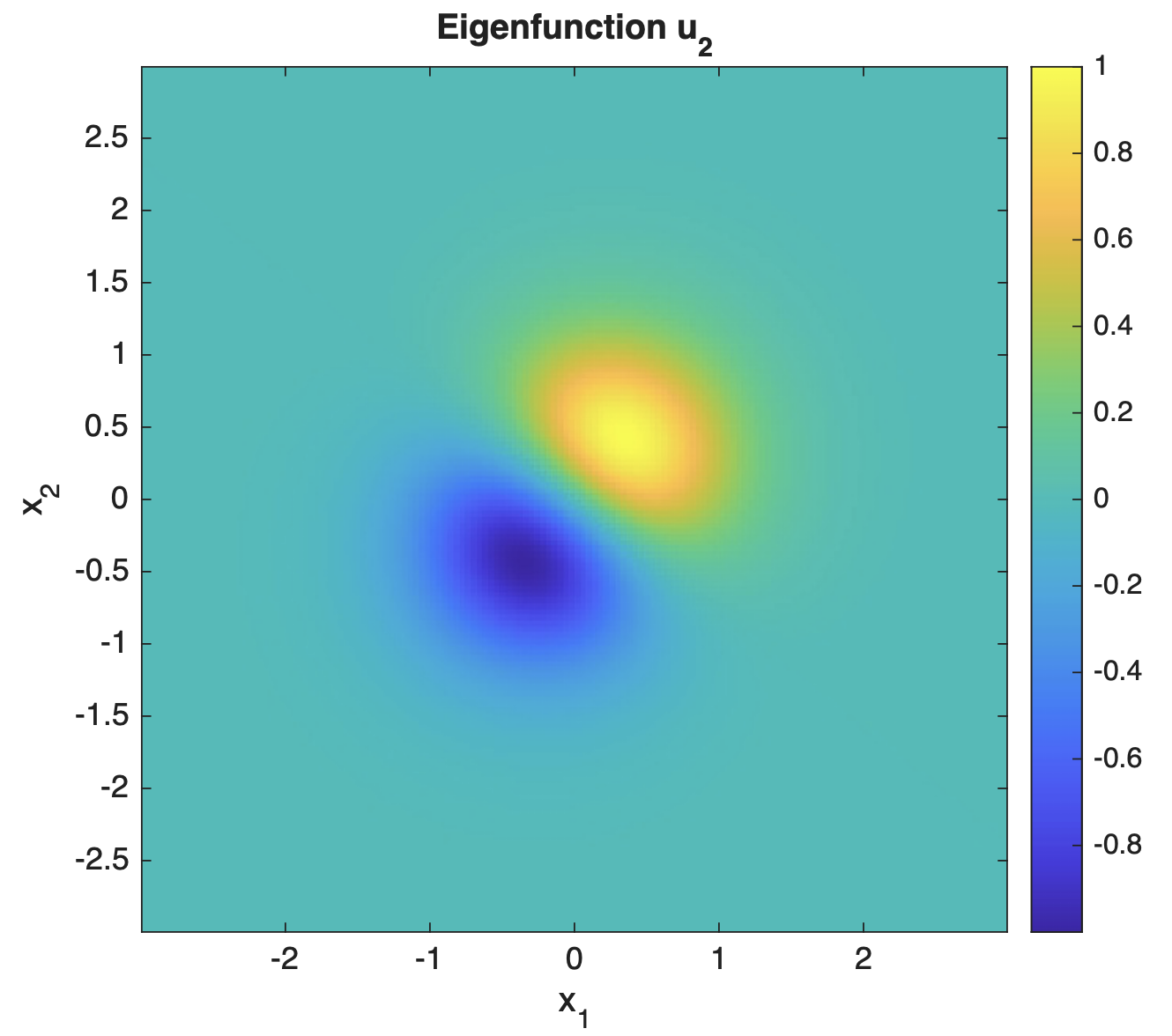}%
} 

\par\medskip 

\scalebox{0.25}{%
\includegraphics{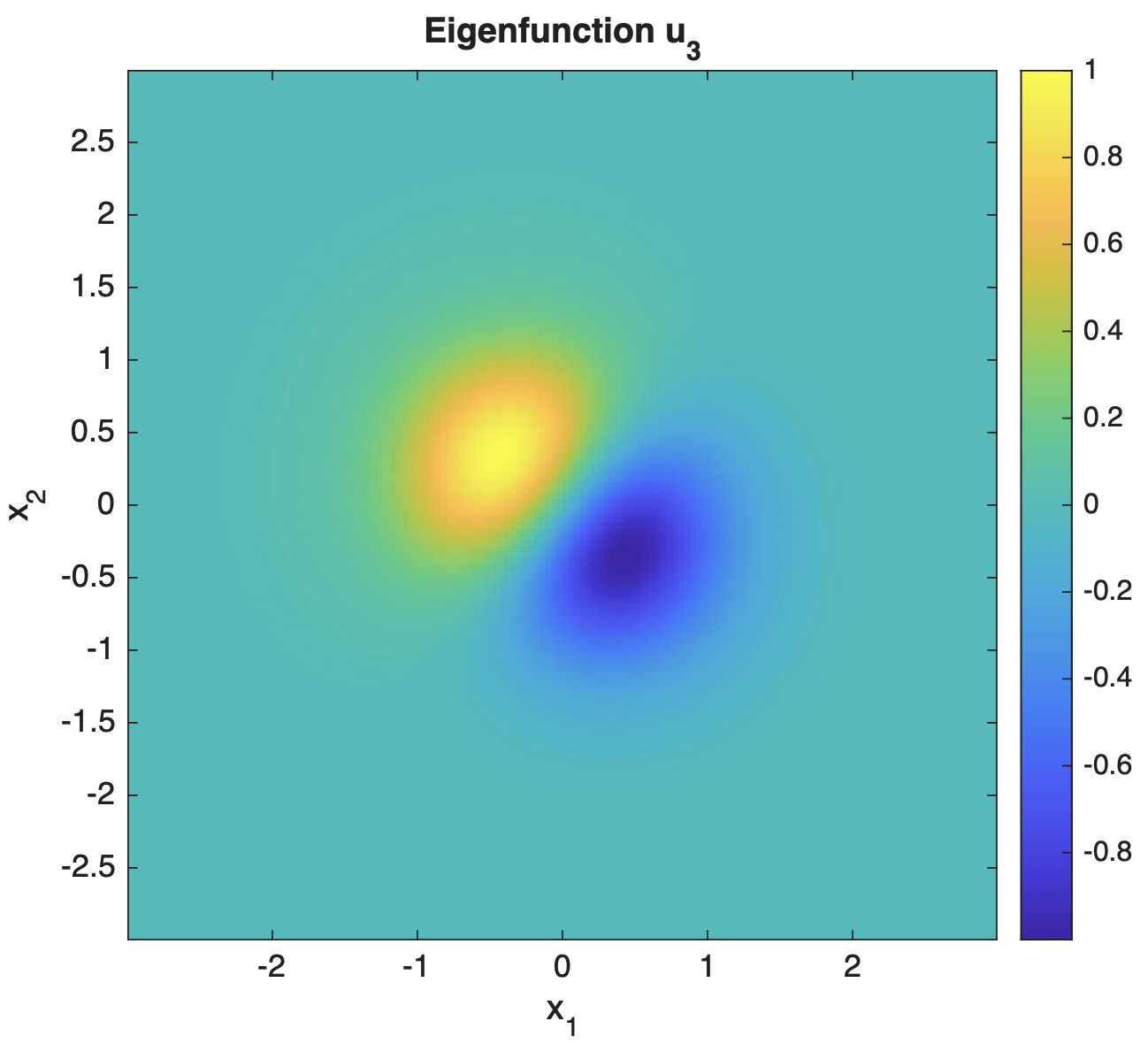}%
} 
\hspace{0.4em}
 \scalebox{0.25}{%
 \includegraphics{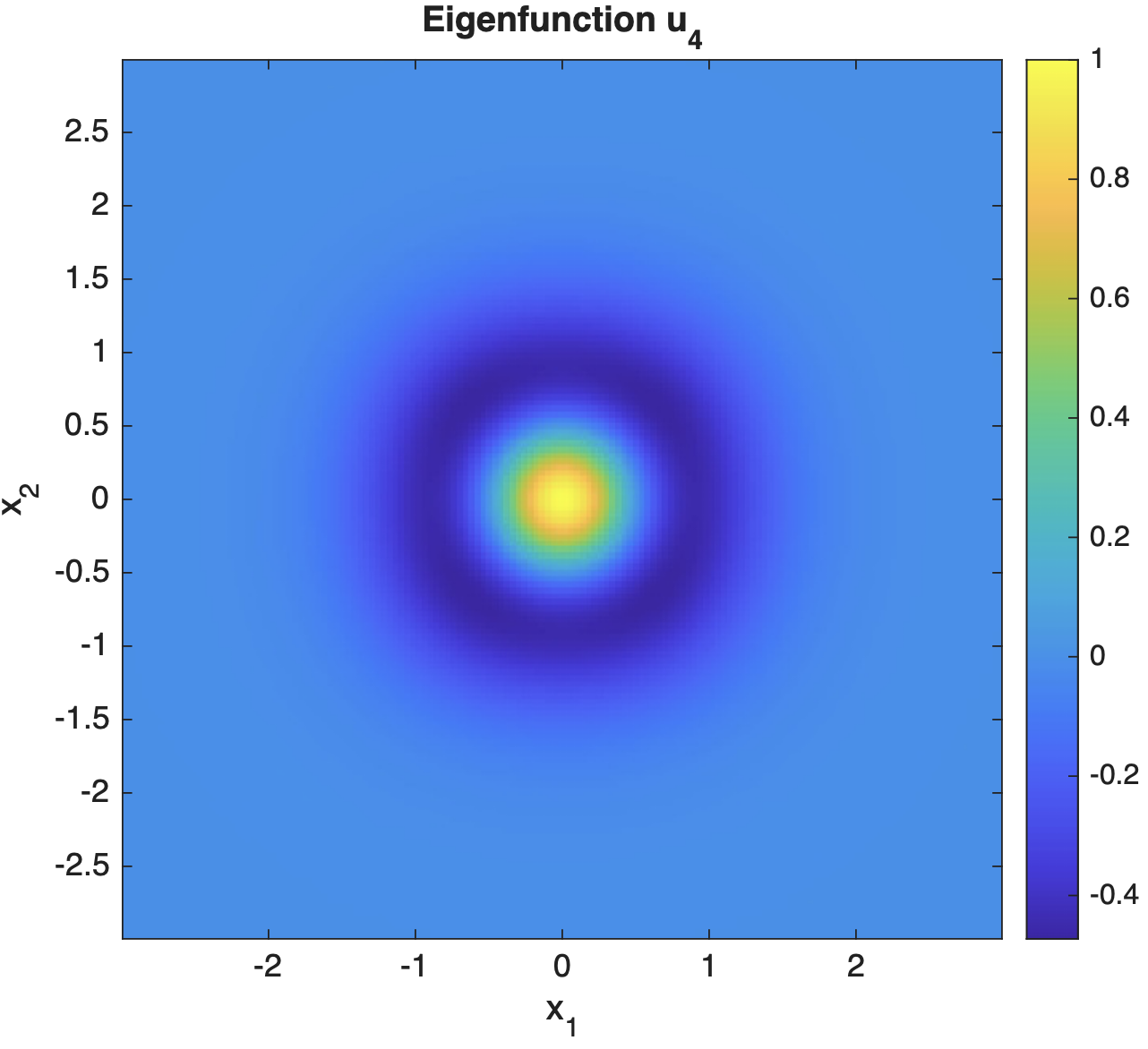}%
} 
\caption{\small The corresponding eigenvalues are 14.934, 22.785 (double), 32.496.} 
\label{Limit_problem_eigenfunctions_2D} 
\end{figure}

The results of direct computation, for $\varepsilon=0.05$ of a defect eigenfunction for the original Dirichlet monolayer with eigenvalue located near 40.151 is provided in Fig.\,\ref{fig:reconstruction_2D}.

 \begin{figure}[!hbp] 
\centerline{
\scalebox{0.22}{\includegraphics{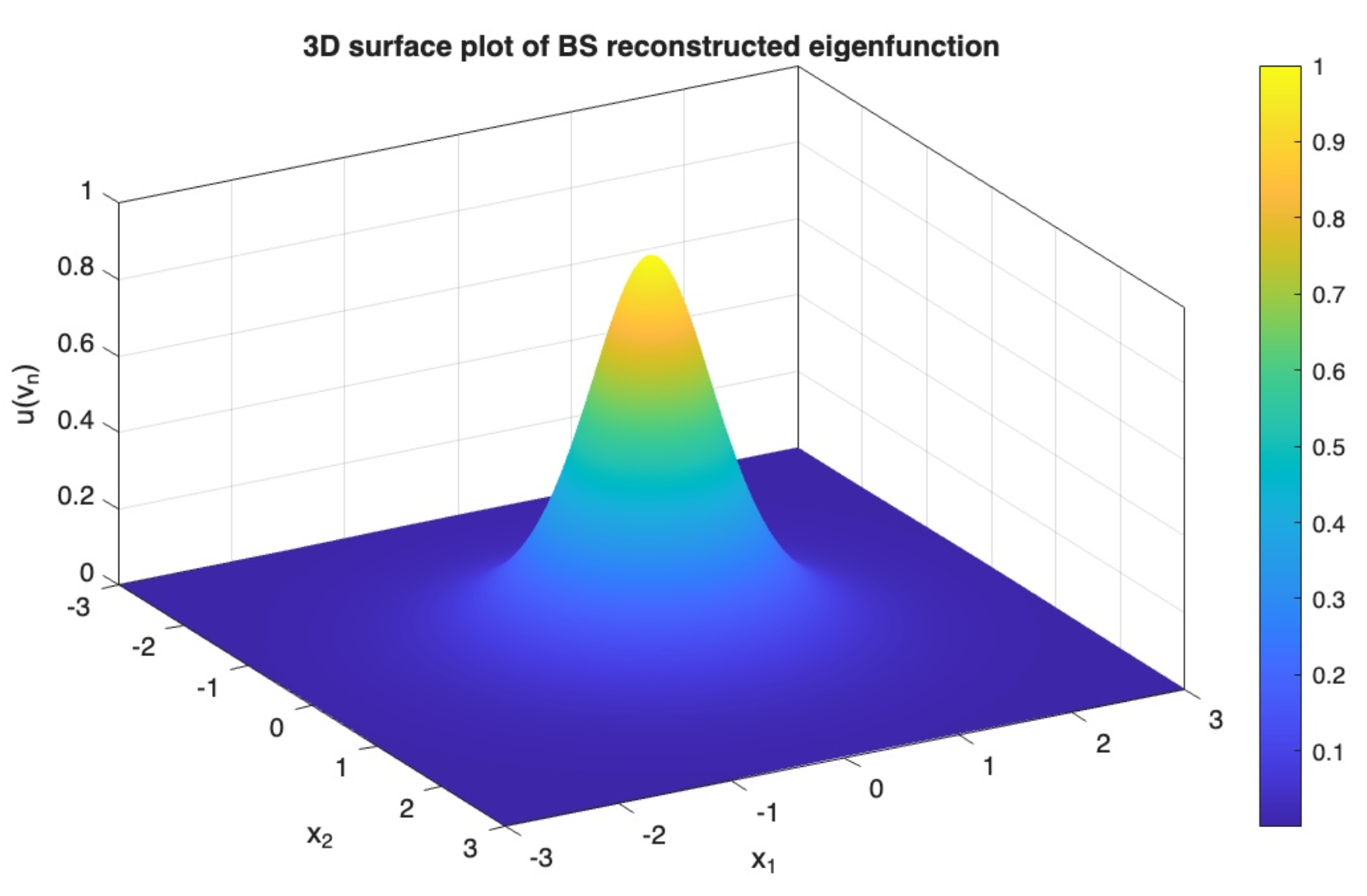}}
\hspace{0.4em}
\scalebox{0.22}{\includegraphics{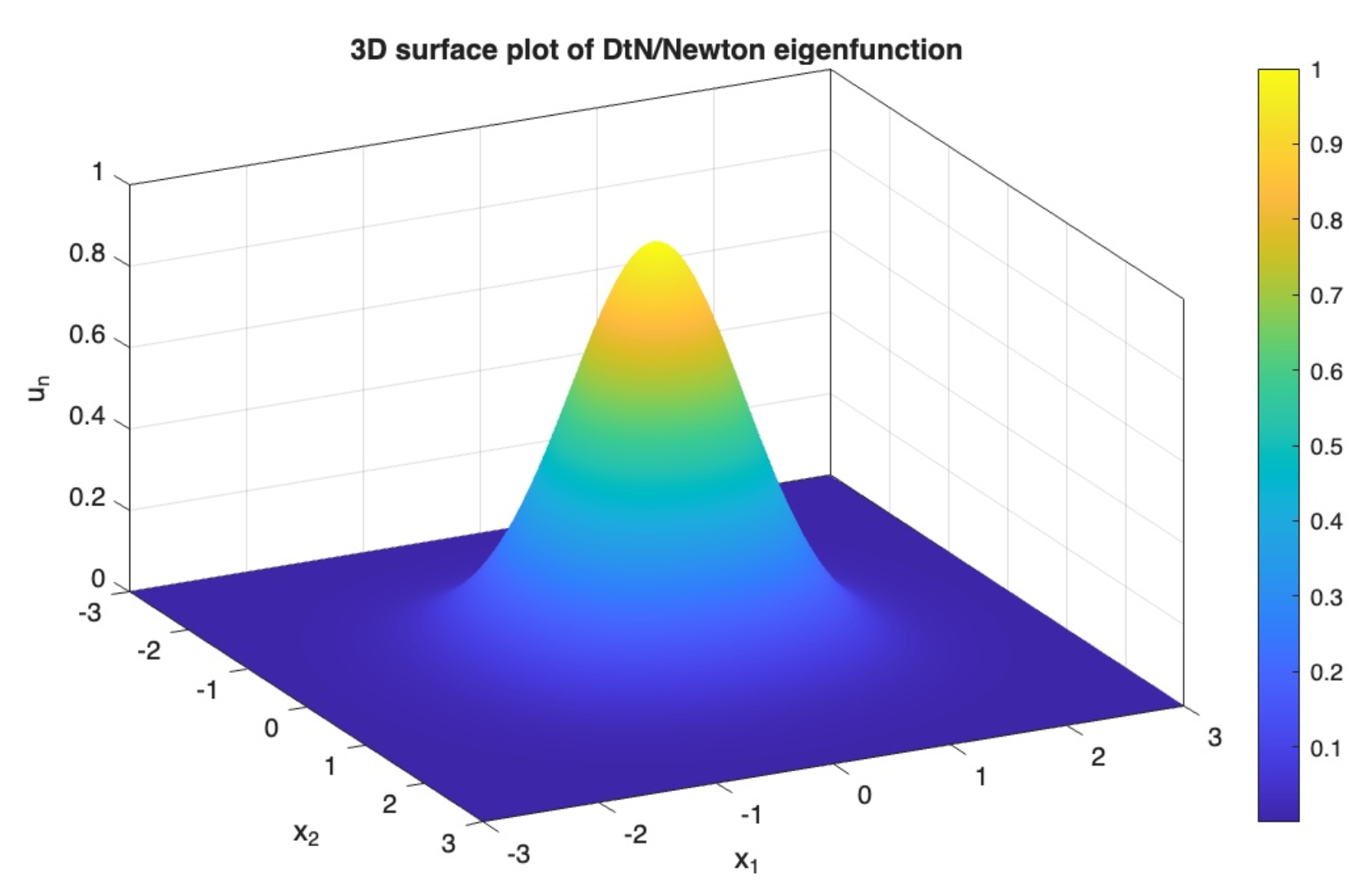}}
}
\caption{\small Left: Original graph via Birman-Schwinger.  Right: Discrete Dirichlet-to-Neumann eigenfunction via Newton method.}
\label{fig:reconstruction_2D}
\end{figure}

\FloatBarrier

\appendix
\section{Calculations}\label{sec:calculations}

\renewcommand{\theequation}{A.\arabic{equation}}
\setcounter{equation}{0}

To compute the DtN matrix $M_\mathrm{soft}$ for an edge $(v,w)$ parameterized by $x\in[0,\eps\ell]$, we compute first the transfer matrix $T_\mathrm{soft}$, which satisfies
\begin{equation*}
  T_\mathrm{soft} \col{1.2}{u(v)}{\Phi[u](v)} \;=\; \col{1.2}{u(w)}{\Phi[u](w)}
\end{equation*}
for any solution $u$ to the differential equation $d((\eps\sigma)^2du/dx)/dx=k^2u$.  Recall that $\Psi[u](v)=(\sigma\eps)^2u'(0)$ and $\Psi[u](w)=-(\sigma\eps)^2u'(\eps\ell)$.  The solution $u_1(x)=\cos(kx/(\eps\sigma))$ has $[u(v),\Phi[u](v)]=[1,0]$, and the solution $u_2(x)=(k\sigma\eps)^{-1}\sin(kx/(\eps\sigma))$ has $[u(v),\Phi[u](v)]=[0,1]$.  Thus,
\begin{equation}
  T_\mathrm{soft} \;=\; \mat{1.3}{u_1(w)}{u_2(w)}{\Phi[u_1](w)}{\Phi[u_2](w)}
  \;=\; \mat{1.3}{\cos k\ell\sigma^{-1}}{(k\sigma\eps)^{-1}\sin k\ell\sigma^{-1}}{k\sigma\eps \sin k\ell\sigma^{-1}}{-\cos k\ell\sigma^{-1}}.
\end{equation}
Some linear algebra yields
\begin{equation}\label{DtNesoft}
  M_\mathrm{soft} \;=\; \frac{\eps\sigma k}{\sin k\ell\sigma^{-1}}\mat{1.3}{-\cos k\ell\sigma^{-1}}{1}{1}{-\cos k\ell\sigma^{-1}}.
\end{equation}
For a soft edge of length $2\ell\eps$, one obtains
\begin{equation}\label{DtNe2}
  \tilde M_\mathrm{soft} \;=\; \frac{\eps\sigma k}{\sin 2k\ell\sigma^{-1}}\mat{1.3}{-\cos 2k\ell\sigma^{-1}}{1}{1}{-\cos 2k\ell\sigma^{-1}}.
\end{equation}

For a stiff edge of length $\eps\ell$, $u$ satisfies $d(\sigma^2du/dx)/dx=k^2u$, and the flux is $\Psi[u](v)=\sigma^2u'(0)$.  One computes
\begin{equation}\label{DtNestiff}
  M_\mathrm{stiff} \;=\; \frac{\sigma k}{\sin \eps k\ell\sigma^{-1}}\mat{1.3}{-\cos \eps k\ell\sigma^{-1}}{1}{1}{-\cos \eps k\ell\sigma^{-1}}.
\end{equation}

To obtain the scalar DtN map for a soft edge with length $\eps\ell$ and Dirichlet condition at one end, we consider solutions $u$ on the interval $[0,2\eps\ell]$ with $u(\eps\ell)\!=\!0$.  These are exactly the odd solutions on the interval $[0,2\ell\eps]$, that is, $u(x)\!=\!u(x-\eps\ell)$.  Such a solution satisfies
\begin{equation}
  \tilde M_\mathrm{soft}\col{1}{u(v)}{-u(v)} \;=\;  \col{1}{\Phi[u](v)}{-\Phi[u](v)}.
\end{equation}
Since $u(v)[1,-1]$ is an eigenvector of $\tilde M_\mathrm{soft}$ with eigenvalue is $-\eps\sigma k\cot k\ell\sigma^{-1}$, one obtains
\begin{equation}
  \Phi[u](v) \;=\; -\eps\sigma k\cot (k\ell\sigma^{-1}) u(v).
\end{equation}
Similarly, to obtain the DtN map for the Neumann condition at $\eps\ell$, consider solutions on $[0,2\eps\ell]$ with $u'(\eps\ell)\!=\!0$.  Such solutions are even, yielding
\begin{equation}
  \tilde M_\mathrm{soft}\col{1}{u(v)}{u(v)} \;=\;  \col{1}{\Phi[u](v)}{\Phi[u](v)}.
\end{equation}
Since $u(v)[1,1]= $ is an eigenvector of $\tilde M_\mathrm{soft}$ with eigenvalue is $\eps\sigma k\tan k\ell\sigma^{-1}$, one obtains
\begin{equation}
  \Phi[u](v) \;=\; \eps\sigma k\tan (k\ell\sigma^{-1}) u(v).
\end{equation}

Now consider the compound edge in Figure~\ref{fig:bdedge} with the lengths of all component edges being $\eps\ell$ and functions $u$ satisfying $d(\sigma^2 du/dx)/dx=k^2u$ on the stiff edges and $d((\eps\sigma)^2 du/dx)/dx=k^2u$ on the soft edge.  Write the matrix $\tilde M_{\bedge}(k,\eps)$ introduced in Section \ref{sec:DtNcompound} in block form
\begin{equation}\label{DtNb2}
\tilde M_{\bedge}(k,\eps) \;=\;
\sigma k \left[\!\!
  \begin{array}{cc}
     A & B \\
    B & D
  \end{array}
\!\right].
\end{equation}
Let us denote $\kappa:=k\ell\sigma^{-1}$.  Using $M_\mathrm{stiff}$ and $M_\mathrm{soft}$ calculate above, we obtain
\begin{equation*}
 A = -(\cot\eps \kappa) I_2,\quad B (\csc\eps \kappa) I_2,\quad
  D \;=\; \mat{2}
  {-2\cot\eps \kappa - \eps\cot \kappa}
  {\csc\eps \kappa + \eps\csc \kappa}
  {\csc\eps \kappa + \eps\csc \kappa}
  {-2\cot\eps \kappa - \eps\cot \kappa},
\end{equation*}
with Schur complement $\Sigma=A-BD^{-1}B$.

The matrix $D$ is diagonalized as follows:
\begin{align*}
  D &\;=\;  \frac{1}{\sin\eps \kappa}\mat{1.2}{-2\cos\eps\kappa}{1}{1}{-2\cos\eps\kappa}
             + \frac{\eps}{\sin\kappa}\mat{1.2}{-\cos\kappa}{1}{1}{-\cos\kappa}\\
   &\;=\;  \frac{1}{2}\mat{1}{1}{-1}{1}{1}
      \left\{ \frac{1}{\sin\eps\kappa}\mat{1.2}{1-2\cos\eps\kappa}{0}{0}{-1-2\cos\eps\kappa}
         + \frac{\eps}{\sin\kappa}\mat{1.2}{1-\cos\kappa}{0}{0}{-1-\cos\kappa} \right\}
      \mat{1}{1}{1}{-1}{1}.
\end{align*}
The inverse eigenvalues are
\begin{align*}
   \rho_\pm &\;=\;  \left[ \frac{\pm1-2\cos\eps\kappa}{\sin\eps\kappa} + \frac{\eps(\pm1-\cos\kappa)}{\sin\kappa} \right]^{-1}
   \;=\; \frac{\sin\kappa\sin\eps\kappa}{\sin\kappa(\pm1-2\cos\eps\kappa) + \eps\sin\eps\kappa(\pm1-\cos\kappa)} \\
   &\;=\; \frac{\sin\eps\kappa}{\pm1-2\cos\eps\kappa + \eps\sin\eps\kappa\csc\kappa(\pm1-\cos\kappa)}.
\end{align*}
Thus we obtain
\begin{equation*}
  D^{-1} \;=\; \frac{1}{2} \mat{1.2}{\rho_++\rho_-}{\rho_+-\rho_-}{\rho_+-\rho_-}{\rho_++\rho_-}
\end{equation*}
and
\begin{align*}
  -\Sigma
  \;=\; (\cot\eps\kappa) I_2 + (\csc\eps\kappa)^2 \frac{1}{2} \mat{1.2}{\rho_++\rho_-}{\rho_+-\rho_-}{\rho_+-\rho_-}{\rho_++\rho_-} 
   \;=\; \frac{1}{\sin\eps\kappa} \left\{ \cos\eps\kappa I_2 + \frac{1}{2} \mat{1.1}{\tau_++\tau_-}{\tau_+-\tau_-}{\tau_+-\tau_-}{\tau_++\tau_-} \right\},
\end{align*}
with
\begin{equation*}
  \tau_\pm \;=\; \big( \pm1-2\cos\eps\kappa + \eps\sin\eps\kappa\csc\kappa (\pm1-\cos\kappa) \big)^{-1}.
\end{equation*}

The $\eps$- and $k$-dependent quantities $\Sigma_{ij}$ are
\begin{align}
  \Sigma_{12} = \Sigma_{21} &\;=\; \frac{-1}{\sin\eps\kappa} \frac{\,\tau_+-\tau_-}{2} \\
   \Sigma_{11} + \Sigma_{22} \;=\; 2\Sigma_{11} & \;=\; \frac{-1}{\sin\eps\kappa} \left( 2\cos\eps\kappa + \tau_++\tau_- \right).
\end{align}
In terms of the quantities
\begin{align}
  b_1 &\;=\; 1+\eps\sin\eps\kappa\csc\kappa \;=\; 1 + \eps^2\kappa\sinc \eps\kappa \csc\kappa  \\
  b_2  &\;=\; 2\cos\eps\kappa + \eps\sin\eps\kappa\cot\kappa \;=\; 2 - \eps^2k \big[\kappa\sinc \textstyle\frac{\eps\kappa}{2} + \sinc\eps\kappa \cot\kappa \big].
\end{align}
we can write
\begin{equation}
  \tau_+-\tau_- \;=\; \frac{-2b_1}{b_2^2-b_1^2}\,,  \qquad
  \tau_++\tau_- \;=\; \frac{-2b_2}{b_2^2-b_1^2}.
\end{equation}

One computes the following asymptotic expansions in $\eps$, noting that the $O(\eps^n)$ terms are not uniform in~$k$.
\begin{align}
  b_2^2-b_1^2 &\;=\; -\left(1+\eps^2\kappa\csc\kappa\right)^2 + \left(2-(\eps\kappa)^2+\eps^2\kappa\cot\kappa\right)^2 + O(\eps^4)  \\
    &\;=\; 3 - \eps^2\left( 2k\csc\kappa + 4(k^2-k\cot\kappa) \right) + O(\eps^4),\\
  b_1^{-1} &\;=\; 1 - \eps^2\kappa\csc\kappa + O(\eps^4),\\
  \left(b_2^2-b_1^2\right)/b_1 &\;=\; 3 - \eps^2\left( 5k\csc\kappa - 4k\cot\kappa + 4k^2 \right) + O(\eps^4), \\
      \frac{b_2^2-b_1^2}{b_1}\cos \eps\kappa &\;=\; 3 - \eps^2\left( 5k\csc\kappa + \textstyle\frac{11}{2}k^2 - 4k\cot\kappa \right) + O(\eps^4), \\
          b_2/b_1 &\;=\; 2 + \eps^2\left( -k^2 +\kappa\cot\kappa - 2k\csc\kappa \right) + O(\eps^4).
\end{align}
Thus,
\begin{equation}
  \frac{\Sigma_{11}}{\Sigma_{12}} \;=\; \frac{b_2}{b_1} - \frac{(b_2^2-b_1^2)\cos \eps\kappa}{b_1}
\end{equation}
\begin{equation}
  -\frac{\eps}{\ell\kappa\,\Sigma_{12}} \;=\; -\eps^2\frac{3}{\ell} + O(\eps^4).
\end{equation}
Putting these together yields
\begin{equation}
    \tilde V(k,\eps;n) \;=\; 2d\,\frac{\Sigma_{11}}{\Sigma_{12}} - \frac{\eps(\alpha_n-m(k))}{\ell\kappa\,\Sigma_{12}}
    \;=\; -2d + \eps^2F(k,n) +  \eps^4\left[ g_1(k,\eps) + g_2(k,\eps)(\alpha_n-m(k))\right],
\end{equation}
in which
\begin{equation}
  F(k,n) \;=\; d\left[ 9k^2 + 6k\left( \csc\kappa - \cot\kappa \right) \right] - \frac{3}{\ell}\left( \alpha_n - m(k) \right).
\end{equation}

\section{A macroscopic model retaining the soft component}
\label{sec:explicit-soft-limit}

The dispersive continuum limit presented in Section~\ref{sec:results} can be realized as the projection of a larger, self-adjoint, system that couples macroscopic and microscopic variables.

\subsection{Monolayer system}

The macroscopic equation obtained by eliminating the soft edges is a frequency-dispersive equation for the field carried by the stiff component. Although this formulation is convenient for the construction of defect states, it suppresses the microscopic degrees of freedom supported by the
soft part of the periodicity cell.  Here we give an equivalent augmented formulation in which the soft-edge fields are retained as independent variables.  The resulting spectral problem is linear in the spectral parameter and is generated by a self-adjoint operator. The frequency-dispersive equation is recovered from the augmented problem by taking a Schur complement with respect to the soft variables.

\subsubsection{The geometry and the natural normalisation}

For each \(n\in\mathbb Z^d\) and \(j\in\{1,\ldots,d\}\), the vertices \(\varepsilon n\) and \(\varepsilon(n+e_j)\) are joined by a compound edge consisting of three stiff edges and one soft edge. All four component edges have metric length \(\varepsilon\ell\). In addition, a soft decoration of length \(\varepsilon\ell\) is attached to each lattice vertex \(\varepsilon n\).
The coefficient of the differential expression is taken to be $\ell^2$ on the stiff edges, $(\varepsilon\ell)^2$ on the soft edges. At the lattice vertex $\varepsilon n$ the Robin parameter is $\alpha_\Psi(\varepsilon n):=\alpha+\mu\Psi(\varepsilon n)$.
The terminal endpoint of every decoration is supplied either with the Dirichlet condition or with the Neumann condition.  We denote the chosen
condition by $B\in\{\mathrm D,\mathrm N\}$.

It is convenient to equip $L^2(\Gamma_\varepsilon)$ with the normalised inner product
\[
(u,v)_{\mathscr H_\varepsilon}:=\frac{3\varepsilon^{d-1}}{\ell}\int_{\Gamma_\varepsilon}u\overline{v}ds.
\]
Thus $\mathscr H_\varepsilon:=L^2\left(\Gamma_\varepsilon; (3\varepsilon^{d-1}/\ell)\,ds\right)$. Multiplying both the energy form and the Hilbert-space inner product by the same positive factor does not change the associated differential operator. This normalisation is chosen so that the limiting coefficient of $|\nabla u|^2$ is equal to one. Indeed, three stiff edges of length $\varepsilon\ell$, connected in series, have effective conductance $\ell/(3\varepsilon)$.  After multiplication by $3\varepsilon^{d-1}/\ell$, the corresponding lattice energy converges to
\[
\int_{\mathbb R^d}|\nabla u|^2\,dx.
\]
We write $H_{\varepsilon,\Psi}^{B}$ for the resulting self-adjoint metric-graph operator in $\mathscr H_\varepsilon$.

\subsubsection{The augmented macroscopic Hilbert space}

Let $I:=(0,1)$. The macroscopic variable is denoted by $x\in\mathbb R^d$, while $y\in I$ is the rescaled coordinate along a soft edge. The augmented Hilbert space is
\[
\mathscr H_{\mathrm{hom}}:=L^2(\mathbb R^d;9d\,dx)\oplus\bigoplus_{j=1}^dL^2(\mathbb R^d\times I;3\,dxdy)\oplus L^2(\mathbb R^d\times I;3\,dxdy).
\]
Its elements are written as $U=(u,q_1,\dots,q_d,r)$. Here, $u=u(x)$ is the macroscopic field associated with the connected stiff component; $q_j=q_j(x,y)$ is the field on the soft edge contained in the compound edge oriented in the $j$-th coordinate direction; $r=r(x,y)$ is the field on the soft decoration. The inner product is
\[
(U,V)_{\mathscr H_{\mathrm{hom}}}=9d\int_{\mathbb R^d}u\overline{v}dx+3\sum_{j=1}^d\int_{\mathbb R^d}\int_0^1q_j\overline{p_j}dydx+3\int_{\mathbb R^d}\int_0^1r\overline{s}dydx,
\]
where $V=(v,p_1,\dots,p_d,s)$. 

\subsubsection{The effective energy form}

Set $\alpha_\Psi(x):=\alpha+\mu\Psi(x)$. For $\mathrm B=\mathrm D$, define
\[
\begin{aligned}
\mathscr V_{\mathrm D}:=\bigl\{
        &(u,q_1,\dots,q_d,r): u\in H^1(\mathbb R^d), q_j,r\in L^2\bigl(\mathbb R^d;H^1(I)\bigr),
\\[0.2em]
        &q_j(\cdot,0)=q_j(\cdot,1)=u,\ \ \ j=1,\dots,d,\quad r(\cdot,0)=u,\quad r(\cdot,1)=0\bigr\}.
\end{aligned}
\]
For $\mathrm B=\mathrm N$, define
\[
\begin{aligned}
    \mathscr V_{\mathrm N}:=\bigl\{
        &(u,q_1,\dots,q_d,r): u\in H^1(\mathbb R^d), q_j,r\in L^2\bigl(\mathbb R^d;H^1(I)\bigr),
        \\[0.2em]
        &q_j(\cdot,0)=q_j(\cdot,1)=u,\ \ j=1,\dots,d,\quad r(\cdot,0)=u\bigr\}.
\end{aligned}
\]
All trace equalities are understood in $L^2(\mathbb R^d)$. On $\mathscr V_{\mathrm B}$ consider the sesquilinear form
\begin{equation}
\mathfrak a_\Psi^{\mathrm B}[U,V]:=\int_{\mathbb R^d}\nabla u\cdot\overline{\nabla v}dx+\frac{3}{\ell}\int_{\mathbb R^d}\alpha_\Psi(x)u\overline vdx+3\sum_{j=1}^d\int_{\mathbb R^d}\int_0^1\partial_yq_j\overline{\partial_yp_j}dydx+3\int_{\mathbb R^d}\int_0^1\partial_yr\overline{\partial_ys}dydx.
\label{eq:augmented-form}
\end{equation}

\begin{proposition}
\label{prop:augmented-operator}
The form $\mathfrak a_\Psi^{\mathrm B}$ is densely defined, closed and bounded from below in $\mathscr H_{\mathrm{hom}}$.  It therefore generates a self-adjoint operator $\mathcal A_\Psi^{\mathrm B}$ in $\mathscr H_{\mathrm{hom}}$.
\end{proposition}

\begin{proof}
The trace maps $L^2\bigl(\mathbb R^d;H^1(I)\bigr)\longrightarrow L^2(\mathbb R^d)$ are bounded.  Consequently, the trace constraints defining
$\mathscr V_{\mathrm B}$ are closed.  The derivative terms in \eqref{eq:augmented-form}, together with the $\mathscr H_{\mathrm{hom}}$-norm, define a complete norm on $\mathscr V_{\mathrm B}$.  Since $\alpha_\Psi$ is real-valued and bounded, the zeroth-order term is a bounded form perturbation. The assertion now follows from Kato's first representation theorem.
\end{proof}

\subsubsection{Strong form of the augmented problem}

For sufficiently regular $U=(u,q_1,\dots,q_d,r)$, the operator $\mathcal A_\Psi^{\mathrm B}$ acts according to
\[
    \mathcal A_\Psi^{\mathrm B}
    \begin{pmatrix}
        u\\ q_1\\ \vdots\\q_d\\r
    \end{pmatrix}
    =
    \begin{pmatrix}
    \displaystyle
    \frac{1}{9d}\biggl\{-\Delta_xu+\frac{3}{\ell}\alpha_\Psi u+3\sum_{j=1}^d\bigl(\partial_yq_j(\cdot,1)-\partial_yq_j(\cdot,0)\bigr)-3\partial_yr(\cdot,0)\biggr\}
    \\[0.3em]
        -\partial_y^2q_1\\[0.4em]
        \vdots\\[0.4em]
        -\partial_y^2q_d\\[0.6em]
        -\partial_y^2r
    \end{pmatrix}.
\]
The coupling conditions are
\[
q_j(x,0)=q_j(x,1)=u(x),\quad j=1,\dots,d,\qquad
r(x,0)=u(x).
\]
At the terminal endpoint of the decoration one has $r(x,1)=0$ if $\mathrm B=\mathrm D$, and $\partial_yr(x,1)=0$ if $\mathrm B=\mathrm N$. The spectral problem
\begin{equation}
\mathcal A_\Psi^{\mathrm B}U=zU
\label{eq:augmented-spectral}
\end{equation}
is therefore equivalent to
\begin{equation}
\begin{aligned}
-\Delta_xu+\frac{3}{\ell}\alpha_\Psi(x)u &+3\sum_{j=1}^d\left[\partial_yq_j(x,1)-\partial_yq_j(x,0)\right]-3\partial_yr(x,0)=9dz\,u, &&x\in\mathbb R^d,
\\[-0.2em]
-\partial_y^2q_j(x,y) &=zq_j(x,y), &&x\in\mathbb R^d,\quad y\in I,
\\
q_j(x,0)=q_j(x,1)&=u(x),&&j=1,\dots,d,
\\[1mm]
-\partial_y^2r(x,y)&=zr(x,y), &&x\in\mathbb R^d,\quad y\in I,
\\
r(x,0)&=u(x),
\end{aligned}
\label{eq:augmented-system}
\end{equation}
supplemented by \eqref{eq:augmented-system} and either $r(x,1)=0$ or $\partial_yr(x,1)=0$. Unlike the scalar frequency-dispersive equation, problem \eqref{eq:augmented-system} is linear in $z$.  Moreover, it continues to make sense at the resonant values of $z$ at which the Dirichlet-to-Neumann functions of the soft edges have poles.

\subsubsection{Recovery of the frequency-dispersive equation}

We now show that the scalar macroscopic equation is the Schur complement of $\mathcal A_\Psi^{\mathrm B}-z$ with respect to the soft variables. Let $z=k^2$, where a branch of the square root is fixed.  We first assume that $z$ is separated from the Dirichlet spectra of the one-dimensional soft problems.

For every $j=1,\dots,d$, the solution of $-\partial_y^2q_j=k^2q_j$, $q_j(x, 0)=q_j(x, 1)=u(x)$, is
\[
q_j(x,y)=u(x)\frac{\cos\bigl(k(y-1/2)\bigr)}{\cos(k/2)}.
\]
Hence
\begin{equation}
\partial_yq_j(x,1)-\partial_yq_j(x,0)=-2k\tan(k/2)\,u(x)=-2k\bigl(\csc k-\cot k\bigr)u(x).
\label{eq:q-flux}
\end{equation}
For the Dirichlet decoration,
\[
r_{\mathrm D}(x,y)=u(x)\frac{\sin(k(1-y))}{\sin k},
\]
and therefore
\begin{equation}
\partial_yr_{\mathrm D}(x,0)=-k\cot k\,u(x)=m_{\mathrm D}(k)\ell^{-1}u(x),\qquad m_{\mathrm D}(k):=-\ell k\cot k.
\label{eq:6.11D}
\end{equation}
For the Neumann decoration,
\[
r_{\mathrm N}(x,y)=u(x)\frac{\cos(k(1-y))}{\cos k},
\]
and therefore
\begin{equation}
\partial_yr_{\mathrm N}(x,0)=k\tan k\,u(x)=m_{\mathrm N}(k)\ell^{-1}u(x),\qquad m_{\mathrm N}(k):=\ell k\tan k.
\label{eq:6.11N}
\end{equation}

Substitution of \eqref{eq:q-flux} and either \eqref{eq:6.11D} or \eqref{eq:6.11N} into the first equation of
\eqref{eq:augmented-system} gives
\[
-\Delta u+3\ell^{-1}\bigl(\alpha+\mu\Psi\bigr)u-6dk\bigl(\csc k-\cot k\bigr)u-3\ell^{-1}m_{\mathrm B}(k)u=9dk^2u.
\]
Equivalently,
\[
\bigl(-\Delta-V_{\mathrm B}(k)\bigr)u=-3\mu\ell^{-1}\Psi u,
\tag{6.14}
\label{eq:compressed-problem}
\]
where
\begin{equation}
V_{\mathrm B}(k)=d\left[9k^2+6k\bigl(\csc k-\cot k\bigr)\right]-3\ell^{-1}\bigl(\alpha-m_{\mathrm B}(k)\bigr).
\label{eq:dispersion-function}
\end{equation}
Thus, \eqref{eq:compressed-problem} is the Schur complement of $\mathcal A_\Psi^{\mathrm B}-k^2$ with respect to the components
$(q_1,\dots,q_d,r)$. At the exceptional values of $k$, the meromorphic expression \eqref{eq:dispersion-function} need not be defined.  The augmented self-adjoint problem \eqref{eq:augmented-spectral}, however, remains well-defined and retains, in particular, the soft-edge modes whose trace on the stiff component vanishes.

\subsection{Bilayer system}
\label{sec:bilayer-augmented-limit}

We now formulate the macroscopic problem corresponding to the bilayer metric graph.  The two layers are assumed to be identical, and the defect
perturbation is imposed symmetrically on the two layers.  Consequently, the involution interchanging the layers commutes with both the periodic and the defective operators.

At the microscopic level, cutting each connecting soft edge at its midpoint reduces the even part to a monolayer problem with a Neumann condition at the terminal endpoint of the resulting half-edge, and the odd part to a monolayer problem with a Dirichlet condition at that endpoint.  Thus, $\dot H_{\varepsilon,\Psi}\simeq H_{\varepsilon,\Psi}^{\mathrm N}\oplus H_{\varepsilon,\Psi}^{\mathrm D}$.
The purpose of this section is to retain this decomposition at the level of the full augmented homogenised operator and then to rewrite the result in
terms of the physical fields carried by the two layers.

Let $I:=(0,1)$. For the two macroscopic layer amplitudes $u_1,u_2$, introduce the even and odd combinations $u_{\mathrm e}:=(u_1+u_2)/\sqrt2$, $u_{\mathrm o}:=(u_1-u_2)/\sqrt2$.
For every $j=1,\dots,d$, let $q_{\nu,j}=q_{\nu,j}(x,y)$, $\nu=1,2$, denote the field on the soft branch contained in the compound edge of the
$\nu$-th layer in the $j$-th coordinate direction.  We similarly set $q_{\mathrm e,j}:=(q_{1,j}+q_{2,j})/\sqrt2$, $q_{\mathrm o,j}:=(q_{1,j}-q_{2,j})/\sqrt2$.
The field on the two halves of a connecting soft edge is denoted by $c_\nu=c_\nu(x,y)$, $y\in I$, $\nu=1,2$, where $y=0$ corresponds to the layer vertex and $y=1$ corresponds to the midpoint of the connecting edge.  Their even and odd combinations are $r_{\mathrm e}:=(c_1+c_2)/\sqrt2$, $r_{\mathrm o}:=(c_1-c_2)/\sqrt2$, respectively.
The matching conditions at the midpoint of the full connecting edge are equivalent to $\partial_y r_{\mathrm e}(x,1)=0$, $r_{\mathrm o}(x,1)=0$.

\subsubsection{The effective operator in parity variables}

Let $\mathcal A_\Psi^{\mathrm N}$, $\mathcal A_\Psi^{\mathrm D}$ be the augmented monolayer operators defined in Appendix~\ref{sec:explicit-soft-limit}, with Neumann and Dirichlet terminal conditions, respectively.  The augmented bilayer operator in parity
variables is $\mathcal A_{\Psi,\mathrm{par}}^{\mathrm{bi}}:=\mathcal A_\Psi^{\mathrm N}\oplus\mathcal A_\Psi^{\mathrm D}$.
ts Hilbert space is $\mathscr H_{\mathrm{par}}^{\mathrm{bi}}:=\mathscr H_{\mathrm{hom}}^{\mathrm N}\oplus\mathscr H_{\mathrm{hom}}^{\mathrm D}$,
and its elements are written as
\[
U_{\mathrm{par}}=\bigl(u_{\mathrm e}, q_{\mathrm e,1},\dots,q_{\mathrm e,d}, r_{\mathrm e}; u_{\mathrm o}, q_{\mathrm o,1},\dots,q_{\mathrm o,d}, r_{\mathrm o}\bigr).
\]
The spectral problem $\mathcal A_{\Psi,\mathrm{par}}^{\mathrm{bi}}U_{\mathrm{par}}=zU_{\mathrm{par}}$
is the direct sum of the two systems
\[    
-\Delta u_\sigma+\frac{3}{\ell}\alpha_\Psi(x)u_\sigma+3\sum_{j=1}^d\left[\partial_yq_{\sigma,j}(x,1)-\partial_yq_{\sigma,j}(x,0)\right]-3\partial_y r_\sigma(x,0)=9dzu_\sigma,\qquad
\sigma\in\{\mathrm e,\mathrm o\},
\]
together with
\begin{equation*}
\begin{aligned}
-\partial_y^2q_{\sigma,j}&=zq_{\sigma,j},\qquad q_{\sigma,j}(x,0)=q_{\sigma,j}(x,1)=u_\sigma(x),\\[0.3em]
-\partial_y^2r_\sigma&=zr_\sigma,\qquad r_\sigma(x,0)=u_\sigma(x).
\end{aligned}
\end{equation*}
The terminal conditions are $\partial_y r_{\mathrm e}(x,1)=0$, $r_{\mathrm o}(x,1)=0$.
Consequently, $\mathcal A_{\Psi,\mathrm{par}}^{\mathrm{bi}}=\mathcal A_\Psi^{\mathrm N}\oplus\mathcal A_\Psi^{\mathrm D}$ is self-adjoint and bounded below.

\subsubsection{Formulation in the physical layer variables}

Let
\begin{equation*}
\mathscr H_{\mathrm{bi}}:=\bigoplus_{\nu=1}^2\biggl[L^2(\mathbb R^d;9ddx)\oplus\bigoplus_{j=1}^dL^2(\mathbb R^d\times I;3dxdy)\oplus L^2(\mathbb R^d\times I;3dxdy)\biggr].
\end{equation*}
An element of this space is denoted by $U=\bigl(u_1,q_{1,1},\dots,q_{1,d},c_1;u_2,q_{2,1},\dots,q_{2,d},c_2\bigr)$. Furthermore, let $\mathcal S:\mathscr H_{\mathrm{bi}}\longrightarrow\mathscr H_{\mathrm{par}}^{\mathrm{bi}}$ be the unitary parity transform defined by
\[
 \mathcal S\begin{pmatrix}
        X_1\\X_2
    \end{pmatrix}
    =
    \frac1{\sqrt2}
    \begin{pmatrix}
        X_1+X_2\\
        X_1-X_2
    \end{pmatrix}
\]
simultaneously on the macroscopic amplitudes, the internal soft-edge fields and the two halves of the connecting edge. The physical bilayer operator is
\begin{equation}
    \mathcal A_\Psi^{\mathrm{bi}}
    :=
    \mathcal S^*
    \bigl(
        \mathcal A_\Psi^{\mathrm N}
        \oplus
        \mathcal A_\Psi^{\mathrm D}
    \bigr)
    \mathcal S.
\label{eq:physical-bilayer-operator}
\end{equation}
Its form domain consists of all $U=\bigl(u_1,\{q_{1,j}\}_{j=1}^d,c_1; u_2,\{q_{2,j}\}_{j=1}^d,c_2\bigr)$ such that $u_\nu\in H^1(\mathbb R^d)$, $q_{\nu,j},c_\nu\in L^2\bigl(\mathbb R^d;H^1(I)\bigr)$, and
\begin{equation}
q_{\nu,j}(\cdot,0)=q_{\nu,j}(\cdot,1)=u_\nu,\qquad
c_\nu(\cdot,0)=u_\nu,\qquad 
c_1(\cdot,1)=c_2(\cdot,1).
\label{eq:physical-midpoint-continuity}
\end{equation}
The effective quadratic form is
\begin{equation}
\mathfrak a_\Psi^{\mathrm{bi}}[U]:=\sum_{\nu=1}^2\biggl\{\int_{\mathbb R^d}|\nabla u_\nu|^2dx+\frac{3}{\ell}\int_{\mathbb R^d}\alpha_\Psi(x)|u_\nu|^2dx+3\sum_{j=1}^d\int_{\mathbb R^d}\int_0^1|\partial_yq_{\nu,j}|^2dydx+3\int_{\mathbb R^d}\int_0^1|\partial_yc_\nu|^2dydx\biggr\}.
\label{eq:bilayer-form}
\end{equation}
The flux-balance condition at the midpoint of the connecting edge, 
\begin{equation}
\partial_yc_1(x,1)+\partial_yc_2(x,1)=0,
\label{eq:midpoint-flux-balance}
\end{equation}
is the natural boundary condition generated by \eqref{eq:bilayer-form}.  The third equation in \eqref{eq:physical-midpoint-continuity} and equation \eqref{eq:midpoint-flux-balance} are precisely the continuity and Kirchhoff conditions at the midpoint of the full connecting edge.

\begin{proposition}
\label{prop:bilayer-effective-operator}
The form $\mathfrak a_\Psi^{\mathrm{bi}}$ is densely defined, closed and bounded below in $\mathscr H_{\mathrm{bi}}$.  The self-adjoint operator generated by this form is the operator $\mathcal A_\Psi^{\mathrm{bi}}$ defined by \eqref{eq:physical-bilayer-operator}.
\end{proposition}

\begin{proof}
The parity transform $\mathcal S$ maps the form domain of $\mathfrak a_\Psi^{\mathrm{bi}}$ unitarily onto $\mathscr V_{\mathrm N}\oplus\mathscr V_{\mathrm D}$. Indeed, the midpoint conditions give $r_{\mathrm o}(\cdot,1)=\bigl(c_1(\cdot,1)-c_2(\cdot,1)\bigr)/\sqrt2=0$,
while the natural flux condition gives $\partial_y r_{\mathrm e}(\cdot,1)=\bigl(\partial_yc_1(\cdot,1)+\partial_yc_2(\cdot,1)\bigr)/\sqrt2=0$. Moreover, $\mathfrak a_\Psi^{\mathrm{bi}}[U]=\mathfrak a_\Psi^{\mathrm N}[U_{\mathrm e}]+\mathfrak a_\Psi^{\mathrm D}[U_{\mathrm o}]$. The assertion therefore follows from the corresponding result for the two monolayer forms.
\end{proof}

\subsubsection{Strong form in the two physical layers and elimination of the soft variables}\label{sec:fulllimitsystem}

The spectral problem $\mathcal A_\Psi^{\mathrm{bi}}U=zU$
is equivalent to
\[
-\Delta u_\nu+\frac{3}{\ell}\alpha_\Psi(x)u_\nu+3\sum_{j=1}^d\bigl[\partial_yq_{\nu,j}(x,1)-\partial_yq_{\nu,j}(x,0)\bigr]-3\partial_yc_\nu(x,0)=9dzu_\nu,
    \qquad \nu=1,2,
\]
where
\begin{equation*}
\begin{aligned}
&-\partial_y^2q_{\nu,j}=zq_{\nu,j},\qquad q_{\nu,j}(x,0)=q_{\nu,j}(x,1)=u_\nu(x),\\[0.3em]
&-\partial_y^2c_\nu=zc_\nu,\qquad\quad c_\nu(x,0)=u_\nu(x).
\end{aligned}
\end{equation*}
At $y=1$, one imposes the conditions $c_1(x,1)=c_2(x,1)$,
$\partial_yc_1(x,1)+\partial_yc_2(x,1)=0$.
Thus, the two macroscopic layer amplitudes are coupled only through the explicit field on the connecting soft edge.

Let $z=k^2$ and assume temporarily that $k$ is separated from the relevant one-dimensional resonances. The internal soft branches in either layer satisfy
\[
\partial_yq_{\nu,j}(x,1)-\partial_yq_{\nu,j}(x,0)=-2k(\csc k-\cot k)u_\nu(x).
\]
It is simpler to eliminate the connecting-edge field in parity variables. One has
\[
\partial_y r_{\mathrm e}(x,0)=k\tan k\,u_{\mathrm e}(x)=m_{\mathrm N}(k)\ell^{-1}u_{\mathrm e}(x),\qquad
\partial_y r_{\mathrm o}(x,0)=-k\cot k\,u_{\mathrm o}(x)=m_{\mathrm D}(k)\ell^{-1}u_{\mathrm o}(x),
\]
where $m_{\mathrm N}(k)=\ell k\tan k$, $m_{\mathrm D}(k)=-\ell k\cot k$.
Consequently, the even and odd amplitudes satisfy
\[
\bigl(-\Delta-V_{\mathrm N}(k)\bigr)u_{\mathrm e}=-3\mu\ell^{-1}\Psi u_{\mathrm e},\qquad
\bigl(-\Delta-V_{\mathrm D}(k)\bigr)u_{\mathrm o}=-3\mu\ell^{-1}\Psi u_{\mathrm o},
\]
where
\[
\begin{aligned}
V_{\mathrm B}(k)=d\left[9k^2+6k(\csc k-\cot k)\right]-3\ell^{-1}\bigl(\alpha-m_{\mathrm B}(k)\bigr),\qquad \mathrm B\in\{\mathrm N,\mathrm D\}.
\end{aligned}
\]

Returning to $u_1,u_2$, define $V_{\mathrm{av}}(k):=\bigl(V_{\mathrm N}(k)+V_{\mathrm D}(k)\bigr)/2$,
$V_{\mathrm{c}}(k):=\bigl(V_{\mathrm N}(k)-V_{\mathrm D}(k)\bigr)/2$.
The compressed bilayer equation is
\begin{equation}
\left[-\Delta-
\begin{pmatrix}
V_{\mathrm{av}}(k) & V_{\mathrm{c}}(k)\\
V_{\mathrm{c}}(k) & V_{\mathrm{av}}(k)
\end{pmatrix}
\right]
\begin{pmatrix}u_1\\u_2\end{pmatrix}
=-\frac{3\mu}{\ell}\Psi
\begin{pmatrix}u_1\\u_2\end{pmatrix}.
\label{eq:compressed-bilayer}
\end{equation}
The diagonal term is
\[
V_{\mathrm{av}}(k)=d\left[9k^2+6k(\csc k-\cot k)\right]-\frac{3\alpha}{\ell}+\frac{3}{2\ell}\bigl(m_{\mathrm N}(k)+m_{\mathrm D}(k)\bigr),
\]
whereas the interlayer coupling is
\[
V_{\mathrm c}(k)=\frac{3}{2\ell}\bigl(m_{\mathrm N}(k)-m_{\mathrm D}(k)\bigr)=\frac{3k}{2}\bigl(\tan k+\cot k\bigr)=\frac{3k}{\sin(2k)}.
\]
Equation \eqref{eq:compressed-bilayer} is the Schur complement of the augmented bilayer operator.

\end{document}